\pdfoutput=1

\documentclass[11pt]{article}

\usepackage{amsthm}
\usepackage{graphicx} 
\usepackage{array} 

\usepackage{amsmath, amssymb, amsfonts, verbatim}
\usepackage{hyphenat,epsfig,subcaption,multirow}
\usepackage{nicefrac}
\usepackage{paralist}

\usepackage[usenames,dvipsnames]{xcolor}
\usepackage[ruled]{algorithm2e}

\DeclareFontFamily{U}{mathx}{\hyphenchar\font45}
\DeclareFontShape{U}{mathx}{m}{n}{
      <5> <6> <7> <8> <9> <10>
      <10.95> <12> <14.4> <17.28> <20.74> <24.88>
      mathx10
      }{}
\DeclareSymbolFont{mathx}{U}{mathx}{m}{n}
\DeclareMathSymbol{\bigtimes}{1}{mathx}{"91}

\usepackage{tcolorbox}
\tcbuselibrary{skins,breakable}
\tcbset{enhanced jigsaw}

\usepackage[normalem]{ulem}
\usepackage[compact]{titlesec}

\definecolor{DarkRed}{rgb}{0.5,0.1,0.1}
\definecolor{DarkBlue}{rgb}{0.1,0.1,0.5}

\usepackage{nameref}
\definecolor{ForestGreen}{rgb}{0.1333,0.5451,0.1333}
\definecolor{Red}{rgb}{0.9,0,0}
\usepackage[linktocpage=true,
	pagebackref=true,colorlinks,
	linkcolor=DarkRed,citecolor=ForestGreen,
	bookmarks,bookmarksopen,bookmarksnumbered]
	{hyperref}
\usepackage[noabbrev,nameinlink]{cleveref}
\crefname{property}{property}{Property}
\creflabelformat{property}{(#1)#2#3}
\crefname{equation}{eq}{Eq}
\creflabelformat{equation}{(#1)#2#3}

\usepackage{bm}
\usepackage{url}
\usepackage{xspace}
\usepackage[mathscr]{euscript}

\usepackage{tikz}
\usetikzlibrary{arrows}
\usetikzlibrary{arrows.meta}
\usetikzlibrary{shapes}
\usetikzlibrary{backgrounds}
\usetikzlibrary{positioning}
\usetikzlibrary{decorations.markings}
\usetikzlibrary{patterns}
\usetikzlibrary{calc}
\usetikzlibrary{fit}
\usetikzlibrary{decorations}

\usepackage[framemethod=TikZ]{mdframed}

\usepackage[noend]{algpseudocode}
\makeatletter
\def\BState{\State\hskip-\ALG@thistlm}
\makeatother

\usepackage{cite}
\usepackage{enumitem}

\usepackage[margin=1in]{geometry}

\newtheorem{theorem}{Theorem}
\newtheorem{lemma}{Lemma}[section]
\newtheorem{proposition}[lemma]{Proposition}

\newtheorem{claim}[lemma]{Claim}
\newtheorem{fact}[lemma]{Fact}
\newtheorem{assumption}[lemma]{Assumption}

\newtheorem{problem}{Problem}

\newtheorem*{claim*}{Claim}
\newtheorem*{assumption*}{Assumption}
\newtheorem*{proposition*}{Proposition}
\newtheorem*{lemma*}{Lemma}
\newtheorem*{problem5*}{Problem}

\newtheorem{observation}[lemma]{Observation}

\newtheorem*{theorem*}{Theorem}

\crefname{lemma}{Lemma}{Lemmas}
\crefname{claim}{claim}{claims}

\newtheorem{mdresult}{Result}
\newenvironment{result}{\begin{mdframed}[backgroundcolor=lightgray!40,topline=false,rightline=false,leftline=false,bottomline=false,innertopmargin=2pt]\begin{mdresult}}{\end{mdresult}\end{mdframed}}

\theoremstyle{definition}

\newtheorem{mdproblem}{Problem}

\newtheorem*{mdproblem*}{Problem}
\newenvironment{Problem*}{\begin{mdframed}[hidealllines=false,innerleftmargin=10pt,backgroundcolor=gray!10,innertopmargin=5pt,innerbottommargin=5pt,roundcorner=10pt]\begin{mdproblem*}}{\end{mdproblem*}\end{mdframed}}
\newtheorem{mddefinition}[lemma]{Definition}
\newenvironment{Definition}{\begin{mdframed}[hidealllines=false,innerleftmargin=10pt,backgroundcolor=white!10,innertopmargin=5pt,innerbottommargin=5pt,roundcorner=10pt]\begin{mddefinition}}{\end{mddefinition}\end{mdframed}}
\newtheorem*{mddefinition*}{Definition}
\newenvironment{Definition*}{\begin{mdframed}[hidealllines=false,innerleftmargin=10pt,backgroundcolor=white!10,innertopmargin=5pt,innerbottommargin=5pt,roundcorner=10pt]\begin{mddefinition*}}{\end{mddefinition*}\end{mdframed}}
\newtheorem{mdremark}{Remark}
\newenvironment{Remark}{\begin{mdframed}[backgroundcolor=lightgray!40,topline=false,rightline=false,leftline=false,bottomline=false,innertopmargin=2pt]\begin{mdremark}}{\end{mdremark}\end{mdframed}}

\newtheoremstyle{restate}{}{}{\itshape}{}{\bfseries}{~(restated).}{.5em}{\thmnote{#3}}
\theoremstyle{restate}

\allowdisplaybreaks

\renewcommand{\qed}{\nobreak \ifvmode \relax \else
      \ifdim\lastskip<1.5em \hskip-\lastskip
      \hskip1.5em plus0em minus0.5em \fi \nobreak
      \vrule height0.75em width0.5em depth0.25em\fi}

\newcommand{\logstar}[1]{\ensuremath{\log^{*}\!{#1}}}

\renewcommand{\leq}{\leqslant}
\renewcommand{\geq}{\geqslant}

\renewcommand{\ge}{\geq}

\newcommand{\Leq}[1]{\ensuremath{\underset{\textnormal{#1}}\leq}}

\newcommand{\tvd}[2]{\ensuremath{\norm{#1 - #2}_{\mathrm{tvd}}}}
\newcommand{\Ot}{\ensuremath{\widetilde{O}}}
\newcommand{\eps}{\ensuremath{\varepsilon}}
\newcommand{\Paren}[1]{\Big(#1\Big)}
\newcommand{\Bracket}[1]{\Big[#1\Big]}
\newcommand{\bracket}[1]{\left[#1\right]}
\newcommand{\paren}[1]{\ensuremath{\left(#1\right)}\xspace}
\newcommand{\card}[1]{\left\vert{#1}\right\vert}

\newcommand{\IR}{\ensuremath{\mathbb{R}}}

\newcommand{\cU}{\ensuremath{\mathcal{U}}}

\newcommand{\norm}[1]{\ensuremath{\|#1\|}}

\newcommand{\ceil}[1]{{\left\lceil{#1}\right\rceil}}
\newcommand{\floor}[1]{{\left\lfloor{#1}\right\rfloor}}

\newcommand{\set}[1]{\ensuremath{\left\{ #1 \right\}}}
\newcommand{\poly}{\mbox{\rm poly}}
\newcommand{\polylog}{\mbox{\rm  polylog}}

\newcommand{\alg}{\ensuremath{\mathcal{A}}\xspace}

\DeclareMathOperator*{\Exp}{\ensuremath{{\mathbb{E}}}}
\DeclareMathOperator*{\Prob}{\ensuremath{\textnormal{Pr}}}
\renewcommand{\Pr}{\Prob}

\newcommand{\Ex}{\Exp}

\newenvironment{ourbox}{\begin{mdframed}[hidealllines=false,innerleftmargin=10pt,backgroundcolor=white!10,innertopmargin=5pt,innerbottommargin=5pt,roundcorner=10pt]}{\end{mdframed}}

\newenvironment{tbox}{\begin{tcolorbox}[
		enlarge top by=5pt,
		enlarge bottom by=5pt,
		breakable,
		boxsep=0pt,
		left=4pt,
		right=4pt,
		top=10pt,
		arc=0pt,
		boxrule=1pt,toprule=1pt,
		colback=white
		]
	}
	{\end{tcolorbox}}

\newcommand{\event}{\ensuremath{\mathcal{E}}}

\newcommand{\rv}[1]{\ensuremath{{\mathsf{#1}}}\xspace}
\newcommand{\rA}{\rv{A}}
\newcommand{\rB}{\rv{B}}
\newcommand{\rC}{\rv{C}}
\newcommand{\rD}{\rv{D}}

\newcommand{\rX}{\rv{X}}
\newcommand{\rY}{\rv{Y}}
\newcommand{\rZ}{\rv{Z}}

\newcommand{\cD}{\mathcal{D}}

\newcommand{\supp}[1]{\ensuremath{\textnormal{\text{supp}}(#1)}}
\newcommand{\distribution}[1]{\ensuremath{\textnormal{dist}(#1)}\xspace}

\newcommand{\kl}[2]{\ensuremath{\mathbb{D}(#1~||~#2)}}
\newcommand{\II}{\ensuremath{\mathbb{I}}}
\newcommand{\HH}{\ensuremath{\mathbb{H}}}
\newcommand{\mi}[2]{\ensuremath{\def\mione{#1}\def\mitwo{#2}\mireal}}
\newcommand{\mireal}[1][]{
	\ifx\relax#1\relax%
	\II(\mione \,; \mitwo)%
	\else%
	\II(\mione \,; \mitwo\mid #1)%
	\fi
}
\newcommand{\en}[1]{\ensuremath{\HH(#1)}}

\newcommand{\itfacts}[1]{\Cref{fact:it-facts}-(\ref{part:#1})\xspace}

\newcommand{\rs}{Ruzsa-Szemer\'edi\xspace} 
\newcommand{\Grs}{\ensuremath{G_{\textnormal{rs}}}}

\newcommand{\Ers}{\ensuremath{E_{\textnormal{rs}}}}
\newcommand{\Lrs}{\ensuremath{L_{\textnormal{rs}}}}
\newcommand{\Rrs}{\ensuremath{R_{\textnormal{rs}}}}
\newcommand{\Mrs}[1]{\ensuremath{M^\textnormal{rs}_{#1}}}

\newcommand{\nrs}{\ensuremath{n^{\textnormal{rs}}}}

\newcommand{\rightRS}[1]{\ensuremath{\textnormal{\textsc{right}}(#1)}}
\newcommand{\leftRS}[1]{\ensuremath{\textnormal{\textsc{left}}(#1)}}

\newcommand{\daks}{\ensuremath{d_{\textnormal{\textsc{sort}}}}}
\newcommand{\caks}{\ensuremath{c_{\textnormal{\textsc{sort}}}}}

\newcommand{\conc}{\ensuremath{\circ}}

\newcommand{\first}{\ensuremath{\textnormal{\textsc{First}}}}
\newcommand{\last}{\ensuremath{\textnormal{\textsc{Last}}}}

\newcommand{\mem}[2]{\ensuremath{\textnormal{\textsc{mem}}_{#2}^{#1}}}

\newcommand{\mph}{\ensuremath{\textnormal{\textsf{MPH}}}\xspace}

\newcommand{\QP}[1]{\ensuremath{Q^{(#1)}}}
\newcommand{\sigmaP}[1]{\ensuremath{\sigma^{(#1)}}}
\newcommand{\SigmaP}[1]{\ensuremath{\Sigma^{(#1)}}}

\newcommand{\prot}{\pi}
\newcommand{\Prot}{\Pi}

\newcommand{\ProtP}[1]{\Pi^{(#1)}}

\newcommand{\rProt}{\rv{\Prot}}

\newcommand{\info}{\ensuremath{\textnormal{Inf}}\xspace}

\newcommand{\good}{\ensuremath{\textnormal{Good}}\xspace}

\newcommand{\reps}{\ensuremath{\textnormal{\textsc{RepBasis}}}}
\newcommand{\hrho}[1]{\ensuremath{\widehat{\rho(#1)}}}
\newcommand{\hrhot}[1]{\ensuremath{\widehat{\rho_0(#1)}}}

\newcommand{\distsigma}[1]{\ensuremath{\cD_{\sigma}}}

\newcommand{\protFake}{\ensuremath{\Pi_{\texnormal{fake}}}}

\newcommand{\jstar}{\ensuremath{j^{*}}}

\newcommand{\GammaY}{\ensuremath{\Gamma_{\textnormal{Yes}}}}
\newcommand{\GammaN}{\ensuremath{\Gamma_{\textnormal{No}}}}
\newcommand{\Gammastar}{\ensuremath{\Gamma^*}}

\newcommand{\cP}{\ensuremath{\mathcal{P}}}
\newcommand{\cM}{\ensuremath{\mathcal{M}}}

\newcommand{\IndexHPH}{\ensuremath{\textnormal{\textbf{Multi-HPH}}}\xspace}

\newcommand{\dist}{\mu}

\newcommand{\rGammastar}{\rv{\Gammastar}}

\newcommand{\rSigma}{\rv{\Sigma}}

\newcommand{\rGamma}{\rv{\Gamma}}
\newcommand{\rL}{\rv{L}}

\newcommand{\rgammastar}{\bm{\gamma}^{\star}}

\renewcommand{\cM}{M}

\newcommand{\rcM}{{\rv{\cM}}}

\newcommand{\bSigma}{\ensuremath{\bar{\Sigma}}}
\newcommand{\bsigma}{\ensuremath{\bar{\sigma}}}
\newcommand{\bSigmaP}[1]{\ensuremath{\bSigma^{(#1)}}}
\newcommand{\bsigmaP}[1]{\ensuremath{\bsigma^{(#1)}}}
\newcommand{\bProt}{\ensuremath{\bar{\Prot}}}

\newcommand{\rbProt}{\rv{\bProt}}
\newcommand{\rbsigmaP}[1]{\ensuremath{\bar{\bm{\sigma}}^{(#1)}}}
\newcommand{\rbSigmaP}[1]{\ensuremath{\rv{\bSigma}^{(#1)}}}
\newcommand{\rbSigma}{\ensuremath{\rv{\bSigma}}}

\newcommand{\rell}{\bm{\ell}}
\newcommand{\rvv}{\rv{v}}

\renewcommand{\info}{\ensuremath{\textnormal{\textsc{Info}}}\xspace}

\newcommand{\re}{\rv{e}}

\newcommand{\sigmaId}{\ensuremath{\sigma_{\textnormal{\textsc{id}}}}}

\newcommand{\ww}{\ensuremath{w}}
\newcommand{\dd}{\ensuremath{d}}
\newcommand{\bb}{\ensuremath{b}}
\newcommand{\LL}{\ensuremath{\mathcal{L}}}

\newcommand{\extFull}[1]{\ensuremath{\textnormal{\textsc{Match-Perm}}(#1)}}
\newcommand{\edgepick}{\ensuremath{\textnormal{\textsc{Edge-Pick}}}}
\newcommand{\Estar}{\ensuremath{E^*}}
\newcommand{\sigmaLeft}{\ensuremath{\sigma_{\textnormal{\textsc{L}}}}}
\newcommand{\sigmaRight}{\ensuremath{\sigma_{\textnormal{\textsc{R}}}}}
\newcommand{\Mleft}{\ensuremath{M_{\textnormal{\textsc{L}}}}}
\newcommand{\Mright}{\ensuremath{M_{\textnormal{\textsc{R}}}}}

\newcommand{\basic}[1]{\ensuremath{\textnormal{\textsc{Basic}}(#1)}}
\newcommand{\permgroups}[1]{\ensuremath{\textnormal{\textsc{Permute}}(#1)}}
\newcommand{\encodedRS}{\ensuremath{\textnormal{\textsc{Encoded-RS}}}}
\newcommand{\block}{\ensuremath{\textnormal{\textsc{Block}}}}
\newcommand{\mulblock}{\ensuremath{\textnormal{\textsc{Multi-Block}}}}

\newcommand{\GG}{\mathcal{G}}
\newcommand{\gammastar}{\ensuremath{\gamma^*}}
\newcommand{\Ssim}{\ensuremath{S^{\textnormal{\textsc{sim}}}}}
\newcommand{\GGsim}{\ensuremath{\GG^{\textnormal{\textsc{sim}}}}}
\newcommand{\tovec}{\ensuremath{\textnormal{\textsc{vec}}}}
\newcommand{\joinit}{\ensuremath{\textnormal{\textsc{join}}}}

\newcommand{\lexpart}{\ensuremath{\textnormal{\textsf{Lex}}}}

\newcommand{\swap}{\ensuremath{\textnormal{\textsf{swap}}}}

\newcommand{\memn}{\ensuremath{\textnormal{\textsc{mem}}}}

\newcommand{\PP}{\ensuremath{\cD}}

\newcommand{\extend}{\ensuremath{\textnormal{\textsc{Extend}}}}

\title{Hidden Permutations to the Rescue: Multi-Pass Semi-Streaming Lower Bounds for Approximate Matchings}
\author{Sepehr Assadi\footnote{(sepehr@assadi.info) Cheriton School of Computer Science, University of Waterloo, and Department of Computer Science, Rutgers University. 
Supported in part by an Alfred P. Sloan Fellowship, a University of Waterloo startup grant, an NSF CAREER grant CCF-2047061, and a gift from Google Research. \smallskip}
\and Janani Sundaresan\footnote{(jsundaresan@uwaterloo.ca) Cheriton School of Computer Science, University of Waterloo.}} 

\begin{document}
\maketitle

\pagenumbering{roman}

\begin{abstract}

\bigskip

We prove that any semi-streaming algorithm for $(1-\eps)$-approximation of maximum bipartite matching requires 
\[
\Omega\Paren{\dfrac{\log{(1/\eps)}}{{\log{(1/\beta)}}}}
\]
passes, 
where $\beta \in (0,1)$ is the largest parameter so that an $n$-vertex graph with $n^{\beta}$ edge-disjoint induced matchings of size $\Theta(n)$ exist (such graphs are referred to as \emph{\rs} graphs).  
Currently, it is known that 
\[
\Omega(\frac{1}{\log\log{n}}) \leq \beta \leq 1-\Theta(\frac{\logstar{n}}{{\log{n}}})
\]
and closing this huge gap between upper and lower bounds has remained a notoriously difficult problem in combinatorics.

\medskip

Under the plausible hypothesis that $\beta = \Omega(1)$, our lower bound result provides the {first} \textbf{pass-approximation} lower bound for (small) \textbf{constant approximation} of matchings in the semi-streaming model, a longstanding open question in the graph streaming literature. 

\medskip

Our techniques are based on analyzing communication protocols for compressing (hidden) {permutations}. 
Prior work in this context relied on reducing such problems to Boolean domain and analyzing them via tools like XOR Lemmas and Fourier analysis on Boolean hypercube. In contrast, our main technical contribution
is a hardness amplification result for permutations through concatenation in place of prior XOR Lemmas. This result is proven by analyzing permutations directly 
via simple tools from group representation theory combined with detailed information-theoretic arguments, and can be of independent interest.

\end{abstract}

\clearpage

\setcounter{tocdepth}{3}
\tableofcontents
\clearpage
\pagenumbering{arabic}
\setcounter{page}{1}

\renewcommand{\gammastar}{\gamma^{\star}}

\newcommand{\betars}{\ensuremath{\beta_{\textsc{rs}}}}
\section{Introduction}

In the semi-streaming model for graph computation, formalized by~\cite{FeigenbaumKMSZ05}, the edges of an $n$-vertex graph $G=(V,E)$ are presented to the algorithm in some
arbitrarily ordered stream. The algorithm can make one or few passes over this stream and uses $\Ot(n) := O(n \cdot \polylog{(n)})$ memory to solve the given problem. 
The semi-streaming model has been at the forefront of research on processing massive graphs since its introduction almost two decades ago. 
In this work, we focus on the \emph{maximum matching} problem
in this model.

The maximum matching problem is arguably the most studied problem in the semi-streaming model and 
has been considered from numerous angles (this list is by no means a comprehensive summary of prior results): 
\begin{itemize}
	\item single-pass algorithms~\cite{FeigenbaumKMSZ05,GoelKK12,Kapralov13,Kapralov21,AssadiBKL23},
	\item constant-pass algorithms~\cite{KonradMM12,EsfandiariHM16,KaleT17,Konrad18,KonradN21,Assadi22,FeldmanS22,KonradNS23}, 
	\item $(1-\eps)$-approximation algorithms~\cite{McGregor05,AhnG11,EggertKMS12,AhnG18,Tirodkar18,GamlathKMS19,AssadiLT21,FischerMU22,AssadiJJST22,Assadi23},
	\item random-order streams~\cite{KonradMM12,Konrad18,AssadiBBMS19,GamlathKMS19,FarhadiHMRR20,Bernstein20,AssadiB21a,AssadiS23},
	\item dynamic streams~\cite{Konrad15,ChitnisCHM15,AssadiKLY16,ChitnisCEHMMV16, AssadiKL17,DarkK20,AssadiS22},
	\item weighted or submodular matchings~\cite{FeigenbaumKMSZ05,CrouchS14,ChakrabartiK14,ChekuriGQ15,PazS17,BernsteinDL21,LevinW21},
	\item matching size estimation~\cite{KapralovKS14,EsfandiariHLMO15,BuryS15,McGregorV16,CormodeJMM17,McGregorV18,AssadiKL17,KapralovMNT20,AssadiKSY20,AssadiN21,AssadiS23},  
	\item and, exact algorithms~\cite{FeigenbaumKMSZ05,GuruswamiO13,AssadiR20,LiuSZ20,ChenKPSSY21,AssadiJJST22}.
\end{itemize}

In this paper, we focus on proving \textbf{multi-pass lower bounds for $(1-\eps)$-approximation} of the maximum matching problem via semi-streaming algorithms, primarily for the regime 
of \textbf{small constant} $\eps > 0$ independent of size of the graph. 

The question of understanding the approximation ratio achievable by multi-pass semi-streaming algorithms for matchings was posed by~\cite{FeigenbaumKMSZ05} alongside the introduction of the semi-streaming model itself. Moreover,~\cite{FeigenbaumKMSZ05} also gave a $(2/3-\eps)$-approximation algorithm for this problem in $O(1/\eps)$ passes, which was soon after improved by~\cite{McGregor05} to a $(1-\eps)$-approximation in $(1/\eps)^{O(1/\eps)}$ passes. 
A long line of work since then~\cite{AhnG11,KonradMM12,Kapralov13,EggertKMS12,KaleT17,AhnG18,Konrad18,Tirodkar18,AssadiLT21,FischerMU22,AssadiJJST22} has culminated in semi-streaming algorithms 
with $\poly(1/\eps)$ passes for general graphs~\cite{FischerMU22} and $O(1/\eps^2)$ passes for bipartite graphs~\cite{AssadiLT21} (there are also algorithms with pass-complexity with better dependence on $\eps$ at the cost of 
mild dependence on $n$,
namely, $O(\log{n}/\eps)$ passes in~\cite{AhnG18,AssadiJJST22,Assadi23} or for finding perfect matchings in $n^{3/4+o(1)}$ passes~\cite{AssadiJJST22}; see also~\cite{LiuSZ20}). 

The lower bound front however has seen much less progress with only a handful of results known 
for single-pass algorithms~\cite{GoelKK12,Kapralov13,AssadiKL17,Kapralov21} and very recently two-pass algorithms~\cite{KonradN21,Assadi22}. But no lower bounds
beyond two-pass algorithms are known for constant factor approximation of matchings in the semi-streaming model (lower bounds for computing exact maximum matchings up to almost $\Omega(\log{n})$ passes 
are proven in~\cite{GuruswamiO13}; see also~\cite{AssadiR20,ChenKPSSY21}, but these lower bounds cannot apply to $\eps > n^{-o(1)}$, and we shall discuss them later in more details). It is worth noting that in the much more restricted setting of $\polylog{(n)}$-space
algorithms,~\cite{AssadiN21}, building on~\cite{AssadiKSY20}, proved an $\Omega(1/\eps)$-pass lower bound for estimating the matching size. 

\subsection{Our Contribution}\label{sec:results}

We present a new lower bound for multi-pass semi-streaming algorithms for the maximum matching problem. The lower bound is parameterized 
by the density of \emph{\rs} (RS) graphs~\cite{RuzsaS78}, namely, graphs whose edges can be partitioned into induced matchings of size $\Theta(n)$ (see~\Cref{sec:rs}). Let $\betars \in (0,1)$ 
denote the largest parameter such that there exist $n$-vertex RS graphs with $\Omega(n^{\betars})$ edge-disjoint induced matchings of size $\Theta(n)$. We prove the following result in this paper. 

\smallskip

\begin{result}[Formalized in~\Cref{thm:main-lb}]\label{res:main-lb}
	Any (possibly randomized) semi-streaming algorithm 
	for $(1-\eps)$-approximation of even the size of maximum matchings requires 
	\[
	\Omega\Paren{\dfrac{\log{(1/\eps)}}{\log{(1/\betars)}}}
	\]
	passes over the stream. The lower bound holds for the entire range of $\eps \in [n^{-\Theta(\betars)},\Theta(\betars)]$. 
\end{result}
To put this result in more context, we shall note that currently, it is only known that
\[
\Omega\!\paren{\frac{1}{\log\log{n}}} \Leq{\cite{FischerLNRRS02}} \betars \Leq{\cite{FoxHS15}} 1-\Theta\!\paren{\frac{\logstar{n}}{{\log{n}}}},
\]
and closing this gap appears to be a  challenging question in combinatorics (see, e.g.~\cite{Gowers01,FoxHS15,ConlonlF13}). 
Thus, \Cref{res:main-lb} can be interpreted as an $\Omega(\log{(1/\eps)})$ lower bound on pass-complexity of $(1-\eps)$-approximation of matchings 
in the semi-streaming model in two different ways: 
\begin{enumerate}[label=$(\roman*)$]
	\item A \emph{conditional} lower bound, under the plausible hypothesis that $\betars=\Omega(1)$. It is  known how to construct RS graphs with induced matchings of size $n^{1-o(1)}$
	that   have ${{n}\choose{2}} - o(n^2)$ edges~\cite{AlonMS12}, but the regime of $\Theta(n)$-size induced matchings is  wide open.
	\item A \emph{barrier result}; obtaining such algorithms requires 
	reducing $\betars$ from $1-o(1)$ to $o(1)$ which seems beyond the reach of current techniques.
	\end{enumerate}

Let us now compare this result with some prior work. 

A line of work closely related to ours is lower bounds for constant-approximation of matchings in one pass~\cite{GoelKK12,Kapralov13,AssadiKL17,Kapralov21} or two passes~\cite{Assadi22}. 
Specifically,~\cite{Kapralov21} rules out single-pass semi-streaming algorithms for finding $(0.59)$-approximate matchings (see also~\cite{GoelKK12,Kapralov13}). And,~\cite{AssadiKL17} and~\cite{Assadi22} rule out semi-streaming algorithms for approximating \emph{size} of maximum matchings to within a $(1-\eps_0)$ factor for some $\eps_0 > 0$, in one and two passes\footnote{See also~\cite{KonradN21} that give a two-pass lower bound for a restricted family of algorithms that only compute a greedy matching in their first pass but then can be arbitrary in their second pass.}, respectively 
(these two lower bounds, similar to ours, rely on the hypothesis that $\betars=\Omega(1)$).

Another line of closely related work are lower bounds for computing perfect or nearly-perfect matchings in multiple passes~\cite{GuruswamiO13,AssadiR20,ChenKPSSY21}, 
which culminated in the $\Omega(\sqrt{\log{n}})$ pass lower bound of \cite{ChenKPSSY21} even for algorithms with $n^{2-o(1)}$ space (an almost $\Omega(\log{n})$ pass lower bound for semi-streaming algorithms 
was already known by~\cite{GuruswamiO13}). The lower bound of~\cite{ChenKPSSY21} can be further interpreted 
for $(1-\eps)$-approximate matching algorithms as follows: 
\begin{equation}
	\begin{aligned}\label{eq:chenkpssy21}
		&\text{$\Omega(\dfrac{\log{(1/\eps)}}{\sqrt{\log{n}}})$ pass lower bound when $\eps \leq 2^{-\Theta(\sqrt{\log{n}})}$ for $n^{2-o(1)}$-space algorithms}; \\
		&\text{$\Omega(\dfrac{\log{(1/\eps)}}{{\log\log{n}}})$ pass lower bound when $\eps \leq (\log{n})^{-\Theta(1)}$ for $\Ot(n)$-space algorithms}. 
	\end{aligned}
\end{equation}
Yet, these lower bounds, even under the strongest assumption of $\betars = 1-o(1)$ do not imply any non-trivial bounds for constant-factor approximation algorithms. 

Before moving on from this section, we mention some important remarks about our result. 

\paragraph{Constant-factor approximation.} 
Our~\Cref{res:main-lb} is  the \emph{first} lower bound on pass-approximation tradeoffs for semi-streaming matching algorithms
that applies to constant-factor approximations. The fact that the approximation can be a constant is  critical here as we elaborate on below. 

Firstly, in contrast to possibly some other models, 
in the semi-streaming model, the most interesting regime for $(1-\eps)$-approximation is for constant $\eps >0$ (see, e.g.~\cite{FeigenbaumKMSZ05,McGregor05,Tirodkar18,GamlathKMS19,FischerMU22}). 
One key reason, among others,  is that the cost of each additional pass over the stream is non-trivially high and thus algorithms that need super-constant number of passes over the stream (a consequence of sub-constant $\eps$) are prohibitively costly already. 

Secondly, many streaming matching algorithms have rather cavalier space-dependence on $\eps$ (as space is typically much less costly compared to passes), even exponential-in-$\eps$, e.g., in~\cite{McGregor05,GamlathKMS19,BernsteinDL21} (although see~\cite{AssadiLT21,AssadiJJST22} for some exceptions with no space-dependence on $\eps$ at all). Yet, the lower bounds of the type obtained by~\cite{ChenKPSSY21} that require $\eps$ to be at most $(\log{n})^{-\Theta(1)}$ (even assuming $\betars=\Omega(1)$)  
cannot provide any meaningful guarantees for these algorithms, as the space of such algorithms for such small $\eps$ already become more than size of the input.\footnote{To give a concrete example, the state-of-the-art lower bounds before 
	our paper left open the possibility of a $(1-\eps)$-approximation algorithm in $3$ passes and $(1/\eps)^{O(1/\eps)} \cdot \Ot(n)$ space. Obtaining such algorithms would have been a huge breakthrough and quite interesting. 
	Our~\Cref{res:main-lb} however now rules out such an algorithm (conditionally) even in any $o(\log{(1/\eps)})$ passes (or alternatively, identify a challenging barrier toward obtaining such algorithms).}
This however is not an issue for our lower bounds for constant-approximation algorithms. 

\paragraph{Role of RS graphs.}  Starting from the work of~\cite{GoelKK12}, all previous single- and multi-pass lower
bounds for semi-streaming matching problem in~\cite{GoelKK12,Kapralov13,AssadiKL17,AssadiR20,Kapralov21,ChenKPSSY21,KonradN21,Assadi22} are based on RS graphs---the only exception is the lower 
bound of~\cite{GuruswamiO13} for finding perfect matchings (which is improved upon in~\cite{AssadiR20,ChenKPSSY21} using RS graphs).  

Our~\Cref{res:main-lb} is also based on RS graphs and relies on the hypothesis that $\betars=\Omega(1)$ in order
to be applicable to constant-factor approximation algorithms. While not all prior lower bounds rely on this hypothesis, assuming it also is not uncommon (see, e.g.~\cite{AssadiKL17,Assadi22}). Indeed, currently, 
a $(1-\eps)$-approximation lower bound for estimating size of maximum matching that does not rely on this hypothesis is not known even for single-pass algorithms. Similarly, the space lower bound in~\Cref{res:main-lb} is
in fact $n^{1+\Omega(1)}$; again, such bounds are not known even for finding edges of a $(1-\eps)$-approximate matching in a single pass without relying on the $\betars=\Omega(1)$ hypothesis. 
This in fact may not be a coincidence: a very recent work of~\cite{AssadiBKL23} has provided evidence that at least qualitatively, relying on such hypotheses might be necessary. They use RS graph bounds {algorithmically} instead and show
that if $\betars=o(1)$, then one can find a $(1-\eps)$-approximate matching already in a single pass in much better than quadratic space\footnote{Quantitatively however, there is 
	still a large gap between upper bounds of~\cite{AssadiBKL23} even if $\betars=o(1)$, and our bounds or those of~\cite{AssadiKL17,Assadi22} even if $\beta=1-o(1)$. Yet, this still suggests
	that the complexity of matching problem in graph streams is very closely tied to the density of RS graphs from both upper and lower bound fronts.}. 

All in all, while we find the problem of proving (even single-pass) streaming matching lower bounds without relying on RS graphs, or even better yet, improving bounds on the density of RS graphs, 
quite fascinating open questions, we believe those questions are orthogonal to our research direction on multi-pass lower bounds. 

Finally, we mention the current lower bound on $\betars$ due to~\cite{FischerLNRRS02,GoelKK12}
combined with our~\Cref{res:main-lb} leads the following \textbf{unconditional result} for $(1-\eps)$-approximation of matching size:
\begin{align}
	&\text{$\Omega(\frac{\log{(1/\eps)}}{\log\log\log{n}})$ pass lower bound when $\eps \leq (\log\log{n})^{-\Theta(1)}$ for $\Ot(n)$-space algorithms},
	\label{eq:our-non-bounds}
\end{align}
which exponentially improves the range of $\eps$ (and the denominator) compared to~\cite{ChenKPSSY21} in~\Cref{eq:chenkpssy21}. 

\paragraph{Beyond $(\log{(1/\eps)})$ passes.} The pass lower bound in~\Cref{res:main-lb} (for $\betars=\Omega(1)$) appears to hit the same standard barrier of 
proving super-logarithmic lower bounds for most graph streaming problems including reachability, shortest path, and perfect matching~\cite{GuruswamiO13,AssadiR20,ChakrabartiGMV20,ChenKPSSY21} (see~\cite{AssadiCK19b} for an in-depth discussion
on this topic). Even for the seemingly algorithmically harder problem of finding a perfect matching, $\eps=n^{-1}$, or nearly-perfect, $\eps = n^{-\Omega(1)}$, the best lower bounds are only $\Omega(\log{n})=\Omega(\log{(1/\eps)})$ passes~\cite{GuruswamiO13,ChenKPSSY21} (in contrast, the best known upper bounds for perfect matchings are $n^{3/4+o(1)}$ passes~\cite{AssadiJJST22}). Thus, going beyond $(\log{(1/\eps)})$ passes seems to require
fundamentally new techniques and the first step would be improving perfect matching lower bounds beyond $(\log{n)}$ passes, which is another fascinating open question. 

\subsection{Our Techniques}\label{sec:techniques}

We follow the \emph{set hiding} approach of~\cite{AssadiR20,ChenKPSSY21,Assadi22} and an elegant recursive framework of~\cite{ChenKPSSY21} that achieves \emph{permutation hiding} (a primitive entirely missing from~\cite{AssadiR20,Assadi22} and seemingly crucial for proving more than two-pass lower bounds). At a high level, $p$-pass permutation hiding graphs hide a unique permutation of vertex-disjoint augmenting paths---necessary for finding large enough matchings---, in a way that a semi-streaming algorithm cannot find this permutation in $p$ passes (see~\Cref{sec:perm-hiding}); a set hiding graph roughly corresponds to only hiding the endpoints of these paths. 
\cite{ChenKPSSY21} shows a way of constructing $p$-pass set hiding graphs from $(p-1)$-pass permutation hiding graphs (that can be made efficient), and constructing $p$-pass permutation hiding graphs from $p$-pass set hiding graphs rather inefficiently by blowing up the number of vertices by a $\Theta(\log{n})$ factor. This results in having to reduce $\eps$ to $\eps/\Theta(\log{n})$ for each application of this idea in each pass, leading to a lower bound of $\Omega(\log{(1/\eps)}/\log\log{n})$ passes eventually.

In a nutshell, we present a novel approach for directly constructing permutation hiding graphs, without cycling through set hiding ones first and thus avoiding the $\Theta(\log{n})$ overhead of~\cite{ChenKPSSY21} 
in  vertices and the approximation factor. Hence, we only need to reduce $\eps$ to some $\Theta(\eps)$ for each pass, leading to our $\Omega(\log{(1/\eps)})$ lower bound. 
Conceptually, the  technical novelty of our paper can be summarized as working with permutations \emph{directly} both in the construction and in the analysis. 

We first give a novel combinatorial approach for hiding permutations directly inside RS graphs, instead of only using them for hiding Boolean strings and applying Boolean operators
on top of RS graphs as was done previously (e.g., $\wedge$- or $\vee$- operators of~\cite{ChenKPSSY21}). This step uses various ideas developed in~\cite{GoelKK12,AssadiKL17,AssadiR20,AssadiB21a,Assadi22,AssadiS23} 
(see~\cite[Section 1.2]{Assadi22} for an overview of these techniques)
that can then be combined with the general framework of~\cite{ChenKPSSY21} in a non-black-box way, for instance, by replacing  sorting network ideas of~\cite{ChenKPSSY21} 
with $k$-sorter networks in~\cite{ParkerP89,chvatal1992lecture} with lower depth (and various technical changes in the analysis).

The second and the main technical ingredient of our paper 
is to introduce and analyze a ``permutation variant'' of the \emph{Boolean Hidden (Hyper)Matching (BHH)} communication problem of~\cite{GavinskyKKRW07,VerbinY11}. 
The BHH problem, alongside its proof ideas, has been a key ingredient of  streaming matching lower bounds, among many others, in recent years~\cite{AssadiKSY20,ChenKPSSY21,AssadiN21,Assadi22,KapralovMTWZ22,AssadiS23} (see~\cite[Appendix B]{AssadiKSY20} for an overview). 
Roughly speaking, our problem (see~\Cref{sec:hph}), replaces the hidden Boolean strings in BHH and their XOR operator  with permutations and the concatenation operator. Its analysis then boils down to proving a hardness amplification result for the concatenation of permutations, quite similar in spirit to XOR Lemmas for Boolean strings. 
While these XOR Lemmas are primarily proven using Fourier analysis on Boolean hypercube (see, e.g.,~\cite{GavinskyKKRW07,VerbinY11,KapralovKS15,KapralovMTWZ22,Assadi22,AssadiS23}), and in particular
the KKL inequality~\cite{KahnKL88}, we prove our results by working with basic tools from representation theory and Fourier analysis on symmetric groups, combined with  
detailed information-theoretic arguments, including a recent KL-divergence vs $\ell_1/\ell_2$-distance inequality of~\cite{ChakrabartyK18}.

\renewcommand{\first}{\ensuremath{\textnormal{\textsc{First}}}}
\renewcommand{\last}{\ensuremath{\textnormal{\textsc{Last}}}}
\newcommand{\depth}[1]{\ensuremath{\textnormal{\textsc{depth}}(#1)}}
\newcommand{\mblock}[1]{\ensuremath{\textnormal{\textsc{Multi-Block}}(#1)}}

\newcommand{\gammaspec}{\ensuremath{\gamma^{*}}}

\newcommand{\btwn}{\ensuremath{\xi}}

\newcommand{\Left}{\ensuremath{\textnormal{\textsc{l}}}}
\newcommand{\Right}{\ensuremath{\textnormal{\textsc{r}}}}

\newcommand{\bipartite}[1]{\ensuremath{\textnormal{\textsc{Bipartite}}(#1)}}

\section{Main Result}\label{sec:result}

We  present our main theorem in this section that formalizes~\Cref{res:main-lb} from the introduction, plus the key ingredients we use to prove it.

We define a bipartite graph $\Grs=(\Lrs,\Rrs,\Ers)$ to be a $(2\nrs)$-vertex
bipartite $(r,t)$-RS graph if its edges can be partitioned into $t$ induced matchings $\Mrs{1},\ldots,\Mrs{t}$ each of size $r$; here, an induced matching means that there are no other edges between the endpoints of the matching. See~\Cref{sec:rs} for more details on RS graphs.

\begin{theorem}[Formalization of~\Cref{res:main-lb}]\label{thm:main-lb}
	Suppose that for infinitely many choices of $\nrs \geq 1$, there exists $(2\nrs)$-vertex bipartite $(r,t)$-RS graphs with $r = \alpha \cdot \nrs$ and $t = \paren{\nrs}^{\beta}$ for some fixed parameters $\alpha, \beta \in (0,1)$; the parameters $\alpha,\beta$ can depend
	on $\nrs$. 
	
	Then, there exists an $\eps_0 = \eps_0(\alpha,\beta)$ such that the following is true. For any $0 < \eps < \eps_0$, any streaming algorithm that uses $o(\eps^{2} \cdot n^{1+\beta/2})$ space on $n$-vertex bipartite graphs 
	and can determine with constant probability whether the input graph has a perfect matching or its maximum matchings have size at most $(1-\eps) \cdot n/2$ requires 
	\[
	\Omega\Paren{\dfrac{\log{(1/\eps)}}{\log{(1/\alpha\beta)}}}
	\]
	passes over the stream. 
\end{theorem}

\Cref{res:main-lb} then follows from~\Cref{thm:main-lb} by setting $\alpha = \Theta(1)$ and using the bound $\eps > n^{-\beta/6}$ to obtain a space lower bound of $\Omega(n^{1+\beta/6})$ 
for $o(\log{(1/\eps)}/\log{(1/\beta)})$ pass algorithms by~\Cref{thm:main-lb}. This is because any $(1-\eps)$-approximation of size of  maximum matchings  distinguishes between the two families of the graphs in the theorem. 
Finally, given that we know by~\cite{FischerLNRRS02,GoelKK12} that $\beta = \Omega(1/\log\log{(n)})$, the space bound of $o(n^{1+\beta/6})$ will always rule out semi-streaming algorithms. 

We now go over the two main ingredients in the proof of this theorem, and state our main results for them. In the next section, we 
present a proof outline of each of these ingredients, plus that of~\Cref{thm:main-lb}. The rest of the paper is then dedicated to formalizing these proof outlines.



\subsection{Ingredient I: (Multi) Hidden Permutation Hypermatching}\label{sec:hph} 

A key to our lower bound constructions is a problem in spirit of the Boolean Hidden Hypermatching (BHH) problem of~\cite{VerbinY11} (itself based on~\cite{GavinskyKKRW07}) that we introduce in this paper. 
The definition of the problem is rather lengthy and can be daunting at first, so we build our way toward it by considering BHH first, and then move from there. 

\subsubsection*{The Boolean Hidden Hypermatching (BHH) problem}
The BHH problem can be phrased as follows (this is slightly different from the presentation in~\cite{VerbinY11} but is equivalent to the original problem). We have Alice with a string $x \in \set{0,1}^{r \times k}$ and Bob who has a hypermatching $\cM$ over $[r]^k$ with size $r/2$ plus a string $w \in \set{0,1}^{r/2}$. The players are promised that the parity of $x$ on hyperedges of $M$, i.e.,  
\[
\cM \cdot x = (\oplus_{i=1}^{k}x_{\cM_{1,i},i}~ ,~ \oplus_{i=1}^{k}x_{\cM_{2,i},i}~,~ \cdots~,~ \oplus_{i=1}^{k} x_{\cM_{r/2,i},i}) \in \set{0,1}^{r/2}
\]
is either equal to $w$ or $\bar{w}$. The goal is for Alice to send a single message to Bob and Bob outputs which case the input belongs to. 
It is known that $\Theta(r^{1-1/k})$ communication is necessary and sufficient for solving BHH with constant probability~\cite{VerbinY11}. 

A natural variant of BHH (defined as a direct-sum version of BHH) is also used in~\cite{ChenKPSSY21} as one of the main building blocks for constructing their permutation hiding graphs. Roughly speaking, in that problem, 
Alice is given several different strings $x$ and Bob's input additionally identifies which string to compute the parities of hyperedges of $M$ over. 

Nevertheless, the Boolean nature of this problem is  too restrictive for the purpose of our constructions (and in fact, this Boolean nature is the
key bottleneck in the construction of~\cite{ChenKPSSY21}). Thus, for our purpose, we  define a ``permutation variant'' of this problem. 

\subsubsection*{The Hidden Permutation Hypermatching (HPH) Problem}

We define the \textbf{Hidden Permutation Hypermatching (HPH)} problem as follows. Let $b \geq 1$ and $S_b$ be the set of permutations over $[b]$. We have a permutation matrix $\Sigma \in (S_b)^{r \times k}$ and
a hypermatching $\cM$ over $[r]^k$ with size $r/2$, plus a permutation vector $\Gamma \in (S_b)^{r/2}$. We are promised: 
\begin{align*}
	\Gamma^* = \paren{\circ_{i=1}^{k} \sigma_{\cM_{1,i},i}~,~\circ_{i=1}^{k} \sigma_{\cM_{2,i},i}~,~\cdots~,~,\circ_{i=1}^{k} \sigma_{\cM_{r/2,i},i}} \in (S_b)^{r/2}
\end{align*}
is such that $\Gamma^* \circ \Gamma$ is either equal to one of the two fixed known permutation vectors $\GammaY$ or $\GammaN$; here, $\circ_{i=1}^{k}(\cdot)$ concatenates the given $k$ permutations together. 
The goal is to distinguish which case the input belongs to. See~\Cref{fig:hph} for an illustration. 

\medskip
\begin{figure}[h!]
	\centering
	\input{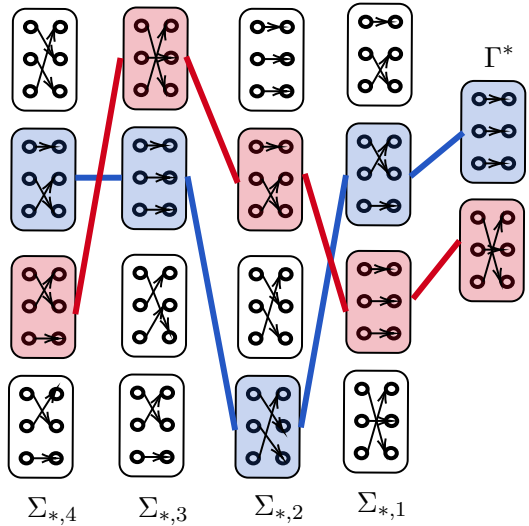}\caption{An illustration the HPH problem for $r=4$, $k=4$ and $b=3$. We have $r/2=2$ hypermatching edges corresponding to the thick edges (red and blue, respectively, for each hyperedge).}\label{fig:hph}
\end{figure}

One can see that this problem is equivalent to BHH whenever $b=2$: we can simply interpret the identity permutation in $S_2$ as a $0$-bit and the cross permutation  as a $1$-bit; the concatenation operator 
in this case will then become XOR naturally and the problem will be identical to BHH. 

Nonetheless, once we move on to larger values of $b > 2$, this problem becomes much ``richer'' than BHH as it can encode 
$b!$ different ``states'' per each entry and the concatenation operator becomes quite different than XOR. Indeed, there are already generalizations of BHH by replacing Boolean domain with finite fields $\mathbb{F}_b$ for $b > 2$, and 
the XOR operator with addition  in this field~\cite{GuruswamiT19}; yet, even those generalizations only correspond to very limited types of permutations and concatenations, and are strict special cases of HPH (which also 
do not work for our constructions as they can only hold $b$ ``states'' per entry as opposed to $b!$). 

There is however an important subtlety in the definition of this problem in our paper that we need to clarify. While we could have turned HPH into a communication game, exactly as in BHH, by giving $\Sigma$ to Alice and 
$\cM,\Gamma$ to Bob, this is \emph{not} what we do actually. Instead, we partition $\Sigma$ column-wise between $k$ different players $\QP{1},\ldots,\QP{k}$ by providing each player $\QP{i}$ for $i \in [k]$ with the 
permutation vector $\Sigma_{*,i} \in (S_b)^{r}$, the $i$-th column of $\Sigma$, and provide $(\cM,\Gamma)$ to a referee (identical to Bob). The communication pattern is also different in that, first, $\QP{1},\ldots,\QP{k}$ can talk with each other, with back and forth communication, 
using a shared blackboard visible to all parties. Then, at the end of their communication, the referee can check the content of the blackboard and output the answer (think of the final state of the blackboard as the message of players to the referee). 
We shall discuss the technical reasons behind this change in our proof outline in~\Cref{sec:overview-hph}.


\subsubsection*{The Multi Hidden Permutation Hypermatching (Multi-HPH) Problem}

We are now ready to present the full version of the problem we study, which can be seen as a certain direct-sum variant of the HPH (similar in spirit to the way BHH is generalized in~\cite{ChenKPSSY21}). 
Roughly speaking, in this problem, there are $t$ instances of HPH and the referee additionally chooses which instance of the HPH they should all solve, without the other players knowing this information
at the time of their communication.

We do caution the reader that since this problem involves ``tensorizing'' the inputs in HPH (e.g., $\Sigma$ becoming a $3$-dimensional tensor and each player receiving a permutation matrix), the indices stated below may not directly map to the ones
stated earlier for HPH (which is the $t=1$ case).

\begin{problem}[\textbf{Multi-Hidden-Permutation-Hypermatching}]\label{prob:mph}
	For any integers $r,t,b,k \geq 1$, $\IndexHPH_{r,t,b,k}$ is a distributional $(k+1)$-communication game, consisting of $k$ players plus a referee, defined as: 	
	\begin{enumerate}[label=$(\roman*)$]
		\item Let $\GammaY, \GammaN \in \paren{S_b}^{r/2}$ be two arbitrary tuples of permutations known to all $(k+1)$ parties. We refer to $\GammaY$ and $\GammaN$ as the \textbf{target tuples}. 
		\item We have $k$ players $\QP{1},\ldots,\QP{k}$ and for all $i \in [k]$ the $i$-th player is given a permutation matrix $\SigmaP{i} \in (S_b)^{t \times r}$ chosen independently and uniformly at random. 
		
		\item The referee receives $k$ indices $L := (\ell_{1},\ldots,\ell_{k})$ each picked uniformly and independently from $[t]$ and a hypermatching $\cM \subseteq [r]^k$ with $r/2$ hyperedges picked uniformly. 	
		
		Additionally, let the permutation vector $\Gammastar = \paren{\gammastar_1, \gammastar_2, \ldots, \gammastar_{r/2}} \in \paren{S_b}^{r/2}$ be defined as,
		\[
		\forall a \in [r/2] \hspace{2mm} \gammastar_a := \sigmaP{1}_{\ell_1, \cM_{a,1}} \circ \cdots \circ \sigmaP{k}_{\ell_k,\cM_{a,k}},
		\]
		where $\cM_{a,i}$ for $a \in [r/2]$ and $i \in [k]$ refers to the $i$-th vertex of the $a$-th hyperedge in $\cM$. 
		
		The referee is also given another permutation vector $\Gamma = (\gamma_1, \gamma_2, \ldots, \gamma_{r/2}) $ sampled from $\paren{S_b}^{r/2}$ conditioned on  $\Gammastar \conc \Gamma$  being equal to either $\GammaY$ or  $\GammaN$ (here, the $\conc$ operator is used to concatenate each permutation in the vectors individually). 
		
	\end{enumerate}
	\noindent
	The players $\QP{1},\ldots,\QP{k}$ can communicate with each other by writing on a shared board visible to all parties with possible back and forth and in no fixed order (the players' messages are functions of their inputs and the board). 
	At the end of the players' communication, the referee can use all these messages plus its input and outputs whether $\Gammastar \conc \Gamma$ is $\GammaY$ or $\GammaN$. 
\end{problem}

The following theorem on the communication complexity of $\IndexHPH$ involves the bulk of our technical efforts in this paper. 

\begin{theorem}\label{thm:mph}
	For any integers $t \geq 1$, $b \geq 2$, and sufficiently large $r, k \geq 1$, any  protocol for $\IndexHPH_{r,t,b,k}$, for any pairs of target tuples, with at most $s$ bits of total 
	communication for $s$ satisfying
	\[
	k \cdot \log{(r \cdot t)} \leq s \leq 10^{-3} \cdot (r \cdot t)
	\]
	can only succeed with probability at most 
	\[
	\frac{1}{2} + r \cdot O\!\paren{\frac{s}{r \cdot t}}^{k/32}.
	\]
\end{theorem}

In our lower bounds, we use this theorem with parameters $k = \Theta(1/\beta)$, and $s = r \cdot \sqrt{t}$ (where $r$ and $t \approx r^{\beta}$ will be determined by the underlying $(r,t)$-RS graph we use). This allows us to prove a lower bound for protocols that can solve $\IndexHPH$ with advantage as as small as $1/\poly(r)$ over random guessing, which will be crucial for our constructions.

\subsection{Ingredient II: Permutation Hiding Graphs}\label{sec:perm-hiding}

We now present the definition of permutation hiding graphs, introduced by~\cite{ChenKPSSY21}, that are also a key building block in our lower bound. To do this, we need some notation.

We call a directed acyclic graph $G=(V,E)$ a \textbf{layered} graph if its vertices can be partitioned into sets $V^1,\ldots,V^d$ for some $d \geq 1$, 
such that any edge of $G$ is directed from some $V^i$ to $V^{i+1}$ for some $i \in [d-1]$; we further define $\first{G} := V^1$ and $\last{G} := V^d$. 

We define a \textbf{layered graph} as $G = (V,E)$ as a directed acyclic graph whose vertex set $V$ can be partitioned into $d$ different sets $V_1, V_2, \ldots, V_{d}$ where any edge $(u,v)$ is such that $u \in L_i$ and $v \in L_{i+1}$ for some $i \in [d-1]$. We use $\first(G)$ and $\last(G)$ to denote the sets $V_1$ (the first layer of $G$) and $V_{\ell}$ (the last layer of $G$) respectively. We index each layer $V_i$ by the set $[\card{V_i}]$ for each $i \in [d]$. We use $\first_{[i]}(G)$ to denote the first $i$ vertices of $\first(G)$ indexed by the set $[i]$ for any $i \in [\card{\first(G)}]$ and similarly for $\last_{[j]}(G)$ for $j \in [\card{\last(G)}]$.

We will use layered graphs to represent permutations over $[m]$, denoted by $S_m$, as follows.

\medskip

\begin{Definition}[Permutation Graph]\label{def:perm-graph}
	For any integer $m \geq 1$, a layered graph $G = (V, E)$ is said to be a \textbf{permutation graph} for $\sigma  \in S_m$ if $\card{\first(G)},\card{\last(G)} \geq m$ and there is a path from $i \in \first_{[m]}(G)$ to $j \in \last_{[m]}(G)$ if and only if $\sigma(i) = j$ for each $i,j \in [m]$.
\end{Definition}

\smallskip

We use $\cD_m$ to denote the set of all permutation graphs on $S_m$. 
We are now ready to define the (distribution of) permutation hiding graphs. Our definition is a slight rephrasing of the one in~\cite{ChenKPSSY21} and we claim no novelty here. 

\begin{Definition}[Permutation Hiding Graphs; c.f.~\cite{ChenKPSSY21}]\label{def:perm-hiding}
	For  integers $m,n,p,s \geq 1$ and real $\delta \in (0,1)$, we define a \textbf{permutation hiding generator} $\GG=\GG(m,n,p,s,\delta)$ as any family of distributions $\GG: S_m \rightarrow \PP_m$ on permutation graphs satisfying the following two properties:
	\begin{enumerate}[label=$(\roman*)$]
		\item For any $\sigma \in S_m$, any permutation graph $G$ in the support of $\GG(\sigma)$ has $n$ vertices. 
		\item For any $\sigma_1,\sigma_2 \in S_m$, the distribution of graphs $\GG(\sigma_1)$ and $\GG(\sigma_2)$ are $\delta$-indistinguishable 
		for any $p$-pass $s$-space streaming algorithm. 
	\end{enumerate}
\end{Definition}

\cite{ChenKPSSY21} presented a permutation hiding generator with the following parameters for any integers $m,p \geq 1$ in terms of the parameters $\alpha$ and $\beta$ of RS graphs: 
\begin{align*}
	&n := {(\frac{\log{n}}{\alpha})}^{\Theta(p)} \cdot \Theta(m) \tag{number of vertices}; \\
	&s := o(m^{1+\beta}) \tag{space of streaming algorithm}; \\
	&\delta := 1/\poly{(n)} \tag{probability of success of the algorithm}. 
\end{align*}
In particular, notice that the number of vertices in the graph grows by a factor of $\log{(n)}$ per pass even when both $\alpha$ and $\beta$ are constant. As we shall see later in this section, 
the ratio of $m/n$ governs the approximation ratio of the algorithms for the maximum matching problem. As such, the ratio in this construction is too small to provide lower bounds for constant-factor approximation streaming algorithms, no matter
the space of the algorithm.  

We present an alternative construction of permutation hiding graphs in this paper, which is more suited for proving streaming lower bounds for maximum matching. 

\begin{theorem}\label{thm:perm-hiding}
	For any integers $m,p \geq 1$, there is a permutation hiding generator $\GG=\GG(m,n,p,s,\delta)$ with the following parameters: 
	\begin{align*}
		&n := {\Theta(\frac{1}{\alpha \cdot \beta^2})}^{p} \cdot \Theta(m) \tag{number of vertices}; \\
		&s := o(m^{1+\beta/2}) \tag{space of streaming algorithm}; \\
		&\delta := (p/\beta)^{\Theta(1/\beta)}\cdot \Theta(1/\beta)^{2p} \cdot 1/\poly{(m)} \tag{probability of success of the algorithm}. 
	\end{align*}
\end{theorem}
Thus, we obtain a different tradeoff than~\cite{ChenKPSSY21} on the parameters of the permutation hiding generator. In particular, now, 
the ratio of $n/m$ is only $2^{\Theta(p)}$ for constant values of $\alpha,\beta$, which as we shall see soon, is sufficient to obtain our $\Omega(\log{(1/\eps)})$-pass lower bound. 

\medskip

Our proof of~\Cref{thm:perm-hiding} considerably deviates from that of~\cite{ChenKPSSY21} both in terms of the combinatorial construction of the permutation hiding graphs and even more so in the information-theoretic analysis of their properties. 
Since the majority of these changes are already apparent even for single-pass algorithms, in~\Cref{sec:one-pass} we first present the construction for single-pass algorithms 
separately as a warm-up to our main construction.~\Cref{sec:multi} then contains the construction and analysis for multi-pass algorithms using the same type of inductive argument as in~\cite{ChenKPSSY21} by replacing 
their induction step with our approach in~\Cref{sec:one-pass} for single-pass algorithms.


\section{Proof Outline}\label{sec:outline}

We present a proof outline of~\Cref{thm:main-lb} and its ingredients in this section. We emphasize that this section oversimplifies many details and the discussions 
will be informal for the sake of intuition. We will start with the two key ingredients of our main theorem, namely,~\Cref{thm:mph,thm:perm-hiding}, and then at the end, 
show how they can easily imply the main theorem as well. 

\subsection{Proof Outline of~\Cref{thm:mph}: (Multi) Hidden Permutation Hypermatching}\label{sec:overview-hph} 

We start with the proof outline of~\Cref{thm:mph}. The formal proof is presented in~\Cref{sec:lb-multi-hph}. 
\subsubsection{From XOR Lemmas to a ``Concatenation Lemma''}
Let us focus on the HPH problem wherein the input is a permutation matrix $\Sigma \in (S_b)^{r \times k}$ given to $k$ players $\QP{1},\ldots,\QP{k}$, and a hypermatching $\cM$, and permutation vector $\Gamma \in (S_b)^{r/2}$
given to the referee. 
The goal is to determine whether for the permutation vector $\Gammastar$ obtained via concatenating permutations on indices of $\cM$, $\Gammastar \circ \Gamma$ is equal to a fixed 
permutation vector $\GammaY$ or another vector $\GammaN$. The communication pattern is as specified in~\Cref{prob:mph}. 

Our starting point is similar to that of~\cite{GavinskyKKRW07,VerbinY11} by breaking
the correlation in the input instance (either all permutations in $\Gammastar \circ \Gamma$ are consistent with $\GammaY$ or all are consistent with $\GammaN$). 
Roughly speaking, this corresponds to showing that if we only have a single random hyperedge $e=(v_1,\ldots,v_k) \in [r]^k$, then, the distribution of
\[
\gamma^* = \sigma_{v_1,1} ~\circ ~\cdots ~\circ~ \sigma_{v_k,k};
\]
is $1/\poly(r)$ close to the uniform distribution in the total variation distance. Having proven this, one can then essentially do a ``union bound'' and show that the distribution of the entire vector
$\Gammastar \circ \Gamma$ remains so close to uniform that the referee cannot distinguish whether it is consistent with $\GammaY$ or $\GammaN$ in this case. Formalizing this step
is not immediate and requires a \emph{hybrid argument} similar to those of~\cite{GavinskyKKRW07,VerbinY11} (see~\cite[Section 4.3]{AssadiN21} also for a general treatment) 
but we shall skip it in this discussion. 

Concretely, our task is to prove the following inequality for any low communication protocol $\prot$ for HPH, with a transcript $\Prot$ between its $k$ players $\QP{1},\ldots,\QP{k}$:
\begin{align}
	\Exp_{(\Prot,e)} \tvd{\paren{\gamma^* \mid \Prot,e}}{\cU_{S_b}} \leq 1/\poly(r), \label{eq:ss1}
\end{align} 
where $\cU_{S_b}$ is the uniform distribution over $S_b$ and $\norm{\cdot}_{\textnormal{tvd}}$ is the total variation distance (TVD). 

\paragraph{XOR Lemmas.} By the equivalence between HPH and BHH for $b=2$, proving the equivalent of~\Cref{eq:ss1} in $b=2$ case for BHH in all prior work~\cite{GavinskyKKRW07,VerbinY11,KapralovKS15,ChenKPSSY21,KapralovMTWZ22,Assadi22,AssadiS23}
corresponds to 
proving an \emph{XOR Lemma} for the \emph{Index} communication problem: We have a string $x \in \set{0,1}^{r \times k}$ conditioned on a short message $\Prot$, and 
we are interested in $\oplus_{i=1}^{k} x_{v_i,i}$ for $(v_1,\ldots,v_k)$ chosen uniformly from $[r]^k$. At an intuitive level, these works rely on the following two statements: $(i)$ since 
$\Prot$ is a short message, each $x_{v_i,i}$ should be individually somewhat close to uniform distribution (by the standard lower bounds for the Index problem~\cite{Ablayev93,KremerNR95}),
and $(ii)$ since we are taking XOR of multiple close-to-uniform bits, the final outcome should be even exponentially-in-$k$ closer to uniform\footnote{This step is easy to see \emph{had} the message $\Prot$ was not correlating
	the values of $x_{v_i,i}$ across different $i \in [k]$: the bias of XOR of independent bits is equal to multiplication of their biases (see, e.g.~\cite[Proposition A.9]{AssadiN21}).  Yet the message can indeed correlate these 
	values and the main challenge in proving any XOR Lemma is to handle this. We refer the reader to Yao's XOR Lemma~\cite{Yao82a} for
	the first example of such approaches in circuit complexity, and~\cite{AssadiN21} for a streaming XOR Lemma and~\cite{Yu22} for a 
	bounded-round communication XOR Lemma.}. Formalizing this intuition is done via different tools such as Fourier analysis on Boolean hypercube~\cite{GavinskyKKRW07,VerbinY11,KapralovKS15,KapralovMTWZ22,Assadi22,AssadiS23}, 
discrepancy bounds~\cite{ChenKPSSY21} or a generic streaming XOR Lemma~\cite{AssadiN21}. 

Our approach here is then to extend these XOR Lemmas to a \textbf{``Concatenation Lemma''} for~\Cref{eq:ss1}, which we describe in the next part. 

\subsubsection{A ``Concatenation Lemma''}  

Following the previous discussion, our goal in proving~\Cref{eq:ss1} is to show that: $(i)$ since 
$\Prot$ is a short message, each permutation $\sigma_{v_i,i}$ should be individually somewhat close to the uniform distribution (which still follows from a similar argument as for the Index problem~\cite{Ablayev93,KremerNR95}),
and $(ii)$ since we are taking concatenation of multiple close-to-uniform permutations, the final outcome should be even exponentially-in-$k$ closer to uniform. 
This is what we consider a ``concatenation lemma'' in this paper, and it is where we start to deviate completely from prior approaches in~\cite{GavinskyKKRW07,VerbinY11,KapralovKS15,ChenKPSSY21,KapralovMTWZ22,Assadi22,AssadiS23}
in this context. 

\paragraph{A Conceptual Roadblock.} Let us point out an important conceptual challenge in formalizing step $(ii)$ of this plan. 
Let $\nu$ denote the uniform distribution over all \emph{even} permutations on $S_b$, namely, permutations with an even number of inversions. 
We have that the TVD of $\nu$ from $\cU_{S_b}$ is $1/2$, so, roughly speaking, $\nu$ is already ``not-too-far'' from the uniform distribution. 

But now consider the distribution $\nu^*$ which is the $k$-fold concatenation of $\nu$ for some $k \geq 1$, meaning that to sample from $\nu^*$, we sample $k$ independent 
permutations from $\nu$ and concatenate them together. To be able to implement step $(ii)$ of our plan, at the very least, we should have that in this purely independent case, 
TVD of $\nu^*$ from $\cU_{S_b}$ exponentially drops, i.e., becomes $2^{-\Omega(k)}$. Alas, it is easy to see that $\nu^*$ is in fact the same as $\nu$ and thus we have \emph{no} change in TVD at all! This is 
in stark contrast with the XOR and Boolean case where the drop in TVD, at least for purely independent inputs, is \emph{always} happening.

\subsubsection{Our Proof Strategy for the ``Concatenation Lemma''} With the above example in mind, we are now ready to discuss our solution for proving the Concatenation Lemma and establishing~\Cref{eq:ss1}. 
For now, let us \underline{limit ourselves} to \emph{nice} protocols that do {not} correlate the outcome of different permutations with each other, meaning that
we are still in this blissful case wherein the permutations $\sigma_{v_i,i}$ are independent even conditioned on $\Prot$. The first and easy step of the argument is to show
that the distribution of each $\sigma_{v_i,i}$, for a random $v_i$, is close to uniform not only in TVD but also KL-divergence, namely, 
\begin{align}
	\Exp_{\Prot,v_i}\bracket{\kl{\sigma_{v_i,i} \mid \Prot}{\cU_{S_b}}} \leq r^{-1/2}. \label{eq:ss2}
\end{align}
This step is still not particularly different from the typical Index lower bounds and is an easy application of the chain rule of KL-divergence. Let us make one further \underline{simplifying assumption} by 
taking~\Cref{eq:ss2} to hold, not in expectation, but rather simultaneously for all coordinates of a fixed $(v_{1},\ldots,v_{k})$ at the same time.  

At this point, one could apply Pinsker's inequality (\Cref{fact:pinskers}) to relate the KL-divergence bound in~\Cref{eq:ss2} to a bound on TVD, but then we may end up in the situation shown in the roadblock, hence 
not allowing for further decay in the distance through concatenation.  

Another approach, taken for instance in~\cite{AssadiKSY20}, is to relate~\Cref{eq:ss2} to an $\ell_2$-distance of the distributions (instead of TVD which is half the $\ell_1$-distance). But then, the best bound one can prove on $\ell_2$-distance will also be $r^{-\Theta(1)}$, which again would not suffice for our purpose (using a ``smooth'' version of the example in the roadblock; see~\Cref{app:tight-remark}). 

Both above cases suggest that we may need a more nuanced understanding of the distribution of $\sigma_{v_i,i} \mid \Prot$. To do so, we consider a combination of TVD and $\ell_2$-distances in the following way. 
We apply the ``KL-vs-$\ell_1/\ell_2$-inequality'' of~\cite{ChakrabartyK18} (\Cref{prop:pinsker++})---a generalization of the Pinsker's inequality---to decompose the support of the distribution of $\sigma_{v_i,i} \mid \Prot$ into two parts: 
\begin{itemize}
	\item $A_i$: A part that does \emph{not} happen that frequently, meaning  $\sigma_{v_i,i} \mid \Prot$ is only in  $A_i$ w.p. $r^{-1/4}$;
	\item $B_i$: A part that is \emph{extremely close} to uniform distribution in $\ell_2$-distance, meaning 
	\begin{align}
		\norm{(\sigma_{v_i,i} \mid \Prot,B_i)-\cU_{S_b}}_2 \leq \frac{r^{-\Theta(1)}}{b!}. \label{eq:ss3}
	\end{align}
\end{itemize}
Putting these together, plus our simplifying assumption that $\sigma_{v_i,i}$'s are still independent conditioned on $\Prot$ allows us to argue that with probability $1-r^{-\Theta(k)}$, there 
are at least $\Theta(k)$ coordinates $i \in [k]$ that satisfy~\Cref{eq:ss3}. 

This brings us to the last part of the argument. Having obtained $\Theta(k)$ coordinates that are extremely close to uniform, we use basic tools from representation theory and Fourier analysis on permutations
to analyze the distribution of their concatenation. The Fourier basis here is a set of irreducible representation matrices (see~\Cref{app:fourier-permutation}), and the convolution theorem
allows us to relate Fourier coefficients of the concatenated permutation via multiplication of Fourier coefficients of each individual permutation. Finally, the distance 
to the uniform distribution can be bounded by Plancharel's inequality for this Fourier transform (\Cref{prop:Four-plancherel}), similar to the standard analysis on the Boolean hypercube (see, e.g.,~\cite{Wolf08}). 

All in all, this step allows us to bound the TVD of the concatenation of the permutation---conditioned on the case of $\Theta(k)$ extremely-close to uniform indices in $[k]$ which happens with probability $1-r^{-\Theta(k)}$---by another $r^{-\Theta(k)}$. Putting all these together then gives us the desired inequality in~\Cref{eq:ss1} under all our earlier simplifying assumptions. 

\paragraph{Removing simplifying assumptions.} The above discussion oversimplified many details, chief among them, the main challenge that stems from the inputs becoming correlated through the transcript $\Prot$ (which is the \emph{key} challenge
in proving XOR Lemmas as well). For the BHH problem and XOR Lemmas (for the Index problem) in the Boolean setting, a key tool to handle this is the \emph{KKL inequality} of~\cite{KahnKL88} for the Fourier transform
on Boolean hypercube, which as a corollary, almost immediately gives an XOR Lemma for the Index problem (see~\cite[Section 4.2]{Wolf08}). 
For our permutation problem, however, we are not aware of any similar counterpart. 

We handle the aforementioned challenge by replacing the role of a single player, Alice, in BHH with $k$ separate players in HPH and use a detailed information-theoretic argument to analyze how much these players can correlate their inputs. 
Roughly speaking, this reduces the problem to proving that the inputs of players are only correlated through the message $\Prot$ and not beyond that (a consequence of the rectangle property of communication protocol), and then 
making a \emph{direct product} argument to show a single message $\Prot$, cannot, \emph{simultaneously}, change the distribution of multiple coordinates. 

\subsubsection{From HPH to \IndexHPH: (Not) A Direct-Sum Result}

To prove~\Cref{thm:mph}, we need a lower bound for the $\IndexHPH$ problem which is stronger than the lower bound for HPH by a factor of $t$. Given that $\IndexHPH$ is effectively a direct-sum version of HPH---we need to solve one \emph{unknown} copy
out of $t$ given copies---it is natural to expect the complexity of the problem also increases by a factor $t$. Moreover, given various direct-sum results known using \emph{information complexity} (see, e.g.~\cite{ChakrabartiSWY01,JainRS03,BarakBCR10,BravermanR11,BravermanRWY13} and references therein), one might expect this step to be an easy corollary. 

Unfortunately, this is in fact not the case, due to the crucial reason that we need a lower bound for protocols with an extremely small advantage of only $1/\poly(r)$ over random guessing. In general, one should not expect a generic 
direct sum result to hold in this low-probability regime\footnote{A short explanation is based on the equivalence of information complexity and direct sum result~\cite{BravermanR11}, plus the fact that information complexity is an ``expected'' term
	while communication complexity is a ``worst case'' measure.}. Moreover, these information complexity approaches typically fail on problems with a low probability of success, and in our case, they 
cannot be applied readily.\footnote{Concretely, a protocol that with probability $1/\poly(r)$, communicates its input and otherwise is silent has an $O(1)$ information complexity (albeit large communication complexity) and can 
	lead to the desired advantage of $1/\poly(r)$ trivially; thus information complexity of HPH in such a low-probability-of-success regime is simply $O(1)$.} 

Consequently, in our proof of~\Cref{thm:mph}, we directly work with the $\IndexHPH$ problem, which means all the arguments stated in the previous part should be implemented 
for this ``higher'' dimensional problem. Nevertheless, most of these changes appear in the first part of the argument that reduces the original problem to proving~\Cref{eq:ss1} and~\Cref{eq:ss2} (or rather their equivalents for $\IndexHPH$), 
as well as the direct product arguments mentioned at the end of the last subsection for removing our simplifying assumptions. Thus, these changes do not fundamentally alter our previously stated plan in the lower bound and 
we postpone their details to~\Cref{sec:lb-multi-hph}.

\subsection{Proof Outline of~\Cref{thm:perm-hiding}: Permutation Hiding Graphs}\label{sec:overview-perm-hiding}

We now switch to the proof outline of~\Cref{thm:perm-hiding}. The formal proof is presented in~\Cref{sec:one-pass} for single-pass algorithms and~\Cref{sec:multi} for multi-pass ones. 

Our proof primarily builds on~\cite{ChenKPSSY21} (for the general framework) and~\cite{Assadi22} (for the construction of the graphs used in the framework). 
Both these papers have excellent overviews of their technical approach,~\cite[Section 2]{ChenKPSSY21} and~\cite[Section 1.2]{Assadi22}, and we refer the reader to those parts
for further background as well as an overview of prior techniques.

As stated earlier, our main point of departure from the framework of~\cite{ChenKPSSY21} already appears for single-pass algorithms. So, in this overview also, we mostly focus on this case. 
For the simplicity of exposition, in the following, we assume the parameters $\alpha,\beta$ of RS graphs are both $\Theta(1)$ and ignore the dependence of our bounds on them. 

\subsubsection{Permutation Hiding for Single-Pass Algorithms}

Recall the definition of permutation hiding graphs from~\Cref{def:perm-hiding}. 
\Cref{thm:mph} for $\IndexHPH$ is our main tool for constructing these graphs, so let us see how we can turn an instance $(\Sigma,L,\cM,\Gamma)$ of $\IndexHPH$ into a permutation graph. We start with each component separately. 

\subsubsection*{Encoded RS graphs and $\Sigma$}
Recall that $\Sigma$ consists of $k$ permutation matrices $\SigmaP{i} \in (S_b)^{t \times r}$. Let $\Grs = (\Lrs,\Rrs,\Ers)$ be a bipartite RS graph with $\nrs$ vertices
on each side and $t$ induced matchings $\Mrs{1},\ldots,\Mrs{t}$ each of size $r$ (for the same parameters as the dimensions of $\SigmaP{i}$). 

We  ``encode'' $\SigmaP{i}$ in $\Grs$ via the following ``graph product'' $\Grs \otimes \SigmaP{i}$ (strictly speaking, this is the product of the graph $\Grs$
and the matrix $\SigmaP{i}$):
\begin{itemize}
	\item Vertices of $\Grs \otimes \SigmaP{i}$ are the bipartition $L^{(i)}=\Lrs \times [b]$ and $R^{(i)} = \Rrs \times [b]$. 
	\item Edges of $\Grs \otimes \SigmaP{i}$ are directed from $(u,x) \in L$ to $(v,y) \in R$ whenever $(u,v)$ is an edge in $E$ \emph{and} $y=\sigmaP{i}_{c_1,c_2}(x)$ where $(c_1,c_2) \in [t] \times [r]$ is chosen
	so that $e$ is the $c_2$-th edge of the $c_1$-th induced matching $\Mrs{c_1}$. 
\end{itemize}

\begin{figure}[h!]
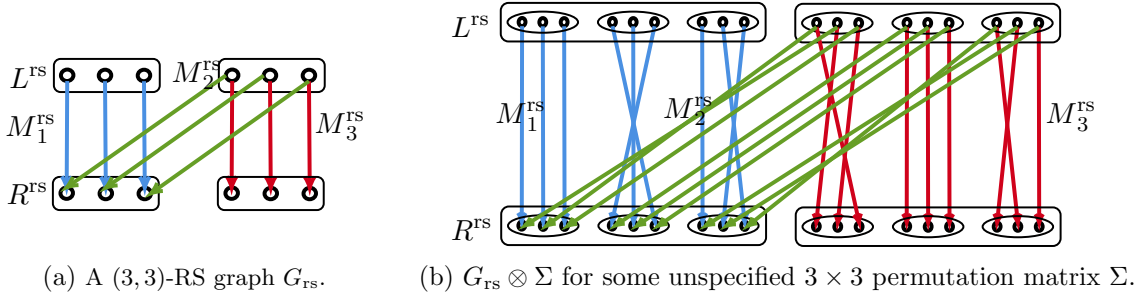

	\centering
	\subcaptionbox{\footnotesize A $(3,3)$-RS graph $\Grs$.}%
	[.3\linewidth]{
		\input{Figures/s-rsgraph}
	} 
	\hspace{0.2cm} 
	\subcaptionbox{$\Grs \otimes \Sigma$ for some unspecified $3 \times 3$ permutation matrix $\Sigma$.}%
	[.6\linewidth]{
		\input{Figures/s-encodedrs}
	}
	\caption{An illustration of RS graphs and our ``graph product''.}
	\label{fig:encoded-rs}
\end{figure}

In words, we ``stretch'' each vertex of $\Grs$ to become $b$ separate vertices, and replace the edge $e=(u,v) \in \Ers$ by a matching of size $b$ between vertices $(u,*)$ and $(v,*)$ in the product; the choice of this matching
is then determined by the entry of $\SigmaP{i}$ that ``corresponds'' to this edge, where in this correspondence we interpret rows of $\SigmaP{i}$ as the induced matchings in $\Grs$, and the columns as the edges of these matchings. 
See~\Cref{fig:encoded-rs} for an illustration.

It is not hard to verify that $\Grs \otimes \SigmaP{i}$ is itself an RS graph with $t$ induced matchings of size $r \cdot b$ on $\nrs \cdot b$ vertices. This means that the product has ``sparsified'' the original RS graph relatively, but 
we shall ensure that $b$ is sufficiently small, such that this product graph is also sufficiently dense still for the purpose of our lower bound. 

We can then encode the entirety of $\Sigma$ into $k$ vertex-disjoint RS graphs $\Grs \otimes \SigmaP{i}$ for $i \in [k]$. We refer to the graph consisting of these $k$ disjoint parts as $G^{\Sigma}$. 
We shall review the properties of this graph after defining the remaining graphs related to $\IndexHPH$. 

Before moving on, a quick remark is in order. RS graphs have been used extensively in the last decade for streaming lower bounds (see~\Cref{sec:rs}). However, \emph{all} prior work on RS graphs 
that we are aware of has used them for encoding Boolean strings. For instance,~\cite{AssadiR20} uses RS graphs to encode input sets to the \emph{set disjointness} communication problem, 
and~\cite{ChenKPSSY21} uses them for encoding their direct-sum BHH problem (outlined earlier). To obtain more ``complex'' structures such as permutations,~\cite{ChenKPSSY21} further defined Boolean operations of RS graphs 
like $\vee$, $\wedge$, and $\oplus$ and used them in conjunction with sorting networks (we shall describe this connection shortly also). This is precisely the source of the $\Theta(\log{n})$ loss in the parameters 
of permutation hiding graphs in~\cite{ChenKPSSY21} that we discussed earlier. 

Despite its simplicity, our new construction---inspired by~\cite{Assadi22}, which still encoded a Boolean string in the RS graph but in a similar fashion---is directly encoding a permutation matrix inside the RS graph, which allows for encoding more complex structures, without having to pay \emph{too much} on the density of the resulting graph. 

\subsubsection*{Simple permutation graphs and $(L,\cM,\Gamma)$}
We will create a new graph $H = H(G^{\Sigma},L,\cM,\Gamma)$ out of $G^{\Sigma}$. 
For the purpose of this part, we only care about the \emph{vertices} of $G^{\Sigma}$ and not its edges (this will be crucial for the reduction). We shall not go into the rather tedious definition of 
the graph $H$ here and instead simply mention its main properties (see also~\Cref{fig:H-over} for an illustration): 
\begin{itemize}
	\item The graph $H$ has asymptotically the same number of vertices as $G^{\Sigma}$ and in particular includes an extra set of vertices $S = [r/2] \times [b]$ on its ``left'' most part layer. 
	The extra edges inserted to $H$ at this step form perm perfect matchings between the layers of $H$. 
	\item If we start from the $a$-th group of vertices in $S$ for $a \in [r/2]$, namely, $(a,*) \in S$, 
	and follow the edges of the graph $H$,  we shall be moving according to the hyperedge $\cM_a$ across matchings $\Mrs{\ell_i} \otimes \SigmaP{i}$ for $i \in [k]$ in a way that 
	we will end at the same block of vertices $(a,*) $ in the last layer of $H$; the mapping between the groups $(a,*) \in S$ and $(a,*) $ in the last layer will now be equal to $\gammastar_a \circ \gamma_a$. 
\end{itemize}

\begin{figure}[h!]
	\centering
	\input{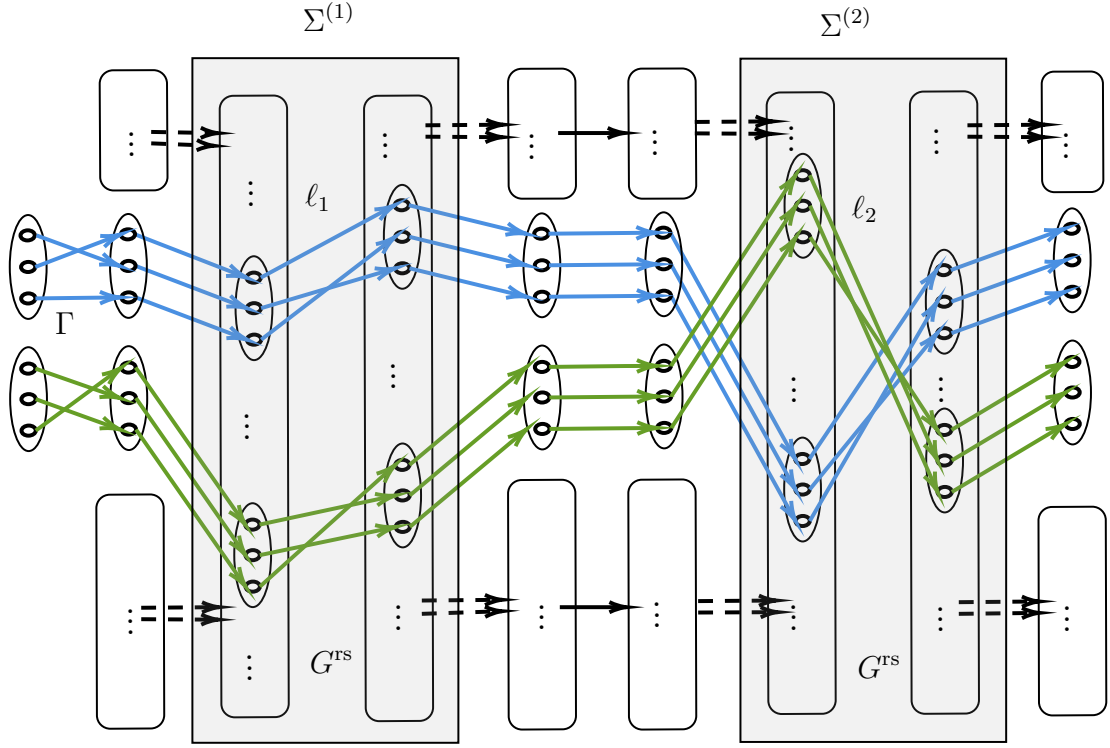}\caption{An illustration of the graph $H = H(G^{\Sigma},L,\cM,\Gamma)$ for some unspecified parameters.}\label{fig:H-over}
\end{figure}

With the properties above, one can see that the graph $H$ is in fact a permutation graph for the permutation induced\footnote{We 
	can interpret each permutation vector in $(S_b)^{r/2}$ as a permutation over $[r/2 \cdot b]$, wherein for every $a \in [r/2]$, the elements $[(a-1) \cdot b+1: a \cdot b]$ are permuted according to the $a$-the permutation
	of the vector. We emphasize that while every permutation vector leads to a permutation, the converse is not true, and this is in fact going to play an important role in our arguments.} 
by the permutation vector $\Gammastar \circ \Gamma \in (S_b)^{r/2}$ in $\IndexHPH$. Thus, learning the hidden permutation of the permutation graph $H$ is equivalent to 
figuring out whether the answer to $\IndexHPH$ was $\GammaY$ or $\GammaN$. 

We can now conclude the following. Given $\Sigma$, the players in $\IndexHPH$ can create the graph $G^{\Sigma}$ without any communication (as each player $\QP{i}$ will be responsible for creating $\Grs \otimes \SigmaP{i}$). The referee can also create the edges of the graph $H$ given the inputs $(L,\cM,\Gamma)$ as they \emph{do not depend} on $\Sigma$.  This in turn implies that a single-pass $s$-space streaming algorithm for learning the hidden permutation of the permutation graph $H$
can be used as a $(k \cdot s)$-communication protocol for $\IndexHPH$. Given our lower bound in~\Cref{thm:mph} for $\IndexHPH$ (for appropriate choices of $s$), this implies that the resulting distribution we have on the graphs $H$ is a 
permutation hiding generator with the important \emph{caveat} that it can only hide permutation vectors, as opposed to arbitrary permutations.

\subsubsection*{From Permutation Vectors to Arbitrary Permutations}
At this point, we have already constructed a generator for permutation vectors. This generator can be seen as a middle ground between the \emph{set hiding} generators of~\cite{ChenKPSSY21} (which are considerably weaker as they can hide  limited permutations, roughly corresponding to the $b=2$ case) and permutation hiding generators (which are considerably stronger, as they can hide every permutation).

Our last step to permutation hiding generators is inspired by a nice strategy of~\cite{ChenKPSSY21} for going from their {set hiding} graphs to their permutation hiding graphs. 
They show that one can use \emph{sorting networks} (see~\Cref{sec:comprators}), and in particular, the celebrated \emph{AKS sorters}~\cite{AjtaiKS83} to compute any arbitrary permutation 
as concatenation of $\Theta(\log{n})$ ``matching permutations'', namely, permutations that can only change two previously fixed pairs of elements with each other. This allows them to obtain a 
permutation hiding generator from $\Theta(\log{n})$ set hiding ones -- this is precisely the source of the factor $\Theta(\log{n})$ loss in the construction of~\cite{ChenKPSSY21}. 

For our purpose, we also use sorting networks but this time with larger comparators (typically called \emph{$b$-sorters}), wherein each comparator can sort $b$ wires simultaneously (see~\Cref{sec:comprators}). 
We can then use an extension of the sorting network of~\cite{AjtaiKS83} to $b$-sorter networks due to~\cite{chvatal1992lecture} with depth $O(\log_b{n})$ instead  (see also~\cite{ParkerP89} and~\Cref{app:sorting-stuff}). 
Finally, we show that our generator for hiding permutation vectors can ``simulate'' each layer of this network efficiently and use this to obtain a generator for every arbitrary permutation, by applying 
our generator for permutation vectors $O(\log_b{n})$ times. By taking $b$ to be a sufficiently small polynomial in $n$, we can ensure that $O(\log_b{n}) = O(1)$ and thus only 
pay a constant factor overhead when constructing our final permutation hiding generator for single-pass algorithms. Finally, $b$ is still sufficiently small such that even though 
in our encoded RS graph, we essentially make the input graph sparser by a factor of $b$, the graph is still sufficiently dense to allow for our desired lower bound. The language used in \Cref{sec:one-pass,sec:multi} is slightly different for readability. We work with permutations induced by permutation vectors, which we call simple permutations (see \Cref{def:simple-perm}) directly.

\subsubsection{Permutation Hiding for Multi-Pass Algorithms} 

The last step of our approach is to go from single-pass algorithms to multi-pass algorithms using the strong guarantee of permutation hiding generators and the power of back-and-forth communication 
between the players, but not the referee, in the $\IndexHPH$ problem. While at a technical level this step still requires addressing several new challenges, at a conceptual level, it more or less mimics that of~\cite{ChenKPSSY21} without any 
particularly novel ideas.  

The goal is now to construct a permutation hiding generator for $p$-pass algorithms, simply denoted by $\GG_p(\cdot)$, from a generator for $(p-1)$-pass algorithms denoted similarly by $\GG_{p-1}(\cdot)$. The main idea is still based on a 
reduction from $\IndexHPH$ for inputs $(\Sigma,L,\cM,\Gamma)$. A key point in our construction of $G^{\Sigma}$ and $H$ for single-pass algorithms is that the edges added by the referee, but certainly not the players, to $H$
form different \emph{matchings} $M_1,M_2,\ldots$ between \emph{disjoint} sets of vertices of $H$. In particular, if we see $H$ as a layered graph, some of these layers are constructed based on $G^{\Sigma}$ and the edges the between
remaining layers are perfect matchings. 

To create our $p$-pass generator, we start with creating a graph $H$ from $(\Sigma,L,\cM,\Gamma)$ as described in the previous part. Then, for every one of the perfect-matching layers $M_j$ described above, we replace those layers 
with a permutation hiding graph $\GG_{p-1}(M_j)$; here, we consider the perfect matching $M_j$ as a permutation over the vertices of the layer. This will increase the number of vertices, and the number of layers, in the entire graph by a constant factor. 
See~\Cref{fig:gen-over} for an illustration. 
\begin{figure}[h!]
	\centering
	\input{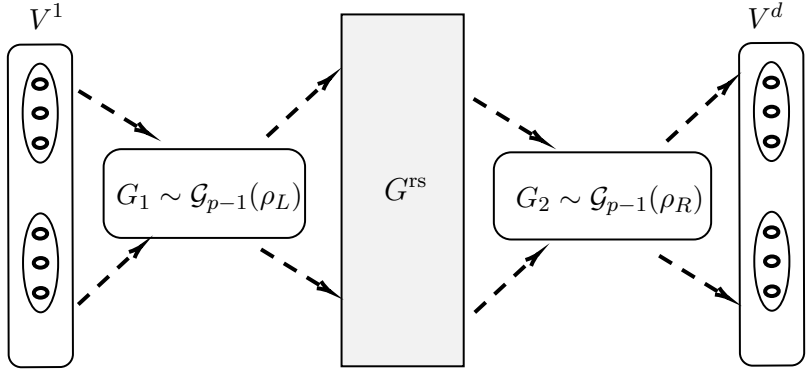}\caption{An illustration a (subgraph of a) $p$-pass generator from $(p-1)$-pass generator.}\label{fig:gen-over}
\end{figure}

We can now claim that this new family of graphs is in fact a $p$-pass generator. At an intuitive level, we can say that the first $(p-1)$ passes of the algorithm cannot figure $(L,\cM,\Gamma)$ as it cannot determine the identity of 
each of the matchings $M_1,M_2,\ldots,$ added to $H$ by the referee. This is by the guarantee of the generator $\GG_{p-1}$ against $(p-1)$ pass algorithms.  Because of this, the problem ``effectively'' become the same 
as the single-pass case again, and we can perform a reduction from $\IndexHPH$ as before. Making this intuition precise requires quite a lot of technical work, but as we stated earlier, it follows a similar 
pattern as that of~\cite{ChenKPSSY21} (plus certain simplifications using a similar argument in~\cite{AssadiR20}), and we postpone their details to~\Cref{sec:multi}.

\subsection{Proof Outline of~\Cref{thm:main-lb}: A Multi-Pass Lower Bound for Matchings}\label{sec:proof-main-lb}

The proof is based on a standard reduction, e.g., in~\cite{AssadiR20}, from finding vertex-disjoint paths to bipartite matching (a similar reduction also appeared in~\cite{ChenKPSSY21}). 
Given the simplicity of this proof, we more or less provide its full details here for completeness. 

Consider any permutation graph $G=(V,E)$ for some $\sigma \in S_m$. A simple observation is that for any pairs of vertices $s_1,s_2 \in \first_{[m]}(G)$, their corresponding paths to $\sigma(s_1),\sigma(s_2) \in \last_{[m]}(G)$
should be vertex disjoint. Suppose not, and let $v \in V$ be one intersecting vertex; then, $s_1 \rightsquigarrow v \rightsquigarrow \sigma(s_1)$ and $s_1 \rightsquigarrow v \rightsquigarrow \sigma(s_2)$ should 
be some paths in $G$, violating the permutation graph property of $G$. We use this observation crucially in our proof. 

Fix a permutation hiding generator $\GG := \GG_{m,n,s,p,\delta}$ for any integers $m,p \geq 1$ and parameters $n,s,\delta$ as in~\Cref{thm:perm-hiding}. For any graph $G$ in the support of any of the distributions output by $\GG$, 
let $S = \first_{[m/2]}(G)$ and $T = \last_{[m/2]}(G)$. Now, 
\begin{itemize}
	\item By taking $\sigma_{=}$ to be the identity permutation, we can ensure that in any $G \sim \GG(\sigma_{=})$, there are $m/2$ vertex-disjoint paths from $S$ to $T$ in $G$; 
	\item But, by taking $\sigma_{\times}$ to be the ``cross identity'', namely, a one mapping $[m/2]$ to $[m/2+1:m]$ identically and vice versa, we can ensure that in $G \sim \GG(\sigma_{\times})$, there are no vertex-disjoint paths from $S$ to $T$ in $G$. 
\end{itemize}

The next step is to turn this vertex-disjoint path problem into a bipartite matching instance. Roughly speaking, this is done by turning these paths into augmenting paths for a canonical matching in the bipartite matching instance. 

Consider the following function $\bipartite{G}$ that given a permutation graph $G$ in the support of any distribution in $\GG$ creates the following bipartite graph: 
\begin{enumerate}[label=$(\roman*)$]
	\item Copy the vertices of $G$ twice to obtain two sets $V^{\Left}$ and $V^{\Right}$. Additionally, add two new sets $S_0$ and $T_0$ to the graph. 
	\item For every vertex $v \in G$, connect $v^{\Left} \in V^{\Left}$ to $v^{\Right} \in V^{\Right}$ ( refer to these edges as a matching $M$). 
	\item For every edge $(u,v) \in G$, connect $u^{\Left} \in V^{\Left}$ to $v^{\Right} \in V^{\Right}$. Additionally, for every vertex $s_i \in S$ of $G$, connect $s_i \in S_0$ to $s_i^{\Right} \in V^{\Right}$ and
	for every vertex $t_i \in T$ of $G$, connect $t_i \in T_0$ to $t_i^{\Left} \in V^{\Left}$. 
\end{enumerate}

It is not hard to see that in this graph, any augmenting path for the matching $M$ has to start from $S_0$ and end in $T_0$. In addition, the structure of the graph, plus the alternating nature of an augmenting path, 
forces the path to basically follow the copies of edges in $G$ from $S$ to $T$. This will turn imply that: 
\begin{itemize}
	\item When $G \sim \GG(\sigma_{=})$, the matching $M$ in $\bipartite{G}$ has $m/2$ vertex-disjoint augmenting paths, and thus can be augmented to a perfect matching of size $n+m/2$; 
	\item On the other hand, when $G \sim \GG(\sigma_{\times})$, the matching $M$ in $\bipartite{G}$ has no augmenting paths and is thus a maximum matching of size $n$. 
\end{itemize} 

Given that streaming algorithms running on $G$ can generate $\bipartite{G}$ ``on the fly'', this is sufficient to prove that any streaming algorithm that can determine whether the input graph has a perfect matching of size $n+m/2$ or its largest matching is of size $n$, 
would also distinguish between the distributions $\GG(\sigma_{=})$ and $\GG(\sigma_{\times})$. But, by the indistinguishability of these distributions in~\Cref{thm:perm-hiding}, we 
obtain that no $p$-pass $s$-space algorithm can solve the underlying matching problem. We can now instantiate the parameters of~\Cref{thm:perm-hiding} as follows. 
\begin{itemize}[itemsep=10pt]
	\item The parameter $\eps$ of~\Cref{thm:main-lb} is obtained via: 
	\[
	(1-\eps) \cdot \paren{n+\frac{m}{2}} = n,
	\]
	which implies that $\eps = \Theta(m/n)$. 
	\item The number of vertices of $\bipartite{G}$ is $n+m/2$ on each side and thus $\Theta(n)$ in general. 
	
	\item The number of passes of the algorithm needs to be at least 
	\[
	p = \Omega\Paren{\frac{\log{({(n+m/2)/m)}}}{\log({1}/{\alpha \cdot \beta^2})}}= \Omega\Paren{\frac{\log{(1/\eps)}}{\log{(1/\alpha\beta)}}}. 
	\]
	\item The bound on the space $s$ of the algorithms needs to be
	\[
	s = o(m^{1+\beta/2}) = o((\alpha \cdot \beta^2)^{p \cdot (1+\beta/2)} \cdot (n+m/2)^{1+\beta/2}) = o(\eps^{2} \cdot (n+m/2)^{1+\beta/2}).
	\]
\end{itemize}
This then concludes the proof of~\Cref{thm:main-lb}.

\section{Preliminaries}


\subsection{Notation}

\subsubsection*{Graphs}
For any graph $G=(V,E)$, we use $n:=\card{V}$ to denote the number of vertices. For an edge $e=(u,v)$, we use $V(e)$ to denote the vertices incident on $e$. 
When $G$ is a directed graph, for any vertices $s,t \in V$, we write $s \rightarrow t$ to mean there is an edge from $s$ to $t$ and $s \rightsquigarrow t$ to mean there is a path from $s$ to $t$ in $G$. 

We sometimes say that two (or more)  disjoint sets $S$ and $T$ of a graph $G$ with size $m=\card{S}=\card{T}$ are identified by the set $[m]$ to mean that any integer $i \in [m]$ can be 
used to refer to both the $i$-th vertex of $S$ as well as $T$, when it is clear from the context (e.g., we say that for any $i \in [m]$, connect vertex $i \in S$ to vertex $i \in T$). 

\subsubsection*{Tuples and matrices} 
  For any $m$-tuple  $X: =(x_1,\ldots,x_m)$, and any integer $i \in [m]$, we define 
  \[
  X_{<i} := (x_1,\ldots,x_{i-1}) \qquad \text{and} \qquad X_{-i} := (x_1,\ldots,x_{i-1},x_{i+1},\ldots,x_m).
  \] 
  We can further define $X_{>i}$ and $X_{\leq i}$ analogously and extend these definitions to vectors as well. 

For a $A \in \IR^{m \times n}$, we use $A_{i,*}$ to denote the $i$-th row of the matrix $A$ and $A_{*,j}$ to denote the $j$-th column. We denote the sub-matrix of $A$ on rows in $S \subseteq [m]$ and $T \subseteq [n]$ 
by $A_{S,T}$, and similarly use $A_{S,*}$ and $A_{*,T}$ to only limited the rows or and columns to $S$ and $T$, respectively.

\subsubsection*{Permutations}
We use $S_m$ to denote the symmetric group of permutations over the set $[m]$.
We use $\sigmaId(m) \in S_m$ to denote the identity permutation on $[m]$. For any permutation matrix $\Sigma \in \paren{S_b}^{t \times r}$, we use $\sigma_{i,j} \in S_b$ to denote the permutation at row $i$ and column $j$ for $i \in [t], j \in [r]$.

\subsection{Bipartite \rs-Graphs}\label{sec:rs}

Let $G=(V,E)$ be an undirected graph, and $M \subseteq E$ be a matching in $G$. We say that $M$ is an \textbf{induced matching} iff the subgraph of $G$ induced on the vertices of $M$ is the matching $M$ itself; in other words, 
there are no other edges between the vertices of this matching. 

\begin{Definition}[{\rs Graphs~\cite{RuzsaS78}}]\label{def:rs}
	For any integers $r,t \geq 1$, a \emph{bipartite} graph $\Grs=(\Lrs \cup \Rrs,\Ers)$ is called a bipartite \textnormal{\textbf{$(r,t)$-\rs graph}} (RS graph for short) iff its edge-set $E$ can be partitioned into $t$ \emph{induced} matchings $\Mrs{1},\ldots, \Mrs{t}$, each of size $r$.
\end{Definition}

\begin{figure}[h!]
	\centering
	\input{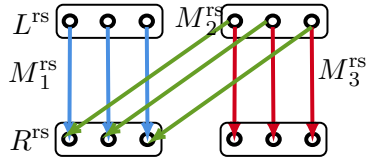}\caption{An illustration of simple RS-Graph construction.}\label{fig:RSbasic}
\end{figure}

It is easy to construct bipartite $(r,t)$-RS graphs with $n$ vertices on each side and parameter $t=(n/r)^2$ for any $r \geq 1$ (see \Cref{fig:RSbasic}). More interestingly, and quite surprisingly, 
one can create much denser RS graphs as well; for instance, the original construction of these graphs due to Ruzsa and Szemer\'edi~\cite{RuzsaS78} create graphs with $r=n/2^{\Theta(\sqrt{\log{n}})}$ and $t=\Omega(n)$. 
Yet another construction is that of~\cite{FischerLNRRS02,GoelKK12} that shows that already for $r < n/2$, we can have $t=n^{\Omega(1/\log\log{n})}$. 
See~\cite{BirkLM93,FischerLNRRS02,Alon02,TaoV06,AlonS06,AlonMS12,GoelKK12,FoxHS15,AssadiB19,KapralovKTY21} for various other constructions and applications of RS graphs. At the same time, 
proving upper bounds on the density of RS graphs turned out to be a  notoriously difficult  question (see, e.g.~\cite{Gowers01,FoxHS15,ConlonlF13}) and the best known upper bounds, even 
for $r=\Omega(n)$, only imply that $t < n/2^{O(\log^*{n})}$~\cite{Fox11,FoxHS15}. 

A line of work initiated by Goel, Kapralov, and Khanna~\cite{GoelKK12} have used different constructions of RS graphs to prove lower bounds for graph streaming algorithms~\cite{GoelKK12,Kapralov13,Konrad15,AssadiKLY16,AssadiKL17,CormodeDK19,AssadiR20,Kapralov21,AssadiB21a,ChenKPSSY21,KonradN21,Assadi22} (very recently, they have also been used in~\cite{AssadiBKL23} for 
providing graph streaming \emph{algorithms} for the matching problem).

\paragraph{Notation for RS-Graphs.}
The number of vertices in one partition of the bipartite graph $\Grs = (\Lrs \cup \Rrs, \Ers)$ is denoted by $\nrs$ (the total number of vertices is $2\nrs$). The two sets $\Lrs$ and $\Rrs$ are both identified by the set $[\nrs]$. For each matching $\Mrs{i}$ for $i \in [t]$, the edges in the matching are identified by the set $[r]$.  We also use $e$ to iterate over the edges in $\Ers$, or point to any arbitrary edge in $\Ers$. Given any $e \in \Ers$, we use $\leftRS{e} \in [\nrs]$ and $\rightRS{e} \in [\nrs]$ to denote the vertices in $\Lrs$ and $\Rrs$ that edge $e$ is incident on, respectively.

For any $(r,t)$-\rs graph $\Grs$, we use $\alpha, \beta \in [0,1]$ to denote the following parameters.
\begin{align}
	\alpha := \frac{r}{\nrs} \quad \beta := \frac{\log{(t)}}{\log{(\nrs)}}.\label{eq:rs-parameters}
\end{align}
That is, $\Grs$ has $\paren{\nrs}^{\beta}$ induced matchings of size $\alpha\cdot \nrs$ each. 

\subsection{Sorting Networks with Large Comparators}\label{sec:comprators} 

Following~\cite{ChenKPSSY21}, we  use sorting networks (see, e.g.~\cite[Section 5.3.4]{knuth1997art}) in our constructions. The key difference is that we need sorting networks that work with larger comparators, typically called \emph{$b$-sorters} (i.e., 
each sorter can sort $b > 2$ wires in one step as a primitive). We refer the reader to~\cite{ParkerP89,chvatal1992lecture} for more information on sorting networks with $b$-sorters (see \Cref{fig:large-sorters}). 

\begin{figure}[h!]
	\centering
	\input{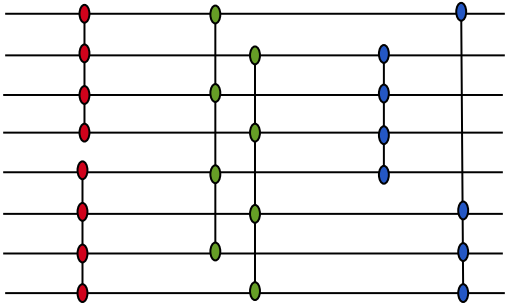}\caption{A sorting network with $b = 4$ size sorters, $m = 8$ wires, and depth $d= 3$.}
\label{fig:large-sorters}
\end{figure}

We shall 
use the following result from~\cite{chvatal1992lecture} that is a generalization of the celebrated \emph{AKS sorter} of~\cite{AjtaiKS83} with  $O(\log{n})$ depth to networks with $b$-sorters with $O(\log_b{n})$ depth instead\footnote{Although we note that given the range of parameters
	used in our paper, namely, $b=n^{\Omega(1)}$, we can alternatively use the much simpler $O(\log^2_b\!{(n)})$-depth construction of~\cite{ParkerP89} as well to the same effect; see~\Cref{app:sorting-stuff}.}. We present this result 
directly in the language we shall use later in this paper, but the equivalence in terms of sorting networks is straightforward. 

\begin{Definition}\label{def:simple-perm}
	Given an equipartition $\cP$ of the set $[m]$ into $r$ groups $P_1, P_2, \ldots, P_r$ of size $b = m/r$ each, a permutation $\sigma \in S_m$ is said to be \textbf{simple} on partition $\cP$ if for each group $P_j$ with $j \in [r]$, for any element $i \in [m]$ belonging to group $P_j$, we have that $\sigma(i) \in P_j$ also. 
\end{Definition}

Informally, a simple permutation permutes elements only inside the groups in the partition $\cP$.

\begin{proposition}[\!\!\cite{chvatal1992lecture}; see also~\cite{AjtaiKS83}]\label{prop:AKSsorting}
	There exists an absolute constant $\caks > 0$ such that the following is true. For every pair of integers $r,b \geq 1$ and $m = r \cdot b$, there exists 
	\[
	\daks = \daks(r,b) = \caks \cdot \log_b(m)
	\]
	fixed equipartitions of $[m]$ into $\cP_1,\ldots,\cP_{\daks}$, each one consisting of $r$ sets of size $b$, with the following property. 
	Given any permutation $\sigma \in S_m$, there are $\daks$ permutations $\gamma_1,\ldots,\gamma_{\daks}$ where for every $i \in [\daks]$, $\gamma_i$ is simple on partition $\cP_i$ so that we have
	$
	\sigma = \gamma_1 \circ \cdots \circ \gamma_{\daks}. 
	$
\end{proposition}

To interpret this result in terms of sorting networks on $m$ wires, think of each $\cP_i$ for $i \in [\daks]$ as one layer of the sorting network and each $b$-sorter in this layer as one component of $\cP_i$ (see \Cref{fig:large-sorters}). We have provided a proof of a weaker version of \Cref{prop:AKSsorting} when $b = 2$ with a larger depth $\daks = O(\log ^2(m))$ for intuition in \Cref{app:sorting-stuff}. The proof of \Cref{prop:AKSsorting} with $\daks = \caks \log_b(m)$ and $b > 2$ is quite involved, and the details can be found in \cite{chvatal1992lecture}. 

\subsubsection*{Simple Permutations and Permutation Vectors}

In \Cref{sec:lb-multi-hph}, we work with permutation vectors from $\paren{S_b}^{r/2}$ with $r = 2m/b$ and in \Cref{sec:one-pass,sec:multi}, we predominantly work with simple permutations. We will establish the bijection between the two for any fixed partition here. 
Let equipartition $\lexpart$ of $[m]$ be the partition that splits the elements lexicographically into groups of size $b$. That is, for each $i \in [m/b]$:
\begin{align*}
	\lexpart_i = \set{j \mid (i-1) \cdot b+1 \leq j \leq i \cdot b}.
\end{align*} Let $\Ssim_m \subset S_m$ be the set of all permutations which are simple on partition $\lexpart$.

\begin{observation}\label{clm:simple-perm-to-perm-vec}
	There is a bijection $\tovec: \Ssim_m \rightarrow \paren{S_b}^{m/b}$. 
\end{observation}

\begin{proof}
	We know that partition $\lexpart$ lexicographically groups the elements of $[m]$ into groups of size $b$. For any $\rho \in \Ssim_m$, we define $\tovec(\rho)$ as follows. For $i \in [m/b]$, let $\gamma_i \in S_b$ be,
	\begin{align*}
		\gamma_i(j) = \rho((i-1)b+j) - (i-1)b
	\end{align*}
	for each $j \in [b]$. As $\rho((i-1)b+j)$ belongs to the same partition as $(i-1)b+j$ for $j \in [b]$, $\gamma_i$ is a permutation in $S_b$. Then, $\tovec(\rho) = (\gamma_1, \gamma_2, \ldots, \gamma_{m/b})$.
\end{proof}

Note that \Cref{clm:simple-perm-to-perm-vec} is applicable for any set of simple permutations for a fixed partition, but we will work mainly with $\lexpart$. 

\paragraph{Notation.} For any permutation $\rho \in \Ssim_m$, we use $\tovec(\rho)$ to denote the corresponding element from $\paren{S_b}^{m/b}$. For any permutation vector $\Gamma = (\gamma_1, \gamma_2, \ldots, \gamma_{m/b}) \in \paren{S_b}^{m/b}$, we use $\joinit(\Gamma)$ to denote the corresponding element of $\Ssim_m$, as the tuple is combined into a single larger permutation. (This is just the inverse of bijection $\tovec$ from \Cref{clm:simple-perm-to-perm-vec}.)

\subsection{Streaming Algorithms}\label{sec:streaming} 

For the purpose of our lower bounds, we shall work with a {more powerful} model than what is typically considered when designing streaming algorithms (this is the common approach when proving streaming
lower bounds; see, e.g.~\cite{GuhaM08,LiNW14,BravermanGW20}). In particular, we shall define streaming algorithms 
similar to branching programs as follows and then point out the subtle differences with what one typically expect of a streaming algorithm.  

\begin{Definition}[\textbf{Streaming algorithms}]\label{def:streaming}
	For any integers $m,p,s \geq 1$, we define a $p$-pass $s$-space streaming algorithm $A$ working on a length-$m$ stream $\sigma=(\sigma_1,\ldots,\sigma_m)$ with entries from a universe $U$
	as follows: 
	\begin{enumerate}[label=$(\roman*)$]
		\item The algorithm $A$ has access to a read-only tape of uniform random bits $r$ from a finite, but arbitrarily large range $R$ without having to pay for the cost of storing these bits.
		\item There is a function $f_A : \set{0,1}^{s} \times U \times R  \rightarrow \set{0,1}^s$ that updates the state of the algorithm as follows. The algorithm starts with the state $S:=0^{s}$, 
		and for $i \in [m]$, whenever $A$ reads $\sigma_i$ in the stream during its $p$ passes, it updates its current state $S$ to
		$S \leftarrow f_A(S,\sigma_i,r)$ (the algorithm is computationally unbounded when computing its next state). 
		\item At the end of the last pass, the algorithm outputs the answer as a function of its state $S$ and random bits $r$. 
	\end{enumerate}
	(We note that this model is non-uniform and is defined for each choice of $m,p,s$ individually.)
\end{Definition}

Let us  point out two main differences with what one may expect of streaming algorithms.  Firstly, we allow our streaming algorithms to do an unbounded amount of work using an unbounded amount of space \emph{between} the
arrival of each stream element; we only bound the space in transition between two elements. Secondly, we do not charge the streaming algorithms
for storing random bits. 

Clearly, any lower bound proven for streaming algorithms in~\Cref{def:streaming} will hold also for more restrictive (and ``algorithmic'') definitions of streaming algorithms. 
We shall note that however almost all streaming lower bounds we are aware of directly 
work with this definition and thus we claim no strengthening in proving our lower bounds under this definition; rather, we merely use this formalism to carry out the reductions in our arguments formally. 

For any $p$-pass algorithm $A$, $q \in [p]$, and input distribution $\dist$, let $\mem{q}{A}(\dist)$ denote the memory state of $A$ after $q$ passes plus  the content of its random tape on an input sampled from $\dist$.
Note that $\mem{q}{A}(\dist)$ is a random variable depending on randomness of $A$ as well as $\dist$. 

\begin{Definition}[$\delta$-Indistinguishable Distributions]
	For any $\delta \in [0,1]$, the two distributions $\mu,\nu$ are said to be $\delta$-indistinguishable for $p$-pass $s$-space streaming algorithms if for every such algorithm $A$, 
	\[
	\tvd{\mem{p}{A}(\mu)}{\mem{p}{A}(\nu)} \leq \delta.
	\]
\end{Definition}

Finally, we have the following standard hybrid argument for multi-pass streaming algorithms (see, e.g.~\cite{ChenKPSSY21,AssadiN21}). We present its short proof for completeness.

\begin{proposition}[c.f.~\cite{ChenKPSSY21}]\label{prop:hybrid-arg}
	Let $\ell$ be a positive integer and $\delta \in \IR^{\ell}_{\geq 0}$ denote $\ell$ parameters. Let $(\mu_1, \nu_1), (\mu_2, \nu_2) , \ldots, (\mu_{\ell}, \nu_{\ell})$ be $\ell$ pairs of distributions, such that
	for every $i \in [\ell]$, $\mu_i$ and $\nu_i$ are $\delta_i$-indistinguishable for $p$-pass $s$-space streaming algorithms. Then, $\mu := (\mu_1,\ldots,\mu_\ell)$ and $\nu := (\nu_1,\ldots,\nu_{\ell})$ are $\norm{\delta}_1$-indistinguishable 
	for $p$-pass $s$-space streaming. 
\end{proposition} 
\begin{proof}
	For any $i \in [\ell]$, define the hybrid distribution: 
	\[
	h_i := (\nu_1,\ldots,\nu_{i},\mu_{i+1},\ldots,\mu_{\ell}). 
	\]
	This way, $h_0 = \mu$ and $h_{\ell} = \nu$. Consider any $p$-pass $s$-space streaming algorithm $A$. We prove that for every $i \in [\ell]$, 
	\[
	\tvd{\mem{p}{A}(h_{i-1})}{\mem{p}{A}(h_{i})} \leq \delta_i,
	\]
	by turning $A$ into a $p$-pass $s$-space streaming algorithm $B$ for distinguishing between $\mu_i$ and $\nu_i$. Algorithm $B$, given a stream $\sigma$ from either $\mu_i$ or $\nu_i$, is defined as follows: 
	\begin{enumerate}[label=$(\roman*)$]
		\item Sample the inputs $\sigma_1,\ldots,\sigma_{i-1} \sim \nu_1,\ldots,\nu_{i-1}$ and $\sigma_{i+1},\ldots,\sigma_{\ell} \sim \mu_{i+1},\ldots,\mu_{\ell}$. This sampling is free of charge for the algorithm $B$ by~\Cref{def:streaming}. 
		\item Run $A$ on the input $(\sigma_1,\ldots,\sigma_{i-1},\sigma,\sigma_{i+1},\ldots,\sigma_{\ell})$ in $p$ passes and $s$ space. 
	\end{enumerate}
	We thus have, 
	\[
	\tvd{\mem{p}{A}(h_{i-1})}{\mem{p}{A}(h_i)} = \tvd{\mem{p}{B}(\mu_i)}{\mem{p}{B}(\nu_i)} \leq \delta_i,
	\]
	by the $\eps_i$-indistinguishability of $\mu_i$ and $\nu_i$ for $p$-pass $s$-space streaming algorithm $B$.  
	The final result now follows from triangle inequality. 
\end{proof}


\section{The Multi Hidden Permutation Hypermatching Problem}\label{sec:lb-multi-hph}

We prove~\Cref{thm:mph} on the communication cost of $\IndexHPH$  in this section. For the convenience of the reader, we restate the definition of this communication game
and our theorem below.

\begin{Problem*}[Restatement of \Cref{prob:mph}]
	For any integers $r,t,b,k \geq 1$, $\IndexHPH_{r,t,b,k}$ is a distributional $(k+1)$-communication game, consisting of $k$ players plus a referee, defined as: 	
	\begin{enumerate}[label=$(\roman*)$]
		\item Let $\GammaY, \GammaN \in \paren{S_b}^{r/2}$ be two arbitrary tuples of permutations known to all $(k+1)$ parties. We refer to $\GammaY$ and $\GammaN$ as the \textbf{target tuples}. 
		\item We have $k$ players $\QP{1},\ldots,\QP{k}$ and for all $i \in [k]$ the $i$-th player is given a permutation matrix $\SigmaP{i} \in (S_b)^{t \times r}$ chosen independently and uniformly at random. 
	
		\item The referee receives $k$ indices $L := (\ell_{1},\ldots,\ell_{k})$ each picked uniformly and independently from $[t]$ and a hypermatching $\cM \subseteq [r]^k$ with $r/2$ hyperedges picked uniformly. 	
			
		 Additionally, let the permutation vector $\Gammastar = \paren{\gammastar_1, \gammastar_2, \ldots, \gammastar_{r/2}} \in \paren{S_b}^{r/2}$ be defined as,
		\[
			\forall a \in [r/2] \hspace{2mm} \gammastar_a := \sigmaP{1}_{\ell_1, \cM_{a,1}} \circ \cdots \circ \sigmaP{k}_{\ell_k,\cM_{a,k}},
		\]
		where $\cM_{a,i}$ for $a \in [r/2]$ and $i \in [k]$ refers to the $i$-th vertex of the $a$-th hyperedge in $\cM$. 
		
		 The referee is also given another permutation vector $\Gamma = (\gamma_1, \gamma_2, \ldots, \gamma_{r/2}) $ sampled from $\paren{S_b}^{r/2}$ conditioned on  $\Gammastar \conc \Gamma$  being equal to either $\GammaY$ or  $\GammaN$. 
	
	\end{enumerate}
	\noindent
	The players $\QP{1},\ldots,\QP{k}$ can communicate with each other by writing on a shared board visible to all parties with possible back and forth and in no fixed order (the players' messages are functions of their inputs and the board). 
	At the end of the players' communication, the referee can use all these messages plus its input and outputs whether $\Gammastar \conc \Gamma$ is $\GammaY$ or $\GammaN$. 
\end{Problem*}

\begin{theorem*}[Restatement of \Cref{thm:mph}]
	For any $t \geq 1$, $b \geq 2$, and sufficiently large $r, k \geq 1$, any communication protocol for $\IndexHPH_{r,t,b,k}$, for any choice of target tuples, with at most $s$ bits of total 
	communication for $s$ satisfying
	\[
	k \cdot \log{(r \cdot t)} \leq s \leq 10^{-3} \cdot (r \cdot t)
	\]
	can only succeed with probability at most 
	\[
		\frac{1}{2} + r \cdot O\!\paren{\frac{s}{r \cdot t}}^{k/32}.
	\]
\end{theorem*}

When going through the proof, we shall condition on several probabilistic events and properties at different points, and from thereon continue all 
our analysis conditioned on them. We will mark each of these steps clearly in the corresponding subsection as a ``conditioning step'' (typically as the conclusion of the subsection), in 
order to help the reader keep track of them. 





\renewcommand{\ProtP}[1]{\Prot^{(#1)}}

\subsection{Part One: Setup and the Basic Problem}\label{sec:part-one-setup} 

We begin our proof of~\Cref{thm:mph} in this subsection by defining the ``basic'' problem we will need to focus on to prove this result. 
This will set up the stage for the main parts of the proof in the subsequent subsections. We start with some notation. 

\paragraph{Notation.} We fix $t \geq 1$, $b \geq 2$, and sufficiently large choice of integers $r,k \geq 1$ and consider the $\IndexHPH_{r,t,b,k}$ problem (henceforth, denoted simply by $\IndexHPH$) with any two arbitrarily target tuples $\GammaY$ and $\GammaN$. 
We further define
\[
\Sigma := (\SigmaP{1},\ldots,\SigmaP{k}),
\]
to denote the input to all players $\QP{1},\ldots,\QP{k}$. 

From now on, fix a \textbf{deterministic} protocol $\prot$ for $\IndexHPH$ and let $\Prot=(\Prot_1,\ldots,\Prot_k)$ denote the transcript 
of the protocol, where $\Prot_i$ is the messages communicated the player $\QP{i}$ for $i \in [k]$. 
 As per the theorem statement, we assume that $\prot$ communicates $s$ bits in total. 

Finally, throughout this section, we use sans-serif letters to denote the random variables for corresponding objects, e.g., $\rSigma$ and $\rProt$ for the input $\Sigma$ and messages $\Prot$ of players, and $\rL,\rcM,\rGamma$ for the input of referee $(L,\cM,\Gamma)$. 
Moreover, we may use random variables and their distributions interchangeably when the meaning is clear from the context.

\subsubsection{Removing the Role of $\GammaY,\GammaN$, and $\Gamma$} 

We are now ready to start the proof. We can interpret the goal of the referee in $\IndexHPH$  as follows: Given $(\Prot,L,\cM)$, the final input of the referee is the tuple of permutations $\Gamma$ chosen uniformly at random from one of the following 
two distributions: 
\[
	({\rGammastar}^{-1} \circ \GammaY \mid \Prot, L, \cM) \qquad \text{or} \qquad ({\rGammastar}^{-1} \circ \GammaN \mid \Prot, L,\cM); 
\]
here, and throughout, ${\Gammastar}^{-1}$ is interpreted as taking the inverse of each permutation in the tuple. 

Thus, given one sample $\Gamma$ from a uniform mixture of the above two distributions, the referee needs to determine which distribution $\Gamma$ was sampled from. This, together with~\Cref{fact:tvd-sample} (on success probability of distinguishing distributions with one sample), implies that
the probability of success of the referee is equal to: 
\begin{align}
	\Pr_{\Prot,L,\cM,\Gamma}\paren{\text{$\pi$ succeeds}} = \frac12 + \frac12 \cdot \!\!\Exp_{\Prot,L,\cM} \tvd{({\rGammastar}^{-1} \circ \GammaY \mid \Prot, L, \cM)}{({\rGammastar}^{-1} \circ \GammaN \mid \Prot, L,\cM)}. \label{eq:pi-success-tvd1}
\end{align}

Our plan is to bound the RHS of~\Cref{eq:pi-success-tvd1}. To do so, we are going to do this indirectly by proving that both distributions considered in this equation are quite close to the uniform distribution. 
Another simple step also allows us to entirely ignore the choice of $\GammaY$ and $\GammaN$ and simply consider the distribution of $\Gammastar$ itself. These parts are captured in the following claim. 

\begin{claim}\label{clm:get-rid-of-GammaYN}
	For any choice of $\Prot,L,\cM$, 
	\[
		\tvd{({\rGammastar}^{-1} \circ \GammaY \mid \Prot, L, \cM)}{({\rGammastar}^{-1} \circ \GammaN \mid \Prot, L,\cM)} \leq 2 \cdot \tvd{(\rGammastar \mid \Prot, L, \cM)}{\cU_{S_b^{r/2}}}, 
	\]
	where $\cU_{S_b^{r/2}}$ is the uniform distribution over tuples in $S_b^{r/2}$. 
\end{claim}
\begin{proof}
	By the triangle inequality, 
	\[
		\text{LHS of the claim} \leq	\tvd{({\rGammastar}^{-1} \circ \GammaY \mid \Prot, L, \cM)}{\cU_{S_b^{r/2}}} \, + \, \tvd{\cU_{S_b^{r/2}}}{({\rGammastar}^{-1} \circ \GammaN \mid \Prot, L,\cM)}.
	\]
	 For the first term above, we have, 
	\begin{align*}
		\tvd{({\rGammastar}^{-1} \circ \GammaY \mid \Prot, L, \cM)}{\cU_{S_b^{r/2}}} &= \frac12 \cdot \sum_{\Gamma \in S_b^{r/2}} \card{\Pr\paren{\rGammastar^{-1} \circ \GammaY = \Gamma \mid \Prot,L,\cM} - \paren{\frac{1}{b!}}^{r/2}} 
		\tag{by the definition of TVD in~\Cref{eq:tvd}} \\
		&=  \frac12 \cdot  \sum_{\Gamma \in S_b^{r/2}} \card{\Pr\paren{\rGammastar  = (\Gamma \circ \GammaY^{-1})^{-1} \mid \Prot,L,\cM} - \paren{\frac{1}{b!}}^{r/2}} \tag{each tuple of permutations has a unique inverse by taking  inverse of each
		of its permutations} \\
		&= \frac12 \cdot \sum_{\Gammastar \in S_b^{r/2}} \card{\Pr\paren{\rGammastar  = \Gammastar  \mid \Prot,L,\cM} - \paren{\frac{1}{b!}}^{r/2}} \tag{as there is a one-to-one mapping 
		between $S_b^{r/2}$ and  $(\Gamma \circ \GammaY^{-1})^{-1}$ when $\Gamma$ ranges over $S_b^{r/2}$} \\
		&= \tvd{({\rGammastar} \mid \Prot, L, \cM)}{\cU_{S_b^{r/2}}} \tag{again, by the definition of TVD in~\Cref{eq:tvd}}.
	\end{align*}
	Applying the same  calculation (by replacing $\GammaY$ with $\GammaN$) to the second term above concludes the proof of this claim. 
\end{proof}
Plugging in the bounds we obtained in~\Cref{clm:get-rid-of-GammaYN} in~\Cref{eq:pi-success-tvd1}, we obtain that 
\begin{align}
	\Pr_{\Prot,L,\cM,\Gamma}\paren{\text{$\pi$ succeeds}} \leq \frac12 +  \Exp_{\Prot,L,\cM} \tvd{(\rGammastar \mid \Prot, L, \cM)}{\cU_{S_b^{r/2}}}. \label{eq:pi-success-tvd2}
\end{align}
The rest of the proof is dedicated to bounding the RHS of~\Cref{eq:pi-success-tvd2}. 
Notice this at this point, we have entirely removed the role of $\GammaY, \GammaN$, and $\Gamma$. Thus, from now on, we can solely focus on the power of players in $\IndexHPH$ in changing the distribution of $\Gammastar$ from 
its original distribution which is uniform over $(S_b)^{r/2}$. 

A remark about the importance of $\Gamma$ is in order. All the inputs in \Cref{prob:mph}, barring $\Gamma$ are chosen uniformly at random and independently from their support sets. This is crucial to our proof, as we will see in the later parts of this section. Without $\Gamma$, (say if we set $\Gammastar = \GammaY$ or $\GammaN$) the original distribution of $\Gammastar$ would not be uniform, and we cannot bound the RHS of \Cref{eq:pi-success-tvd2}.

\subsubsection{From the Hypermatching to a Single Hyperedge} 

We now further simplify our task in proving an upper bound on the RHS of~\Cref{eq:pi-success-tvd2}. Recall that $\Gammastar = (\gammastar_1,\ldots,\gammastar_{r/2})$. By the (weak) chain-rule property of total variation distance in~\Cref{fact:tvd-chain-rule}, 
\begin{align}
	\Exp_{\Prot,L,\cM} \tvd{(\rGammastar \mid \Prot, L, \cM)}{\cU_{S_b^{r/2}}} \leq \sum_{a=1}^{r/2} \,\Exp_{\Prot,L,\cM} \, \Exp_{\Gammastar_{<a} \mid \Prot,L,\cM} \tvd{(\rgammastar_a \mid \Prot,L,\cM,\Gammastar_{<a})}{\cU_{S_b}}. 
	\label{eq:after-tvd-cr}
\end{align}
Our goal is thus to bound each individual term in the RHS of~\Cref{eq:after-tvd-cr}. To continue we need the following definitions. For any choice of $\Prot, L, \cM$, and integers $a \in [r/2]$, $i \in [k]$, and $j \in [t]$, define:  
\begin{alignat*}{2}
	&\cM_a &&:= (\cM_{a,1},\ldots,\cM_{a,k}) \tag{the $a$-th hyperedge in $\cM$} \\
	&\cM_{<a} &&:= (\cM_1,\ldots,\cM_{a-1}) \tag{the first $a-1$ hyperedges in $\cM$} \\ 
	&\SigmaP{i}_{j,\cM_{<a}} &&:= (\sigmaP{i}_{j,\cM_{1,i}},\ldots,\sigmaP{i}_{j , \cM_{a-1,i}}) \tag{the  $a-1$ entries of $\SigmaP{i}_{j,*}$ indexed by $\cM_{<a}$} \\
	&\SigmaP{i}_{*,\cM_{<a}} &&:= (\SigmaP{i}_{1,\cM_{<a}},\ldots,\SigmaP{i}_{t,\cM_{<a}}) \tag{the input of $\QP{i}$ on hyperedges in $\cM_{<a}$} \\ 
	&\Sigma_{*,\cM_{<a}} &&:= (\SigmaP{1}_{\cM_{<a}},\ldots,\SigmaP{k}_{\cM_{<a}}) \tag{the  input of $\QP{1},\ldots,\QP{k}$ on hyperedges in $\cM_{<a}$}. 
\end{alignat*}
Notice that this way, 
\[
	\Gammastar_{<a} = \SigmaP{1}_{\ell_1, \cM_{<a}}  \circ \, \cdots \, \circ \, \SigmaP{k}_{\ell_k, \cM_{<a}}. 
\]
The following  claim allows us to ``over condition'' on $\Sigma_{*,\cM_{<a}}$ instead of $\Gammastar_{<a}$ in the 
RHS of~\Cref{eq:after-tvd-cr}. The purpose of this step is to isolate the dependence on $L$ which is needed for our subsequent proofs (as $\Gammastar_{<a}$ depends on $L$ but $\Sigma_{*,\cM_{<a}}$ does not).

\begin{claim}\label{clm:single-hyperedge-1}
	We have: 
	\[
		\Exp_{\Prot,L,\cM} \,\,\, \Exp_{\Gammastar_{<a} \mid \Prot,L,\cM} \tvd{(\rgammastar_a \mid \Prot,L,\cM,\Gammastar_{<a})}{\cU_{S_b}} \leq
		 \Exp_{\Prot,L,\cM,\Sigma_{*,\cM_{<a}}} \tvd{(\rgammastar_a \mid \Prot,L,\cM,\Sigma_{*,\cM_{<a}})}{\cU_{S_b}}.
	\]
\end{claim}
\begin{proof}
	For any choice of $\Prot,L,\cM$ and $\Gammastar_{<a}$, we have, 
	\begin{alignat*}{2}
		\tvd{(\rgammastar_a \mid \Prot,L,\cM,\Gammastar_{<a})}{\cU_{S_b}} &\leq \Exp_{\Sigma_{*,\cM_{<a}} \mid \Prot,L,\cM,\Gammastar_{<a}} \tvd{(\rgammastar_a \mid \Prot,L,\cM,\Gammastar_{<a}, \Sigma_{*,\cM_{<a}})}{\cU_{S_b}} \tag{by
		the over conditioning property of TVD in~\Cref{fact:tvd-over-conditioning}} \\
		&=  \Exp_{\Sigma_{*,\cM_{<a}} \mid \Prot,L,\cM,\Gammastar_{<a}} \tvd{(\rgammastar_a \mid \Prot,L,\cM,\Sigma_{*,\cM_{<a}})}{\cU_{S_b}} \tag{as $\Sigma_{*,\cM_{<a}}$ and $L$ together deterministically fix $\Gammastar_{<a}$}. 
	\end{alignat*}
	Using this in the LHS of the claim gives us, 
	\begin{align*}
		\Exp_{\Prot,L,\cM} \, \Exp_{\Gammastar_{<a} \mid \Prot,L,\cM} \tvd{(\rgammastar_a \mid \Prot,L,\cM,\Gammastar_{<a})}{\cU_{S_b}} &= 	\Exp_{\Prot,L,\cM,\Gammastar_{<a}} 
		\tvd{(\rgammastar_a \mid \Prot,L,\cM,\Gammastar_{<a})}{\cU_{S_b}}	\\
		&\leq \Exp_{\Prot,L,\cM,\Gammastar_{<a},\Sigma_{*,\cM_{<a}}}\!\!\! \tvd{(\rgammastar_a \mid \Prot,L,\cM,\Sigma_{*,\cM_{<a}})}{\cU_{S_b}} \tag{by the above inequality} \\
		&= \Exp_{\Prot,L,\cM,\Sigma_{*,\cM_{<a}}}\!\!\! \tvd{(\rgammastar_a \mid \Prot,L,\cM,\Sigma_{*,\cM_{<a}})}{\cU_{S_b}}  \tag{again, since $\Sigma_{*,\cM_{<a}}$ and $L$ together deterministically fix $\Gammastar_{<a}$},  
	\end{align*}
	concluding the proof. 
\end{proof}
\noindent
\paragraph{The main distribution.} 
For any choice of $\cM_{<a}$ and $\Sigma_{*,\cM_{<a}}$, we define, 
\begin{ourbox}
 \textbf{Distribution $\dist = \dist_{\cM_{<a},\Sigma_{*,\cM_{<a}}}$:} Joint distribution of variables in $\IndexHPH$ conditioned on this particular choice of $\cM_{<a}$ and $\Sigma_{*,\cM_{<a}}$. 
	
	For simplicity of notation, and to avoid clutter, we denote the variables in this distribution, for $i \in [k]$ and $j \in [t]$ as follows: 
	\begin{alignat*}{2}
		&L &&:= (\ell_1,\ldots,\ell_k) \tag{the same input $L$ of the referee as before} \\
		&R_i &&:= [r] - \set{\cM_{1,i},\ldots,\cM_{a-1,i}} \tag{the set of vertices in layer $i$ not matched by $\cM_{<a}$} \\
		&\cM_a &&:= e := (v_1,\ldots,v_k) \tag{the hyperedge $e=\cM_a$ with vertex $v_i$ chosen from $R_i$} \\
		&\bSigmaP{i}\! &&:= \SigmaP{i}_{*,{R_i}} \tag{the remaining input of player $\QP{i}$ not fixed by $\Sigma_{*,\cM_{<a}}$} \\
		&\bSigma &&:= (\bSigmaP{1},\ldots,\bSigmaP{k}) \tag{the remaining input of all players $\QP{1},\ldots,\QP{k}$} \\
		&\gammastar &&:= \gammastar_a := \bsigmaP{1}_{\ell_1,v_1} \circ \cdots \circ \bsigmaP{k}_{\ell_k,v_k} \tag{the permutation in $S_b$ we are interested in} \\
		&\bProt &&:= (\bProt_{1},\ldots,\bProt_{k}) \tag{the messages of players conditioned on $\cM_{<a}$ and $\Sigma_{*,\cM_{<a}}$}.
	\end{alignat*}
\end{ourbox}
We  list several  useful properties of this distribution in the following claim. 
\begin{claim}\label{clm:single-hyperedge-2}
	For any choice of $\cM_{<a},\Sigma_{*,\cM_{<a}}$, let $\dist := \dist_{\cM_{<a},\Sigma_{*,\cM_{<a}}}$. We have, 
	\begin{enumerate}[label=$(\roman*)$]
		\item\label{item:sh-L} $\dist(L)$ is uniform over $[t]^k$; 
		\item\label{item:sh-Ri} $R_i$ is $\card{R_i} = r-(a-1) \geq r/2$ and we can use $r_a := r-(a-1)$ to denote the size of every $R_i$; 
		\item\label{item:sh-e} $\dist(e=(v_1,\ldots,v_k))$ picks each $v_i \in R_i$ independently and uniformly at random;
		\item\label{item:sh-bsigma} $\dist(\bSigma)$ is uniform distribution over its support $({S_b})^{[t] \times R_1} \times \cdots \times {(S_b)}^{[t] \times R_k}$; 
		\item\label{item:sh-independence} We have the following independence properties: 
		\begin{align*}
			\dist(\rbSigma,\rbProt,\rL,\re) = \dist(\rbSigma,\rbProt) \times \prod_{i=1}^{k} \dist(\rell_i) \times \prod_{i=1}^{k} \dist(\rvv_i). 
		\end{align*}
	\end{enumerate}
\end{claim}
\begin{proof}
	In this proof, all random variables are with respect to the original distribution in $\IndexHPH$ unless explicitly indexed by $\dist$. 

\paragraph{Proof of~\Cref{item:sh-L}.} We have $\dist(\rL) = (\rL \mid \cM_{<a},\Sigma_{*,\cM_{<a}})$. In $\IndexHPH$, $L$ is chosen independent of $\cM,\Sigma$ and is uniform over $[t]^k$, thus $\dist(L)$ is also uniform over $[t]^k$. 

\paragraph{Proof of~\Cref{item:sh-Ri}.} Each hyperedge in $\cM_{<a}$ uses one vertex in each layer of the graph, thus $R_i$ has $a-1$ matched vertices by $\cM_{<a}$ and $(r-(a-1))$ unmatched vertices. Moreover, since 
$a \leq r/2$, we get that $r_a = (r-(a-1)) \geq r/2$ as well. 

\paragraph{Proof of~\Cref{item:sh-e}.} We have $\dist(\re) = (\rcM_a \mid \cM_{<a},\Sigma_{*,\cM_{<a}})$. In $\IndexHPH$, $\rcM_a$ is chosen independent of $\Sigma$ and is uniform over all hyperedges of the $k$-layered graph on $[r]^k$ 
that are not matched by $\cM_{<a}$. Given number of unmatched vertices in each layer is exactly $r-(a-1)$ by the previous part, such a hyperedge can be chosen by picking one vertex uniformly at random from unmatched vertices of 
each level, i.e., from $R_i$ in layer $i$. 

\paragraph{Proof of~\Cref{item:sh-bsigma}.} We have $\dist(\rbSigma) = (\rSigma \mid \cM_{<a},\Sigma_{*,\cM_{<a}})$. In $\IndexHPH$, every permutation in $S_b$ of $\Sigma$ is chosen independently of the rest (and independent of $\cM$). 
Thus, after conditioning on $\Sigma_{*,\cM_{<a}}$,  remaining coordinates, i.e., matrices $\SigmaP{i}_{*,R_i} \in (S_b)^{[t] \times R_i}$ for $i \in [k]$, are chosen independently. 

\paragraph{Proof of~\Cref{item:sh-independence}.} By the chain rule, 
	\[
		(\rProt,\rSigma,\rL,\rcM_{a} \mid \cM_{<a},\Sigma_{*,\cM_{<a}}) = (\rProt,\rSigma \mid \cM_{<a},\Sigma_{*,\cM_{<a}}) \times (\rL,\rcM_{a} \mid  \rProt,\rSigma, \cM_{<a},\Sigma_{*,\cM_{<a}}). 
	\]
	We prove that 
	\[
		(\rL,\rcM_{a} \perp \rProt, \rSigma \mid \cM_{<a},\Sigma_{*,\cM_{<a}}) \quad \equiv \quad \mi{\rL,\rcM_{a}}{\rProt, \rSigma \mid \cM_{<a},\Sigma_{*,\cM_{<a}}} = 0,
	\]
	where the equivalence is by~\itfacts{info-zero}. We have, 
	\begin{align*}
		\mi{\rL,\rcM_{a}}{\rProt, \rSigma \mid \cM_{<a},\rSigma_{*,\cM_{<a}}} &= \mi{\rL,\rcM_{a}}{\rSigma \mid \cM_{<a},\Sigma_{*,\cM_{<a}}} + \mi{\rL,\rcM_{a}}{\rProt \mid \rSigma, \cM_{<a},\Sigma_{*,\cM_{<a}}}
		\tag{by chain rule of mutual information in~\itfacts{chain-rule}} \\
		&= \mi{\rL,\rcM_{a}}{\rSigma \mid \cM_{<a},\Sigma_{*,\cM_{<a}}}
		\tag{as $\rProt$ is fixed by $\rSigma \mid \cM_{<a},\Sigma_{*,\cM_{<a}}$ and thus the second term is zero by~\itfacts{info-zero}} \\
		&=0 \tag{by~\itfacts{info-zero} as $\rSigma$ is chosen independent of $\rL,\rcM$}. 
	\end{align*}
	This implies that for every choice of $\cM_{<a},\Sigma_{*,\cM_{<a}}$, 
	\begin{align*}
		\dist(\rbSigma,\rbProt,\rL,\re) &= (\rProt,\rSigma,\rL,\rcM_{a} \mid \cM_{<a},\Sigma_{*,\cM_{<a}}) = (\rProt,\rSigma \mid \cM_{<a},\Sigma_{*,\rcM_{<a}}) \times (\rL,\rcM_{a} \mid \cM_{<a},\Sigma_{*,\cM_{<a}}) \\
		&= \dist(\rbProt,\rbSigma) \times \dist(\rL,\re) = \dist(\rbProt,\rbSigma) \times \dist(\rL) \times \dist(\re \mid \rL). 
	\end{align*}
	by the definition of variables $(\rbSigma,\rbProt,\rL,\re)$. 
	
	By~\Cref{item:sh-L}, we have $\dist(\rL) = \prod_{i=1}^{k} \dist(\rell_i)$. We also have $\dist(\re \mid \rL) = (\rcM_a \mid \rL, \cM_{<a},\Sigma_{*,\cM_{<a}})$ which is the same as $\dist(e)$ 
	as $\rcM_a$ is chosen independent of $\rL$. Thus, by \Cref{item:sh-e}, we also have $\dist(\re \mid \rL) = \dist(\re) = \prod_{i=1}^{k} \dist(\rvv_i)$, concluding the proof. 
\end{proof}

\noindent
We can now write the RHS of~\Cref{clm:single-hyperedge-1} in terms of the distribution $\dist = \dist_{\cM_{<a},\Sigma_{*,\cM_{<a}}}$ as follows:
\begin{align*}
	\text{RHS of~\Cref{clm:single-hyperedge-1}} &= \Exp_{\cM_{<a},\Sigma_{*,\cM_{<a}}} \Exp_{\Prot, L, \cM_{\geq a} \sim  \dist} \tvd{ \dist(\rgammastar_a \mid \Prot,L,\cM_{\geq a})}{\cU_{S_b}}
	 \tag{by the definition of $\mu$ as conditioning on $\cM_{<a},\Sigma_{*,\cM_{<a}}$} \\
	 &= \Exp_{\cM_{<a},\Sigma_{*,\cM_{<a}}} \Exp_{\Prot, L, e \sim  \dist} \tvd{\dist(\rgammastar \mid \Prot,L,e)}{\cU_{S_b}} 
	 \tag{by~\Cref{clm:single-hyperedge-2}-\Cref{item:sh-independence}, $\gammastar = \gammastar_a$ is only a function of $e=\cM_a$ among $\cM_{\geq a}$} \\
	 &= \Exp_{\cM_{<a},\Sigma_{*,\cM_{<a}}} \, \Exp_{\bProt \sim \dist} \, \Exp_{L,e=(v_1,\ldots,v_k) \sim \dist} \tvd{\dist(\rbsigmaP{1}_{\ell_1,v_1} \circ \cdots \circ \rbsigmaP{k}_{\ell_k,v_k} \mid \Prot)}{\cU_{S_b}}
	 \tag{by~\Cref{clm:single-hyperedge-2}-\Cref{item:sh-independence}, $\bProt \perp L,e$ in $\dist$ and $\bSigma$ is only a function of $\bProt$ after fixing $(\ell_i,v_i)$ for $i \in [k]$}. 
\end{align*}
\noindent
By plugging in these bounds for the RHS of~\Cref{clm:single-hyperedge-1} in 	\Cref{eq:after-tvd-cr} and then in~\Cref{eq:pi-success-tvd2}, 
we get that 
\begin{align}
	\Pr\paren{\text{$\pi$ succeeds}} \leq \frac12 + 
	\frac r2 \cdot \!\!\! \max_{\substack{\cM_{<a},\Sigma_{*,\cM_{<a}}\\ \mu=\mu_{\cM_{<a},\Sigma_{*,\cM_{<a}}}}} 
	\!\!\!\bracket{\Exp_{\bProt \sim \dist} \, \Exp_{\substack{\ell_1,\ldots,\ell_k \\ v_1,\ldots,v_k} \sim \dist} \tvd{\dist(\rbsigmaP{1}_{\ell_1,v_1} \circ \cdots \circ \rbsigmaP{k}_{\ell_k,v_k} \mid \Prot)}{\cU_{S_b}}}.
	\label{eq:basic-problem}
\end{align}
This concludes the basic setup we have for proving the lower bound for $\IndexHPH$. The remaining subsections prove an upper bound for the RHS of~\Cref{eq:basic-problem}, which is the heart of the proof.

\paragraph{Conditioning Step:} For the remainder of this section, 
we  {fix an arbitrary choice of $\cM_{<a},\Sigma_{*,\cM_{<a}}$ and $\mu = \mu_{\cM_{<a},\Sigma_{*,\cM_{<a}}}$ and all random variables are with respect to $\mu$ unless specified otherwise}. 

\begin{problem}\label{rem:basic-problem}
	It is worth explicitly identifying the problem we need to solve when bounding the RHS of~\Cref{eq:basic-problem} for the distribution $\mu$. The problem we are interested in is as follows:  
	\begin{enumerate}[label=$(\roman*)$]
		\item We have $k$ players $\QP{1},\ldots,\QP{k}$ each getting a matrix $\bSigmaP{i} \in (S_b)^{t \times r_a}$, with each entry being a permutation in $S_b$ 
	chosen uniformly at random and independently. 
		\item We also have a referee that picks $k$ random and independently chosen entries of this matrix, namely, $(\ell_i,v_i) \in [t] \times [r_a]$ for $i \in [k]$. We
		refer to $(\bsigmaP{1}_{\ell_1,v_1},\ldots,\bsigmaP{k}_{\ell_k,v_k})$ as the \textbf{hidden permutations} of the players. 
		We also define the \textbf{target permutation} $\gammastar$ as: 
		\[
			\gammastar := \bsigmaP{1}_{\ell_1,v_1} \circ \cdots \circ \bsigmaP{k}_{\ell_k,v_k}. 
		\]
	\end{enumerate}
	We are interested in communication protocols for this problem that follow the same rules as the ones for $\IndexHPH$, with the different goal 
	of changing the distribution of the target permutation from its original uniform distribution as much as possible. 
\end{problem}




\subsection{Part Two: KL-Divergence of Individual Hidden Permutations}\label{sec:part-two-individual}

In this subsection, we prove that not many players are able to reveal much information about their hidden permutation, measured as the KL-divergence of 
their hidden permutation conditioned on the transcript from the uniform distributions. This is proven using a simple \emph{direct-product} style argument to show that not only each player is unable to reveal much about their hidden permutation, 
but in fact this is true \emph{simultaneously}  for most players.

To continue, we need the following definition. 
\begin{Definition}\label{def:theta}
	Set the parameter
	\begin{align}
		\theta := \paren{\frac{4s + 4k \cdot \log{(t \cdot r_a)}}{t \cdot r_a}}^{1/2} \label{eq:def-theta}. 
	\end{align}
	(Notice that by the bound we have on $s$ in~\Cref{thm:mph}, $\theta < 1/100$). 
	
	For any $i \in [k]$ and $(\ell,v) \in [t] \times [r_a]$, we say that a transcript $\bProt$ is \textnormal{\textbf{informative}} for player $\QP{i}$ on the index $(\ell,v)$ iff: 
	\[
		\kl{\dist(\rbsigmaP{i}_{\ell,v} \mid \bProt)}{\cU_{S_b}} > \theta.
	\]
	Moreover, for a transcript $\bProt$ and input $L=(\ell_1,\ldots,\ell_k)$, $e=(v_1,\ldots,v_k)$ to the referee, we define the set of \textnormal{\textbf{informed indices}} as: 
	\[
		\info=\info(\bProt,L,e) := \set{i \in [k] \mid \text{$\Prot$ is informative for player $\QP{i}$ on $(\ell_i,v_i)$}}. 
	\]
\end{Definition}

Informally, this definition captures {which permutations} the referee has a lot of information about. Conditioned on any transcript and input indices to the referee, if the distribution of $\bsigmaP{i}_{\ell_i,v_i}$ differs from $\cU_{S_b}$ by $\theta$
 in KL-Divergence, the referee has a better chance of guessing $\bsigmaP{i}_{\ell_i,v_i}$, which in turn helps with reducing the uncertainty on the target permutation $\gammastar$. 
 
The following lemma shows that size of $\info$ cannot be too large with high probability. 

\begin{lemma}[``Informed indices cannot be too many'']\label{lem:info-size}
	~
	\[
		\Pr_{(\bProt,L,e)}\Paren{|{\info(\bProt,L,e)}| > k/2} \leq 2 \cdot (16 \cdot \theta)^{k/4}. 
	\]
\end{lemma}
\begin{proof}
	The proof consists of two parts. The first part is to show that with high probability over the choice of $\bProt$, the KL-divergence of 
	$\bSigmaP{i}$ conditioned on the transcript from its original distribution for ``most'' indices $i \in [k]$ is sufficiently ``low''. The second part then shows that 
	only a small fraction of these low KL-divergence indices can become informative, which concludes the proof. 
	
	\paragraph{Part I: KL-divergence of $\bSigmaP{i}$ is ``low'' for ``most'' $i \in [k]$.} For any transcript $\bProt$, define: 
	\begin{alignat*}{2}
		&L(\bProt) &&:= \set{i \in [k] ~\mid~ \kl{\rbSigmaP{i} \mid \rbProt=\bProt}{\rbSigmaP{i}} \leq 4s + 4 k \cdot \log{(1/\theta)}},
	\end{alignat*}
	namely, the indices in $[k]$ with low KL-divergence conditioned on $\bProt$ from their original distribution. We are going to prove that $L(\bProt)$ is 
	going to be large quite likely. 

	\begin{claim}\label{clm:info-size-p1}
		$
			\Pr_{\bProt}\paren{\card{L(\bProt)} \leq 3k/4} \leq \theta^{k}. 
		$
	\end{claim}
	\begin{proof}
		Let $S$ be the set of all transcripts $\bProt$ such that,
		\[
			\kl{\rbSigma \mid \rbProt=\bProt}{\rbSigma} > s+k \cdot \log{(1/\theta)}. 
		\]
		Note that since $\bProt$ is a deterministic function of $\rbSigma$, by~\Cref{fact:kl-event} (the moreover part), we have
		\[
			\kl{\rbSigma \mid \rbProt=\bProt}{\rbSigma} = \log\Paren{\frac{1}{\Pr(\rbProt=\bProt)}}. 
		\]
		Thus, for each $\bProt \in S$, we have, 
		\[
			\Pr\paren{\rbProt=\bProt} \leq 2^{-s-k\cdot \log{(1/\theta)}}. 
		\]
		At the same time, since there are at most $2^s$ choices for $\bProt$, by union bound, we have, 
		\[
			\Pr_{\bProt}\paren{\bProt \in S} \leq \sum_{\bProt \in S} \Pr\paren{\rbProt=\bProt} \leq 2^{s} \cdot 2^{-s-k \cdot \log{(1/\theta)}} = \theta^{k}.  
		\]
		In the following, we condition on $\bProt \notin S$ which happens with probability $1-\theta^k$. 
		
		We now have, 
		\begin{align*}
			s+k\cdot \log{(1/\theta)} &\geq \kl{\rbSigma \mid \rbProt=\bProt}{\rbSigma} \tag{by the definition of $\bProt$ being not informative} \\
			&\geq \sum_{i=1}^{k} \kl{\rbSigmaP{i} \mid \rbProt = \bProt}{\rbSigmaP{i}}. \tag{$\rbSigma$ is a product distribution (by~\Cref{clm:single-hyperedge-2}-\Cref{item:sh-bsigma}) so we can apply~\Cref{fact:kl-chain-rule} (moreover part)}
		\end{align*}
		A simple Markov bound implies that at most one-fourth of the indices can have KL-divergence more than $4s+4k\cdot\log{(1/\theta)}$. Thus, for any $\bProt \notin S$, $L(\bProt)$ is of size at least $3k/4$. 
		The bound of $\theta^k$ above on the probability of $\bProt$ being in $S$ concludes the proof. 
	\end{proof}
	
	In the following, we fix any choice of $\bProt$ with $L(\bProt)$ having size at least $3k/4$, which by~\Cref{clm:info-size-p1} happens with probability at least $1-\theta^k$. 
	
	\paragraph{Part II: ``Few'' indices in $L(\bProt)$ can be informative.} We bound the size of informative indices in $L(\bProt)$ as follows. 
	
	\begin{claim}\label{clm:info-size-p2}
	$\Pr_{(L,e) \mid \bProt}\Paren{\text{$\geq k/4$ indices in $L(\bProt)$ are informative}} \leq (16\theta)^{k/4}$.
	\end{claim}
	\begin{proof}
	For any $i \in L(\bProt)$, 
	\begin{align*}
		\kl{\rbSigmaP{i} \mid \rbProt=\bProt}{\rbSigmaP{i}} &\geq \sum_{\ell=1}^{t} \sum_{v \in R_i} \kl{\rbsigmaP{i}_{\ell,v} \mid \rbProt=\bProt}{\rbsigmaP{i}_{\ell,v}} 
		\tag{$\rbSigma$ is a product distribution (by~\Cref{clm:single-hyperedge-2}-\Cref{item:sh-bsigma}) so we can apply~\Cref{fact:kl-chain-rule} (moreover part)}.
	\end{align*}
	Thus, by the KL-divergence bound of the previous part for indices $i \in L(\bProt)$, 
	\[
		\Exp_{(\ell_i,v_i) \mid \bProt} \kl{\rbsigmaP{i}_{\ell_i,v_i} \mid \rbProt=\bProt}{\rbsigmaP{i}_{\ell_i,v_i}}  \leq \frac{4s + 4  k \cdot \log{(1/\theta)}}{t \cdot r_a} \leq \frac{4s + 4k \cdot \log{(tr_a)}}{t \cdot r_a} = \theta^2,
	\]
	where the first inequality is because $(\ell,v) \mid \bProt$ is uniformly distributed over $[t] \times R_i$ (by~\Cref{clm:single-hyperedge-2}-\Cref{item:sh-independence}), and the second inequality and subsequent equality hold
	by the choice of $\theta$. 
	
	We can now apply Markov bound and obtain that for every $i \in L(\bProt)$,  
	\[
		\Pr_{(\ell_i,v_i) \mid \bProt}\Paren{\kl{\rbsigmaP{i}_{\ell_i,v_i} \mid \rbProt=\bProt}{\rbsigmaP{i}_{\ell_i,v_i}} > \theta} \leq \theta. 
	\]
	Since $\rbsigmaP{i}_{\ell_i,v_i}$ is distributed as $\cU_{S_b}$ (by~\Cref{clm:single-hyperedge-2}-\Cref{item:sh-bsigma}),  this implies that for every $i \in L(\bProt)$, 
	\[
		\Pr_{(\ell_i,v_i) \mid \bProt}\Paren{\text{$i$ is informative}} \leq \theta. 
	\]
	In addition, since the choice of $(\ell_i,v_i) \mid \bProt$ for all $i \in L(\bProt)$ (generally in $[k]$) are independent of each other (by~\Cref{clm:single-hyperedge-2}-\Cref{item:sh-independence}), 
	we have that, 
	\begin{align*}
		\Pr_{(L,e) \mid \bProt}\Paren{\text{$\geq k/4$ indices in $L(\bProt)$ are informative}} &\leq \sum_{\substack{S \subseteq L(\bProt) \\ \card{S} = k/4}} \prod_{i \in S} \Pr_{(\ell_i,v_i) \mid \bProt}\Paren{\text{$i$ is informative}} \\
		&\leq 2^{k} \cdot \theta^{k/4} = (16\theta)^{k/4}. \qed
	\end{align*}
	
	\end{proof}
	\paragraph{Conclusion.} We can now conclude the proof of~\Cref{lem:info-size}. By~\Cref{clm:info-size-p1}, with probability $1-\theta^k$ over the choice of $\bProt$, at least $3k/4$ indices are in $L(\bProt)$ (and so
	we can have at most $k/4$ informative indices outside $L(\bProt)$). 
	For any $\bProt$ with $\card{L(\bProt)} \geq 3k/4$, by~\Cref{clm:info-size-p2}, at most $k/4$ indices $i \in L(\bProt)$ can be informative for $(\bProt,L,e)$. Thus, by union bound, 
	\[
		\Pr_{(\bProt,L,e)}\Paren{|{\info(\bProt,L,e)}| > k/2} \leq \theta^k + (16 \cdot \theta)^{k/4} \leq 2 \cdot (16 \cdot \theta)^{k/4}. 
	\]
	This finalizes the proof of~\Cref{lem:info-size}. 
\end{proof}

\paragraph{Conditioning Step:} 
For the rest of the proof, we condition on any choice of $(\bProt,L,e)$ such that 
\begin{align}
	\info := \info(\bProt,L,e) \leq k/2, \label{eq:info-size}
\end{align}
and we know that by~\Cref{lem:info-size}, this event happens with high probability.




\newcommand{\eventA}{\event_{A}}
\newcommand{\eventB}{\event_{B}}

\renewcommand{\good}{\textsc{Good}}

\newcommand{\tsigma}{\tilde{\sigma}}
\newcommand{\rtsigma}{\tilde{\bm{\sigma}}}
\newcommand{\tgamma}{\tilde{\gamma}}
\newcommand{\rtgamma}{\tilde{\bm{\gamma}}}

\subsection{Part Three: Total Variation Distance of the Target Permutation}

At this point of the argument, we have already fixed the choice of $(\bProt,L,e)$ which, by~\Cref{eq:info-size}, gives us that
for many indices $i \in [k]$, the distribution of the hidden permutation $\bsigmaP{i}_{\ell_i,v_i}$ is close to uniform in the KL-divergence. Recall that the target
distribution $\gammastar$ is
\[
	\gammastar = \bsigmaP{1}_{\ell_1,v_1} \circ \cdots \circ \bsigmaP{k}_{\ell_k,v_k}. 
\]
Our goal is to show that this concatenation is going to make the distribution of $\gammastar$ exponentially closer to uniform (albeit in another measure of distance, not KL-divergence). 
This is proven in the following three steps: 
\begin{itemize}[leftmargin=15pt]
\item \textbf{Step I:} We first prove that even conditioned on $(\bProt,L,e)$, the distribution of hidden permutations are independent of each other. This is an application of the rectangle property of protocols. 
\item \textbf{Step II:} We then show that for many (but not all) of the hidden permutations, the KL-divergence bounds can be translated to bounds on the $\ell_2$-distance of the distribution from uniform. 
A crucial tool we use here is an strengthened Pinsker's inequality, due to~\cite{ChakrabartyK18}, that bounds a mixture of $\ell_1$- and $\ell_2$-distance between distributions via KL-divergence. The main part of the argument 
here is to handle the $\ell_1$-distance bounds, and postpone the $\ell_2$-distances to the next step. 
\item \textbf{Step III:} Finally, we use the Fourier analytic tools of~\Cref{app:fourier-permutation} combined with the $\ell_2$-distance bounds of the previous step, to bound the $\ell_2$-distance of the distribution of $\gammastar$ from uniform, and 
get the desired bound on the total variation distance as well.
\end{itemize}

\subsubsection{Step I: Conditional Independence of Inputs Even After Communication} 

We prove that players' inputs remain independent even conditioned on the messages and input of the referee. 
The proof is a standard application of the rectangle property of communication protocols extended to the multi-party setting and two specific parts in the inputs. But, despite its simplicity, this lemma  plays a crucial rule in our arguments in the last step. 

\begin{lemma}\label{lem:conditional-independence-hcp}
	For any transcript $\bProt$ and input $L=(\ell_1,\ldots,\ell_k), e = (v_1,\ldots,v_k)$ to the referee in $\mu$,    
	\begin{alignat*}{2}
		\dist({\rbsigmaP{1}_{\ell_1,v_1},\ldots,\rbsigmaP{k}_{\ell_k,v_k} \mid \bProt, L,e}) &= \prod_{i=1}^{k} \dist({\rbsigmaP{i}_{\ell_i,v_i} \mid \bProt}); 
	\end{alignat*}
\end{lemma}
\begin{proof}
	By the chain rule of probabilities, 
	\[
		\dist({\rbsigmaP{1}_{\ell_1,v_1},\ldots,\rbsigmaP{k}_{\ell_k,v_k} \mid \bProt, L,e}) = \prod_{i=1}^{k} \dist(\rbsigmaP{i}_{\ell_i,v_i} \mid \bProt, \bsigmaP{<i}_{L_{<i},e_{<i}},L,e) = \prod_{i=1}^{k} \dist(\rbsigmaP{i}_{\ell_i,v_i} \mid \bProt, \bsigmaP{<i}_{L_{<i},e_{<i}}), 
	\]
	where the final equality is by~\Cref{clm:single-hyperedge-2}-\Cref{item:sh-independence} because $(L,e)$ are independent of $(\bSigma,\bProt)$. 
	Thus, to prove the lemma, we only need to prove that for every $i \in [k]$, 
	\[
		\paren{\rbsigmaP{i}_{\ell_i,v_i} \perp \rbsigmaP{<i}_{L_{<i},e_{<i}} \mid \rbProt} \quad  \equiv \quad \mi{\rbsigmaP{i}_{\ell_i,v_i}}{\rbsigmaP{<i}_{L_{<i},e_{<i}} \mid \rbProt} = 0, 
	\]
	where the equivalence is by~\itfacts{info-zero}. 
	
	We have, 
	\begin{align*}
		\mi{\rbsigmaP{i}_{\ell_i,v_i}}{\rbsigmaP{<i}_{L_{<i},e_{<i}} \mid \rbProt} &\leq \mi{\rbSigmaP{i}}{\rbSigmaP{-i} \mid \rbProt} 
	\end{align*}
	by the data processing inequality (\itfacts{data-processing}), as $\rbSigmaP{i}$ fixes $\rbsigmaP{i}_{\ell_i,v_i}$ and $\rbSigmaP{-i}$ fixes $\rbsigmaP{<i}_{L_{<i},e_{<i}}$; notice that here $(L,e)$ are fixed and we are looking
	at specific indices of $\rbsigmaP{i}_{\ell_i,v_i}$ of $\rbSigmaP{i}$ and $\rbsigmaP{<i}_{L_{<i},e_{<i}}$ of $\rbSigmaP{-i}$, and as such we can indeed apply the data processing inequality. 
	
	With a slight abuse of notation, 
	we denote $\rbProt = \rbProt(1),\ldots,\rbProt(s)$ where for all $j \in [s]$, $\rbProt(j)$ denotes the $j$-th bit communicated by the players. We claim that for every $j \in [s]$, 
	\begin{align*}
		\mi{\rbSigmaP{i}}{\rbSigmaP{-i} \mid \rbProt(1),\ldots,\rbProt(j)} &\leq \mi{\rbSigmaP{i}}{\rbSigmaP{-i} \mid \rbProt(1),\ldots,\rbProt(j-1)};
	\end{align*}
	This is because: 
	\begin{itemize}
	\item if the $j$-th bit of the protocol is sent by player $\QP{i}$, then it is a deterministic function of $\rbSigmaP{i}$ and $\rbProt(1),\ldots,\rbProt(j-1)$ and thus
	\[
		\rbProt(j) \perp \rbSigmaP{-i} \mid \rbSigmaP{i}, \rbProt(1),\ldots,\rbProt(j-1);
	\]
	hence, the inequality holds by~\Cref{prop:info-decrease}. 
	\item  if the $j$-th bit of the protocol is sent by any player other than $\QP{i}$, then it is a deterministic function of $\rbSigmaP{-i}$ and $\rbProt(1),\ldots,\rbProt(j-1)$ and thus
	\[
		\rbProt(j) \perp \rbSigmaP{i} \mid \rbSigmaP{-i}, \rbProt(1),\ldots,\rbProt(j-1);
	\]
	hence, again, the inequality holds by~\Cref{prop:info-decrease}. 
	\end{itemize}
	Applying this inequality repeatedly then gives us 
	\[
		\mi{\rbSigmaP{i}}{\rbSigmaP{-i} \mid \rbProt(1),\ldots,\rbProt(j)} \leq \mi{\rbSigmaP{i}}{\rbSigmaP{-i}} = 0,
	\]
	where the second equality holds by~\itfacts{info-zero} because of the independence of the parameters (as shown in~\Cref{clm:single-hyperedge-2}-\Cref{item:sh-bsigma}). This concludes the proof of the lemma. 
\end{proof}

\subsubsection{Step II: From KL-Divergence to ``Strong'' $\ell_2$-Bounds}

Recall that we already fixed a choice of $(\bProt,L,e)$ that guarantees $\info=\info(\bProt,L,e)$ has size at most $k/2$ by~\Cref{eq:info-size}. We further define the following two sets for every $i \in [k]$: 
\begin{alignat*}{2}
	&A_i &&:= \set{\sigma \in S_b \mid \dist(\rbsigmaP{i}_{\ell_i,v_i} = \sigma \mid \bProt) > {2}/{b!}}; \\
	&B_i &&:= S_b \setminus A_i. 
\end{alignat*}
In words, $A_i$ is the set of those permutations in $S_b$ whose probability mass under $\rbsigmaP{i}_{\ell_i,v_i} \mid \bProt$ is ``much higher'' than that of uniform distribution (more precisely, more than twice), and $B_i$ collects the remaining permutations. 
We also define two events: 
\begin{itemize}
	\item \textbf{Event $\eventA(i)$}: the sampled input $\rbsigmaP{i}_{\ell_i,v_i}$ belongs to $A_i$; 
	\item \textbf{Event $\eventB(i)$}: the sampled input $\rbsigmaP{i}_{\ell_i,v_i}$ belongs to $B_i$. 
\end{itemize}
Note that these two events are complement of each other and only depend on the choice of $\bsigmaP{i}_{\ell_i,v_i}$ at this point.  
Finally, we define the set of \textbf{good} indices $i \in [k]$ as: 
\[
	\good := \set{i \in [k] \mid \text{$i$ is \underline{not} in $\info$ and $\eventB(i)$ happens}}. 
\]
The reason we consider indices in $\good$ as ``good'' is because we can translate our KL-divergence bound on $\rbsigmaP{i}_{\ell_i,v_i}$ for $i \in \good$ (even conditioned on them being in good) into a ``strong'' 
bound on their $\ell_2$-distance from the uniform distribution (this is formalized in~\Cref{lem:good}). 

We start by showing that $\eventB$ happens frequently for indices not in $\info$. 

\begin{claim}\label{clm:bad-each-i}
		For any $i \notin \info$: 
		\[
		\Pr\paren{\eventB{(i)} \mid \bProt} \geq 1- \sqrt{8\theta}.
		\] 
\end{claim}
\begin{proof}
	By the definition of $i$ not being in $\info$ and by Pinsker's inequality of~\Cref{fact:pinskers}, we have, 
	\begin{align*}
		\tvd{\dist(\rbsigmaP{i}_{\ell_i,v_i} \mid \bProt)}{\cU_{S_b}} \leq \sqrt{\frac{1}{2} \cdot \kl{\dist(\rbsigmaP{i}_{\ell_i,v_i} \mid \bProt)}{\cU_{S_b}}} \leq \sqrt{\frac{\theta}{2}}. 
	\end{align*}
	On the other hand, 
	\begin{align*}
	\tvd{\dist(\rbsigmaP{i}_{\ell_i,v_i} \mid \bProt)}{\cU_{S_b}} &\geq \frac12 \cdot \sum_{\sigma \in A_i} \card{\dist(\rbsigmaP{i}_{\ell_i,v_i} = \sigma \mid \bProt)-\frac{1}{b!} } \\
			 & \geq \frac{1}{4}  \cdot \sum_{\sigma \in A_i} \dist(\rbsigmaP{i}_{\ell_i,v_i} = \sigma \mid \bProt),
	\end{align*}
	by the definition of $A_i$ (note that all permutations in $A_i$ have a probability of at least~$2/b!$ in our distribution, higher than compared to the uniform distribution). 
	Combining the above two equations gives us: 
	\begin{align*}
		\Pr\paren{\eventA(i) \mid \bProt} = \sum_{\sigma \in A_i} \dist(\rbsigmaP{i}_{\ell_i,v_i} = \sigma \mid \bProt)\leq \sqrt{8\theta}. 
	\end{align*}
	As $\eventB(i)$ is the complement of $\eventA(i)$, we can conclude the proof. 
\end{proof}
\noindent
A simple corollary of~\Cref{clm:bad-each-i} is the following (probabilistic) lower bound on the size of $\good$: 
\begin{align}
	\Pr\paren{\card{\good} < k/4} \leq {k/2 \choose k/4} \cdot ({8\theta})^{k/8} \leq (128\theta)^{k/8} \label{eq:good-size}; 
\end{align}
here, we used the fact that for $\good$ to be of size less than $k/4$, at least $k/4$ indices from the first $k/2$ indices of $\info$ should have $\eventA$ happen for them; we then 
used~\Cref{clm:bad-each-i} to bound this probability and combined it with~\Cref{lem:conditional-independence-hcp} to crucially use the independence of the events $\eventA(i)$ for 
different values of $i \in [k]$ (the second inequality is by just upper bounding ${k/2 \choose k/4}$ by $2^{k/2} = 16^{k/8}$).

The following lemma now establishes why the indices in $\good$ are ``good'' for our purpose: on these indices, the distribution of $\bsigmaP{i}_{\ell_i,v_i}$ is ``quite'' close to uniform in $\ell_2$-distance, much closer than a $\theta^2$-bound that follows from 
directly applying Pinsker's inequality to the KL-divergence bound of indices not in $\info$; we prove this using the $\ell_2/\ell_1$-version of Pinsker's inequality due to~\cite{ChakrabartyK18} mentioned in~\Cref{prop:pinsker++}.  

\begin{lemma}\label{lem:good}
	For any $i \in \good$: 
	\[
		\norm{\dist(\rbsigmaP{i}_{\ell_i,v_i} \mid \bProt, i \in \good) - \cU_{S_b}}_2^2 \leq \frac{20\sqrt{\theta}}{b!}. 
	\]
\end{lemma}
\begin{proof}
	Firstly, consider an index $i \notin \info$ (but we still do not condition on $i$ being in $\good$ or not). By the definition of $\info$, we have, 
	\[
		\kl{\dist({\rbsigmaP{i}_{\ell_i,v_i} \mid \bProt})}{\cU_{S_b}} \leq \theta. 
	\]
	By applying~\Cref{prop:pinsker++} to this KL-divergence bound (taking $A=A_i$ and $B=B_i$ in the proposition), we get
	\[
		\sum_{\sigma \in A_i} \card{\dist({\rbsigmaP{i}_{\ell_i,v_i}=\sigma \mid \bProt})-\frac{1}{b!}} + 
		\sum_{\sigma \in B_i} \frac{\paren{\dist({\rbsigmaP{i}_{\ell_i,v_i}=\sigma \mid \bProt})-\dfrac{1}{b!}}^2}{\dist({\rbsigmaP{i}_{\ell_i,v_i}=\sigma \mid \bProt})} \leq \frac{1}{1-\ln{2}} \cdot \theta \leq 4\theta. 
	\]
	By just taking the second term, and since for all $\sigma \in B_i$, $\dist({\rbsigmaP{i}_{\ell_i,v_i}=\sigma \mid \bProt}) \leq 2/b!$, 
	we get that 
	\begin{align}
		\sum_{\sigma \in B_i} {\paren{\dist({\rbsigmaP{i}_{\ell_i,v_i}=\sigma \mid \bProt})-\dfrac{1}{b!}}^2} \leq \frac{8\theta}{b!}. \label{eq:first-l2-norm}
	\end{align}
	We now consider the effect of conditioning on $i \in \good$ as well. 
	Define $\delta_i \in \IR$ such that  
	\[
		(1+\delta_i) := \dist(i \in \good \mid \bProt)^{-1} = \Paren{\dist({\eventB(i) \mid \bProt})}^{-1},
	\]
	which by~\Cref{clm:bad-each-i} and because $\theta < 1/100$ (see~\Cref{def:theta}), leads to
	\begin{align}
		\delta_i \leq \sqrt{18\cdot\theta}. \label{eq:deltai-eps}
	\end{align}
	This gives us, for every $\sigma \in B_i$,
	\begin{equation}\label{eq:condition-good-delta}
		\dist({\rbsigmaP{i}_{\ell_i,v_i}=\sigma \mid \bProt, i \in \good)} = \frac{\dist({\rbsigmaP{i}_{\ell_i,v_i}=\sigma \wedge i \in \good \mid \bProt})}{\dist({i \in \good \mid \bProt})} = 
		(1+\delta_i) \cdot \dist({\rbsigmaP{i}_{\ell_i,v_i}=\sigma \mid \bProt}),
	\end{equation}
	where in the last equality, we used the fact that for $\sigma \in B_i$, $\rbsigmaP{i}_{\ell_i,v_i} = \sigma$ also implies $i \in \good$ (as $\eventB(i)$ happens). 
	We can now calculate the $\ell_2$-distance of our desired distribution from $\cU_{S_b}$. For the simplicity of exposition in the following calculations, we denote, 
	\[
		\nu(\sigma) := \dist({\rbsigmaP{i}_{\ell_i,v_i}=\sigma \mid \bProt)}. 
	\]
	We have, 
	\begin{align*}
		&\norm{\dist({\rbsigmaP{i}_{\ell_i,v_i}=\sigma \mid \bProt, i \in \good)}- \cU_{S_b}}_2^2 \\
		&~~= \sum_{\sigma \in B_i} \paren{\dist({\rbsigmaP{i}_{\ell_i,v_i}=\sigma \mid \bProt, i \in \good)}-\frac{1}{b!}}^2 \tag{conditioning on $i \in \good$, implies that only $\sigma \in B_i$ have non-zero probability}\\
		&~~= \sum_{\sigma \in B_i} \paren{(1+\delta_i)  \cdot \nu(\sigma)-\frac{1}{b!}}^2 \tag{by \Cref{eq:condition-good-delta}  and the definition of $\nu(\sigma)$ above} \\
		&~~= \sum_{\sigma \in B_i} \paren{\paren{ \nu(\sigma)-\frac{1}{b!}}^2 + {\delta_i}^2 \cdot  \nu(\sigma)^2 + 
		2\delta_i \cdot \nu(\sigma) \cdot (\nu(\sigma)-\frac{1}{b!}))} \\
		&~~\leq \sum_{\sigma \in B_i} \paren{\paren{ \nu(\sigma)-\frac{1}{b!}}^2 + {\delta_i}^2 \cdot  \nu(\sigma)^2 + 
		2\delta_i \cdot \nu(\sigma) \cdot \frac{1}{b!}} \tag{as $\nu(\sigma)  \leq 2/b!$ since $\sigma \in B_i$}\\
		&~~\leq \paren{\sum_{\sigma \in B_i} {\paren{\nu(\sigma)-\dfrac{1}{b!}}^2}} + \frac{4{\delta_i}^2}{b!} + \frac{4\delta_i}{b!} \tag{by the bound of $\nu(\sigma) \leq 2/b!$ for $\sigma \in B_i$, and $\card{B_i} \leq b!$} \\
		&~~\leq \frac{8\theta + 72\theta + \sqrt{288\cdot \theta}}{b!} \tag{by~\Cref{eq:first-l2-norm} for the first term and~\Cref{eq:deltai-eps} for $\delta_i$-terms} \\
		&~~\leq \frac{20\sqrt{\theta}}{b!}, \tag{as $\theta < 1/100$}  
	\end{align*}
	concluding the proof. 
\end{proof}

\paragraph{Conditioning Step:} For the rest of the proof, we further condition on the choice of the set $\good$ such that 
\begin{align}
	\card{\good} \geq k/4, \label{eq:good-k4}
\end{align}
and we know that by~\Cref{eq:good-size}, this event happens with high probability. We then condition on the 
entire choice of $\bsigmaP{i}_{\ell_i,v_i} \in S_b$ for any $i \notin \good$. At this point, the only 
remaining variables which are not fixed are $\bsigmaP{i}_{\ell_i,v_i}$ for $i \in \good$, which are still chosen independently by~\Cref{lem:conditional-independence-hcp} (although we have conditioned on $i \in \good$, which
influences their individual distribution). 

\subsubsection{Step III: Amplified $\ell_2$-Distance for the Target Permutation}

We now go over the last step of our proof by analyzing the distribution of the target permutation 
\[
	\gammastar = \bsigmaP{1}_{\ell_1,v_1} \circ \cdots \circ \sigmaP{k}_{\ell_k,v_k}; 
\]
in particular, we show that the distribution of $\gammastar$ is exponentially closer to the uniform distribution compared to the bounds of~\Cref{lem:good} as a result 
of concatenation (namely, that independent concatenation reduces the ``bias''). 

To continue, we need some notation. 
Let $g := \card{\good}$ and $i_1,\ldots,i_g$ be indices in $\good$ sorted in increasing order. Define the following $g$ permutations for $j \in [g]$ as concatenation of each $\bsigmaP{i_j}_{\ell_{i_j},v_{i_j}}$ 
with all subsequent permutations which are not in $\good$ until we hit the index $i_{j+1}$; formally, 
\begin{align*}
	\tsigma_{j} := \bsigmaP{i_j}_{\ell_{i_j},v_{i_j}} \quad \circ \quad \bsigmaP{i_j+1}_{\ell_{i_j+1},v_{i_j+1}} \quad \ldots \quad \circ \quad \bsigmaP{i_{j+1}-1}_{\ell_{i_{j+1}-1},v_{i_{j+1}-1}}; 
\end{align*}
Notice that the distribution of each $\tsigma_j$ is the same as that of $\bsigmaP{i_j}_{\ell_{i_j},v_{i_j}} \mid \bProt,i_j \in \good $ except for a fixed ``shift'' (by the fixed permutations indexed between $i_j , i_{j+1} \in \good$). Moreover, define 
\[
	\tgamma := \tsigma_{1} \circ \ldots \circ \tsigma_{g};
\]
the distribution of $\tgamma$ is now the same as that of the target permutation $\gammastar  \mid \bProt,i_j \in \good $, again, except for a fixed shift (by the fixed permutations before $i_1 \in \good$). 
Thus, for $j \in [g]$, we have, 
\begin{equation}
    \begin{aligned}
	\norm{\dist(\rtsigma_j) - \cU_{S_b}}_2^2& &&= \norm{\dist(\rbsigmaP{i_j}_{\ell_{i_j},v_{i_j}} \mid \bProt, i_j \in \good) - \cU_{S_b}}_2^2 \leq \frac{20\sqrt{\theta}}{b!};  \\
	\norm{\dist(\rtgamma)- \cU_{S_b}}^2_2& &&= \norm{\dist(\rgammastar \mid \bProt,  \good) - \cU_{S_b}}_2^2; \\
	\dist(\rtsigma_1,\ldots,\rtsigma_g)& &&= ~\dist(\rtsigma_1) \times \cdots \times \dist(\rtsigma_g).
    \end{aligned}
    \label{eq:tilde-ones}
\end{equation}
here, the first inequality of the first equation is by~\Cref{lem:good} and equality of the last equation is by~\Cref{lem:conditional-independence-hcp}. 

We can now prove the following lemma on the distance of $\tgamma$ from the uniform distribution, using a basic application of the Fourier analysis on permutations reviewed in~\Cref{app:fourier-permutation}. 

\begin{lemma}[``Concatenation reduces $\ell_2$-distances'']\label{lem:conc-amplify}
		\[
			\norm{\dist(\rtgamma)- \cU_{S_b}}^2_2 \leq \frac{({20\sqrt{\theta})}^{g}}{b!}
		\]
\end{lemma}
\begin{proof}
	Recall the definition of the Fourier transform and the set of irreducible representations for $S_b$ from~\Cref{app:fourier-permutation}. 
	For the simplicity of exposition in this proof, we denote,  
	\begin{align*}
		\nu_j &:= \dist(\rtsigma_j)~\text{for every $j \in [g]$}, \quad \text{and} \quad \nu := \dist(\rtgamma) = \nu_1 \circ \cdots \circ \nu_g. 
	\end{align*}
	For the distribution $\mu$, by Plancherel's identity of \Cref{prop:Four-plancherel},
	\begin{align}
		\norm{\nu-\cU_{S_b}}_2^2 &= \sum_{\sigma \in S_b} \paren{\nu(\sigma)-\frac1{b!}}^2 \tag{by the definition of the $\ell_2$-norm} \\
		&= \frac1{b!}  \cdot \sum_{\rho \in \reps} d_{\rho} \sum_{i,j \in [d_{\rho}]} \paren{\hrho{\nu}-\hrho{\cU_{S_b}}}^2_{i,j} \tag{by~\Cref{prop:Four-plancherel}} \\
		&= \frac1{b!} \cdot\paren{ d_{\rho_0} \paren{\hrhot{\nu}-\hrhot{\cU_{S_b}}}^2 + \sum_{\rho \neq \rho_0} d_{\rho} \sum_{i,j \in [d_{\rho}]} \paren{\hrho{\nu}-\hrho{\cU_{S_b}}}^2_{i,j}} \tag{by splitting the outer sum over $\rho_0$ and 
		remaining representations in $\reps$}\\
		&= \frac1{b!} \cdot \paren{\sum_{\rho \neq \rho_0} d_{\rho} \sum_{i,j \in [d_{\rho}]} \hrho{\nu}_{i,j}^2} \tag{as $\hrhot{\nu} = \hrho{\cU_{S_b}} = 1$, and $\hrho{\cU_{S_b}} = \mathbf{0} \in \IR^{d_{\rho} \times d_{\rho}}$ for $\rho \neq \rho_0$ by~\Cref{fact:Four-trans}} \\
		&= \frac1{b!} \cdot  \sum_{\rho \neq \rho_0} d_{\rho} \cdot\norm{\hrho{\nu}}^2_F, \label{eq:maybe-1}
	\end{align}
	where $\norm{\cdot}_F$ denotes the Frobenius norm of the matrix $\hrho{\nu} \in \IR^{d_{\rho} \times d_{\rho}}$.  Doing the same exact calculation for each $\nu_j$ for $j \in [g]$, we also get, 
	\begin{align}
		\sum_{\rho \neq \rho_0} d_{\rho} \cdot\norm{\hrho{\nu_j}}^2_F = b! \cdot \norm{\nu_j-\cU_{S_b}}_2^2 \leq {20\sqrt{\theta}}, \label{eq:maybe-2}
	\end{align}
	where the inequality is by~\Cref{eq:tilde-ones}. 

	Finally, note that for every $\rho \in \reps$, by~\Cref{fact:Four-convo} (the convolution theorem), 
	\[
		\hrho{\nu} = \prod_{j=1}^{g} \hrho{\nu_j}. 
	\]
	Combining all these, we now have, 
	\begin{align*}
		 \norm{\nu-\cU_{S_b}}_2^2 &= \frac1{b!} \cdot \sum_{\rho \neq \rho_0} d_{\rho} \cdot \norm{\hrho{\nu}}^2_F \tag{by~\Cref{eq:maybe-1}}\\
		 &= \frac1{b!} \cdot \sum_{\rho \neq \rho_0} d_{\rho} \cdot \norm{\prod_{j=1}^{g}\hrho{\nu_j}}^2_F \tag{by the application of~\Cref{fact:Four-convo} right above} \\
		&\leq \frac1{b!} \cdot \sum_{\rho \neq \rho_0} d_{\rho} \cdot \prod_{j=1}^{g} \norm{\hrho{\nu_j}}^2_F \tag{as Frobenius norm is sub-multiplicative} \\
		&\leq   \frac1{b!} \cdot \sum_{\rho \neq \rho_0} \prod_{j=1}^{g} d_{\rho} \cdot \norm{\hrho{\nu_j}}^2_F \tag{as $d_{\rho}\geq 1$} \\
		&\leq \frac1{b!} \cdot \prod_{j=1}^{g} \sum_{\rho \neq \rho_t} d_{\rho} \cdot \norm{\hrho{\nu_j}}^2_F \tag{as $d_{\rho}, \norm{\hrho{\nu_j}}^2_F \geq 0$ for each $\rho \in \reps$ and $j \in [g]$} \\
		&\leq \frac1{b!} \cdot \prod_{j=1}^{g} (20\sqrt{\theta}) \tag{by~\Cref{eq:maybe-2}} \\
		&= \frac{(20\sqrt{\theta})^{g}}{b!},
	\end{align*}
	concluding the proof. 
\end{proof}

It is worth mentioning that the calculations we have in~\Cref{lem:conc-amplify} are tight (see~\Cref{app:tight-remark} for an example and more discussion). 

Plugging back the original variables we had in~\Cref{eq:tilde-ones} inside~\Cref{lem:conc-amplify}, and using the bound of~\Cref{eq:good-k4} on the size of $\good$, we obtain that 
\[
	\norm{\dist(\rgammastar \mid \bProt,  \good) - \cU_{S_b}}_2^2 \leq \frac{1}{b!} \cdot (20\sqrt{\theta})^{k/4}, 
\]
which in turn, by the $\ell_1$-$\ell_2$ gap, gives us
\begin{align}
	\tvd{\dist(\rgammastar \mid \bProt,  \good)}{\cU_{S_b}} &\leq \norm{\dist(\rgammastar \mid \bProt,  \good) - \cU_{S_b}}_1 \tag{by the definition of total variation distance in~\Cref{eq:tvd}} \\
	&\leq \sqrt{b!} \cdot \norm{\dist(\rgammastar \mid \bProt,  \good) - \cU_{S_b}}_2 \tag{as $\norm{x}_1 \leq \sqrt{m} \cdot \norm{x}_2$ for any $x \in \IR^m$} \\
	&\leq (20\sqrt{\theta})^{k/8}. \label{eq:final-one}
\end{align}
We have thus finally bounded the TVD of the target permutation from the uniform distribution (under conditioning on several high probability events), and are now done with the proof.




\subsection{Putting Everything Together: Proof of~\Cref{thm:mph}}

We are now ready to conclude the proof of~\Cref{thm:mph} by simply retracing back all the steps of the proof and accounting for the probability of several events we conditioned on. 
The rest of this proof is basically bookkeeping and some tedious calculations. 

Recall that $\prot$ is any deterministic protocol for $\IndexHPH$ that uses $s$ bits of communication in total. The proof consisted of the following: 
\begin{itemize}
	\item In~\Cref{eq:basic-problem}, we proved that, 
	\[
	\Pr\paren{\text{$\pi$ succeeds}} \leq \frac12 + 
	\frac r2 \cdot \!\!\! \max_{\substack{\cM_{<a},\Sigma_{*,\cM_{<a}}\\ \mu=\mu_{\cM_{<a},\Sigma_{*,\cM_{<a}}}}} 
	\!\!\!\bracket{\Exp_{\bProt \sim \dist} \, \Exp_{\substack{\ell_1,\ldots,\ell_k \\ v_1,\ldots,v_k} \sim \dist} \tvd{\dist(\rbsigmaP{1}_{\ell_1,v_1} \circ \cdots \circ \rbsigmaP{k}_{\ell_k,v_k} \mid \Prot)}{\cU_{S_b}}}.
	\]
	We then fixed any choice of $\cM_{<a},\Sigma_{*,\cM_{<a}}$ and $\mu = \mu_{\cM_{<a},\Sigma_{*,\cM_{<a}}}$ for the rest of the proof. 
	
	\item In~\Cref{eq:info-size}, we fixed any choice of $(\bProt,L,e)$ that guarantees
	\[
	\info := \info(\bProt,L,e) \leq k/2.
	\]
	By~\Cref{lem:info-size}, the probability of $(\bProt,L,e) \sim \mu$ \underline{not} satisfying this is at most $2 \cdot (64\theta)^{k/4}$. 
	
	\item In~\Cref{eq:good-k4}, we fixed any choice of $\good$ that guarantees 
	\[
		\card{\good} \geq k/4. 
	\]
	By~\Cref{eq:good-size}, the probability that $\good$ does \underline{not} satisfy this guarantee is at most $(128\theta)^{k/8}$.  
	\item In~\Cref{eq:final-one}, we proved that for this choice of $(\bProt,L,e)$ and conditioned on the choice of $\good$, 
	\[
		\tvd{\dist(\rbsigmaP{1}_{\ell_1,v_1} \circ \cdots \circ \rbsigmaP{k}_{\ell_k,v_k} \mid \bProt,\good)}{\cU_{S_b}} \leq (20\sqrt{\theta})^{k/8} \leq 2 \cdot \theta^{k/16}. 
	\]
\end{itemize}
Let us now retrace back. 

\begin{itemize}
\item For any $\cM_{<a},\Sigma_{*,\cM_{<a}}$, $\mu = \mu_{\cM_{<a},\Sigma_{*,\cM_{<a}}}$, and $(\bProt,L,e)$ satisfying $\info(\bProt,L,e) \leq k/2$, we have, 
\begin{align*}
	&\tvd{\dist(\rbsigmaP{1}_{\ell_1,v_1} \circ \cdots \circ \rbsigmaP{k}_{\ell_k,v_k} \mid \bProt)}{\cU_{S_b}}  \\
	&\hspace{15pt} \leq \Pr\paren{\card{\good} <k/4} + \Exp_{\substack{\good \sim \dist \mid \\  \card{\good} \geq k/4}} \tvd{\dist(\rbsigmaP{1}_{\ell_1,v_1} \circ \cdots \circ \rbsigmaP{k}_{\ell_k,v_k} \mid \bProt,\good)}{\cU_{S_b}}
	\tag{by~\Cref{fact:tvd-over-conditioning} and since $\norm{\cdot}_{\textnormal{tvd}} \leq 1$} \\
	&\hspace{15pt} \leq (128\theta)^{k/8} + 2 \cdot \theta^{k/16} 
\end{align*}
\item This in turn implies that for any $\cM_{<a},\Sigma_{*,\cM_{<a}}$, $\mu = \mu_{\cM_{<a},\Sigma_{*,\cM_{<a}}}$, 
\begin{align*}
	&\Exp_{\bProt,L,e \sim \dist} \tvd{\dist(\rbsigmaP{1}_{\ell_1,v_1} \circ \cdots \circ \rbsigmaP{k}_{\ell_k,v_k} \mid \Prot)}{\cU_{S_b}} \\
	&\hspace{15pt} \leq \Pr_{\mu}\paren{\card{\info(\bProt,L,e)} > k/2} + \Exp_{\substack{(\bProt,L,e) \sim \dist \mid \\  \card{\inf(\bProt,L,e} \leq k/2}} \tvd{\dist(\rbsigmaP{1}_{\ell_1,v_1} \circ \cdots \circ \rbsigmaP{k}_{\ell_k,v_k} \mid \bProt)}{\cU_{S_b}} \\
	&\hspace{15pt} \leq 2 \cdot (64\theta)^{k/4} +  (128\theta)^{k/8} + 2 \cdot \theta^{k/16}. 
\end{align*}
\item Finally, we can plug in this bound in~\Cref{eq:basic-problem} and have, 
\begin{align*}
	\Pr\paren{\text{$\pi$ succeeds}} &\leq \frac12 + \frac{r}{2} \cdot  \max_{\substack{\cM_{<a},\Sigma_{*,\cM_{<a}}\\ \mu=\mu_{\cM_{<a},\Sigma_{*,\cM_{<a}}}}} \Exp_{\bProt,L,e \sim \dist} \tvd{\dist(\rbsigmaP{1}_{\ell_1,v_1} \circ \cdots \circ \rbsigmaP{k}_{\ell_k,v_k} \mid \Prot)}{\cU_{S_b}} \\
	&\leq \frac12 + r \cdot O(\theta)^{k/16} \\
	&\leq \frac12 + r \cdot O\!\paren{\frac{8s + 8k \cdot \log{(t \cdot r)}}{t \cdot r}}^{k/32} \tag{by the definition of $\theta$ in~\Cref{def:theta} and since $r_a \geq r/2$ by~\Cref{clm:single-hyperedge-2}-\Cref{item:sh-Ri}} \\
	&\leq \frac12 + r \cdot O\!\paren{\frac{s}{r \cdot t}}^{k/32} \tag{by the lower bound on the size of $s$ in~\Cref{thm:mph}}.
\end{align*}
\end{itemize}
This concludes the proof for all deterministic protocols $\prot$ on the distribution induced by $\IndexHPH$. The lower bound directly extends to randomized protocols 
by the easy direction of Yao's minimax principle (namely, an averaging argument over the randomness of the protocol on the input distribution). This concludes the proof of~\Cref{thm:mph}. \qed

\clearpage

\renewcommand{\first}{\ensuremath{\textnormal{\textsc{First}}}}
\renewcommand{\last}{\ensuremath{\textnormal{\textsc{Last}}}}

\section{One Pass Permutation Hiding}\label{sec:one-pass}

In this section, we will construct permutation hiding graphs for 1-pass streaming algorithms. We repeat the key definitions here (\Cref{def:perm-graph} and~\Cref{def:perm-hiding}) for the convenience of the reader. 

\begin{Definition*}[Permutation Graph]
	For any integer $m \geq 1$, a layered graph $G = (V, E)$ is said to be a \textbf{permutation graph} for $\sigma  \in S_m$ if $\card{\first(G)},\card{\last(G)} \geq m$ and there is a path from $i \in \first_{[m]}(G)$ to $j \in \last_{[m]}(G)$ if and only if $\sigma(i) = j$ for each $i,j \in [m]$.
\end{Definition*}

\begin{Definition*}[Permutation Hiding Graphs; c.f.~\cite{ChenKPSSY21}]
	For  integers $m,n,p,s \geq 1$ and real $\delta \in (0,1)$, we define a \textbf{permutation hiding generator} $\GG=\GG(m,n,p,s,\delta)$ as any family of distributions $\GG: S_m \rightarrow \PP_m$ on permutation graphs satisfying the following two properties:
	\begin{enumerate}[label=$(\roman*)$]
		\item For any $\sigma \in S_m$, any permutation graph $G$ in the support of $\GG(\sigma)$ has $n$ vertices. 
		\item For any $\sigma_1,\sigma_2 \in S_m$, the distribution of graphs $\GG(\sigma_1)$ and $\GG(\sigma_2)$ are $\delta$-indistinguishable 
		for any $p$-pass $s$-space streaming algorithm. 
	\end{enumerate}
\end{Definition*}

We need some further definitions also. Given two layered graphs $G_1 = (V_1, E_1)$ and $G_2 = (V_2, E_2)$, we define their \textbf{concatenation}, $G_1 \circ G_2 = (V_1 \cup V_2, E)$ as follows. Let $\ell_1 = \card{\last(G_2)}$ and $\ell_2 = \card{\first(G_1)}$. The edge set $E$ is made up of $E_1 \cup E_2$ and the identity perfect matching  from the set $\last(G_2)$ to the set $\first_{[\ell_1]}(G_1)$ if $\ell_1 \leq \ell_2$ and identity perfect matching from $\last_{[\ell_2]}(G_2)$ to $\first(G_1)$ if $\ell_1 > \ell_2$. 

Concatenating permutation graphs gives us a graph for the concatenated permutation.

\begin{claim}\label{clm:concatenation}
	Given $G_1, G_2$ which are permutation graphs for $\sigma_1, \sigma_2 \in S_m$ respectively, $G_1 \circ G_2$ is a permutation graph for the permutation $\sigma_1 \circ \sigma_2 \in S_m$.
\end{claim}

\begin{proof}
	For any vertex $i \in \first_{[m]}(G_2)$, there is a path to $\sigma_2(i) \in \last_{[m]}(G_2)$ by \Cref{def:perm-graph}. There is an edge from $\sigma_2(i) \in \last_{[m]}(G_2)$ to $\sigma_2(i) \in \first_{[m]}(G_1)$ by the addition of the identity perfect matching. Again by \Cref{def:perm-graph}, there is a path from $\sigma_2(i) \in \first_{[m]}(G_1)$ to $\sigma_1(\sigma_2(i)) \in \last_{[m]}(G_1)$ for each $i \in [m]$.
	
	Thus, the path from $i \in \first_{[m]}(G_1 \circ G_2)$ to $\sigma_1(\sigma_2(i)) \in \last_{[m]}(G_1 \circ G_2)$ exists for each $i \in [m]$, and no other path from any vertex in $\first_{[m]}(G_1 \circ G_2)$ to $\last_{[m]}(G_1 \circ G_2)$ exists.
\end{proof}

The main lemma of this section follows.
\begin{lemma}\label{lem:1-pass-hiding}
	There exists a permutation hiding generator $\GG:S_m \rightarrow \cD_m$ for 1-pass streaming algorithms using space $s = o(m^{1+\beta/2})$ such that
	\begin{align*}
		n &\leq 2 \cdot \caks \cdot 10^6 \cdot (m/\alpha \beta^2) \\
		\delta &\leq 100\cdot  \caks \cdot(1/\beta) \cdot m^{-5} \cdot (\alpha^{\beta} \cdot \beta^{-1})^{50/\beta},
	\end{align*}
	  where $\alpha, \beta$ are from \Cref{eq:rs-parameters}.
\end{lemma}

In the first subsection, we define the constructs required to describe our permutation hiding graphs and in the next subsection we prove that our graphs can hide simple permutations. In the final subsection, we extend our graphs to hide any general permutation from one pass semi-streaming algorithms.

\subsection{Building Blocks for Permutation Hiding}\label{sec:one-pass-build}

In this section, we will define the constructs needed to describe our permutation hiding graphs.  The main parts are as follows.

\begin{itemize}
	\item \textbf{Group layered graphs, encoded RS graphs, and other helper structures}: these are simple permutation graphs needed for our main construction; 
	\item \textbf{Blocks}: These are permutation graphs that ``encode" a single permutation inside it; 
	\item \textbf{Multi-blocks}: These are permutation graphs that hide the concatenation of several permutations (via concatenation of several blocks).
\end{itemize}

\subsubsection*{Group layered graphs, encoded RS graphs, and other helper structures}

The most basic permutation graph consisting of only two layers is defined first.

\begin{Definition}[Basic Permutation Graph]\label{def:basic-perm}
	Given a permutation $\sigma \in S_m$, we define a graph $G = (L \cup R, E)$ as a \textbf{basic permutation graph} of $\sigma$, denoted by $\basic{\sigma}$ if $\card{L} = \card{R} = m$ and $E = \set{(i, \sigma(i)) \mid i \in [m]}$. 
\end{Definition}

\begin{figure}[h!]
	\centering
	\tikzset{every picture/.style={line width=0.75pt}} 

\begin{tikzpicture}[x=0.75pt,y=0.75pt,yscale=-1,xscale=1]
	
	\draw  [color={rgb, 255:red, 0; green, 0; blue, 0 }  ,draw opacity=1 ][line width=1.5]  (161.05,145.95) .. controls (161,143.27) and (162.96,141.1) .. (165.42,141.1) .. controls (167.88,141.1) and (169.91,143.27) .. (169.96,145.95) .. controls (170.01,148.63) and (168.06,150.8) .. (165.6,150.8) .. controls (163.14,150.8) and (161.1,148.63) .. (161.05,145.95) -- cycle ;
	\draw  [color={rgb, 255:red, 0; green, 0; blue, 0 }  ,draw opacity=1 ][line width=1.5]  (160.6,170.29) .. controls (160.55,167.62) and (162.51,165.45) .. (164.97,165.45) .. controls (167.43,165.45) and (169.46,167.62) .. (169.51,170.29) .. controls (169.56,172.97) and (167.6,175.14) .. (165.14,175.14) .. controls (162.68,175.14) and (160.65,172.97) .. (160.6,170.29) -- cycle ;
	\draw  [color={rgb, 255:red, 0; green, 0; blue, 0 }  ,draw opacity=1 ][line width=1.5]  (160.42,194.53) .. controls (160.37,191.85) and (162.33,189.68) .. (164.79,189.68) .. controls (167.25,189.68) and (169.28,191.85) .. (169.33,194.53) .. controls (169.38,197.21) and (167.43,199.38) .. (164.97,199.38) .. controls (162.51,199.38) and (160.47,197.21) .. (160.42,194.53) -- cycle ;
	\draw  [color={rgb, 255:red, 0; green, 0; blue, 0 }  ,draw opacity=1 ][line width=1.5]  (160.42,219.53) .. controls (160.37,216.85) and (162.33,214.68) .. (164.79,214.68) .. controls (167.25,214.68) and (169.28,216.85) .. (169.33,219.53) .. controls (169.38,222.21) and (167.43,224.38) .. (164.97,224.38) .. controls (162.51,224.38) and (160.47,222.21) .. (160.42,219.53) -- cycle ;
	\draw  [color={rgb, 255:red, 0; green, 0; blue, 0 }  ,draw opacity=1 ][line width=1.5]  (207.05,146.95) .. controls (207,144.27) and (208.96,142.1) .. (211.42,142.1) .. controls (213.88,142.1) and (215.91,144.27) .. (215.96,146.95) .. controls (216.01,149.63) and (214.06,151.8) .. (211.6,151.8) .. controls (209.14,151.8) and (207.1,149.63) .. (207.05,146.95) -- cycle ;
	\draw  [color={rgb, 255:red, 0; green, 0; blue, 0 }  ,draw opacity=1 ][line width=1.5]  (206.6,171.29) .. controls (206.55,168.62) and (208.51,166.45) .. (210.97,166.45) .. controls (213.43,166.45) and (215.46,168.62) .. (215.51,171.29) .. controls (215.56,173.97) and (213.6,176.14) .. (211.14,176.14) .. controls (208.68,176.14) and (206.65,173.97) .. (206.6,171.29) -- cycle ;
	\draw  [color={rgb, 255:red, 0; green, 0; blue, 0 }  ,draw opacity=1 ][line width=1.5]  (206.42,195.53) .. controls (206.37,192.85) and (208.33,190.68) .. (210.79,190.68) .. controls (213.25,190.68) and (215.28,192.85) .. (215.33,195.53) .. controls (215.38,198.21) and (213.43,200.38) .. (210.97,200.38) .. controls (208.51,200.38) and (206.47,198.21) .. (206.42,195.53) -- cycle ;
	\draw  [color={rgb, 255:red, 0; green, 0; blue, 0 }  ,draw opacity=1 ][line width=1.5]  (206.42,220.53) .. controls (206.37,217.85) and (208.33,215.68) .. (210.79,215.68) .. controls (213.25,215.68) and (215.28,217.85) .. (215.33,220.53) .. controls (215.38,223.21) and (213.43,225.38) .. (210.97,225.38) .. controls (208.51,225.38) and (206.47,223.21) .. (206.42,220.53) -- cycle ;
	\draw [color={rgb, 255:red, 0; green, 0; blue, 0 }  ,draw opacity=1 ][fill={rgb, 255:red, 155; green, 155; blue, 155 }  ,fill opacity=1 ][line width=1.5]    (169.96,145.95) -- (203.31,169.02) ;
	\draw [shift={(206.6,171.29)}, rotate = 214.67] [fill={rgb, 255:red, 0; green, 0; blue, 0 }  ,fill opacity=1 ][line width=0.08]  [draw opacity=0] (6.97,-3.35) -- (0,0) -- (6.97,3.35) -- cycle    ;
	\draw [color={rgb, 255:red, 0; green, 0; blue, 0 }  ,draw opacity=1 ][fill={rgb, 255:red, 155; green, 155; blue, 155 }  ,fill opacity=1 ][line width=1.5]    (169.51,170.29) -- (203.12,193.27) ;
	\draw [shift={(206.42,195.53)}, rotate = 214.36] [fill={rgb, 255:red, 0; green, 0; blue, 0 }  ,fill opacity=1 ][line width=0.08]  [draw opacity=0] (6.97,-3.35) -- (0,0) -- (6.97,3.35) -- cycle    ;
	\draw [color={rgb, 255:red, 0; green, 0; blue, 0 }  ,draw opacity=1 ][fill={rgb, 255:red, 155; green, 155; blue, 155 }  ,fill opacity=1 ][line width=1.5]    (169.33,194.53) -- (203.15,218.23) ;
	\draw [shift={(206.42,220.53)}, rotate = 215.03] [fill={rgb, 255:red, 0; green, 0; blue, 0 }  ,fill opacity=1 ][line width=0.08]  [draw opacity=0] (6.97,-3.35) -- (0,0) -- (6.97,3.35) -- cycle    ;
	\draw [color={rgb, 255:red, 0; green, 0; blue, 0 }  ,draw opacity=1 ][fill={rgb, 255:red, 155; green, 155; blue, 155 }  ,fill opacity=1 ][line width=1.5]    (169.33,219.53) -- (205.21,150.5) ;
	\draw [shift={(207.05,146.95)}, rotate = 117.46] [fill={rgb, 255:red, 0; green, 0; blue, 0 }  ,fill opacity=1 ][line width=0.08]  [draw opacity=0] (6.97,-3.35) -- (0,0) -- (6.97,3.35) -- cycle    ;
	
	\draw (154.37,101.44) node [anchor=north west][inner sep=0.75pt]   [align=left] {$\displaystyle V^{1}$};
	\draw (200.37,102.44) node [anchor=north west][inner sep=0.75pt]   [align=left] {$\displaystyle V^{2}$};

\end{tikzpicture}\caption{An illustration of $\basic{\sigma}$ from \Cref{def:basic-perm} with $\sigma = (2,3,4,1)$.}\label{fig:basic-graph}
\end{figure}
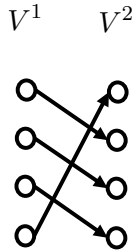
We need to define some more structure on layered graphs.

\begin{Definition}[Group-Layered Graph]\label{def:group-layer}
	For integers $\ww,\dd,\bb \geq 1$, we define a \textbf{group-layered graph} as any directed acyclic graph $G=(V,E)$ satisfying the following: 
	\begin{enumerate}[label=$(\roman*)$]
		\item Vertices of $G$ are partitioned into $d$ equal-size \textbf{layers} $V^1,\ldots,V^d$, each of size $\ww \cdot \bb$. We identify each layer
		with the pairs in $[\ww] \times [\bb]$. 
		\item For any layer $i \in [\dd]$ and $a \in [\ww]$, we define the \textbf{group} $V^{i,a} := \set{a} \times [\bb]$. 
		\item Edges of $G$ can be defined via some tuples $(i,a_1,a_2,\sigma) \in [\dd] \times [\ww] \times [\ww] \times S_\bb$ as follows: 
		We connect $(a_1,j) \in V^i$ to $(a_2,\sigma(j)) \in V^{i+1}$ for all $j \in [\bb]$. 
	\end{enumerate}
	We refer to $\ww$ as the \textbf{width} of the layered graph, to $\dd$ as its \textbf{depth}, and $\bb$ as its \textbf{group size}. 
	We use $\LL_{\ww,\dd,\bb}$ to denote the set of all layered graphs with width, depth, and group size, $\ww,\dd$, and $\bb$, respectively. See \Cref{fig:group-layer} for an illustration. 
\end{Definition}

\begin{figure}[h!]
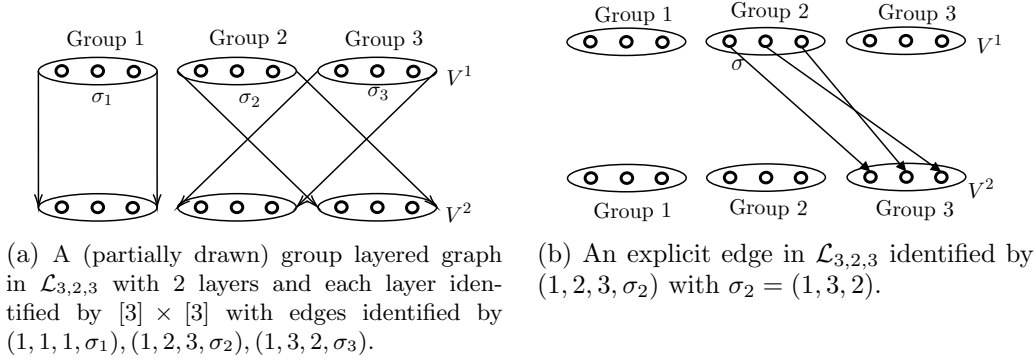

	\centering
	\subcaptionbox{\footnotesize A (partially drawn) group layered graph in $\LL_{3,2,3}$ with 2 layers and each layer identified by $[3] \times [3]$ with edges identified by $(1,1,1,\sigma_1), (1,2,3,\sigma_2), (1,3,2,\sigma_3)$.}%
	[.4\linewidth]{
		\scalebox{0.75}{\input{Figures/grouplayer}}
	} 
	\hspace{0.2cm} 
	\subcaptionbox{An explicit edge in $\LL_{3,2,3}$ identified by $(1,2,3,\sigma_2)$ with $\sigma_2 = (1,3,2)$.}%
	[.4\linewidth]{
		\scalebox{0.75}{\input{Figures/edge-group}}
	}
	\caption{An illustration of a group layered graph from \Cref{def:group-layer}.}
	\label{fig:group-layer}
\end{figure}


While performing concatenation on a group layered graph $G \in \LL_{\ww, \dd, \bb}$ with any other graph $G'$, we treat the vertices in $\first(G)$, indexed by $[\ww] \times [\bb]$, as being indexed by the set $[\ww \cdot \bb]$ by directly mapping any vertex $(i,j)$ to $(i-1) \cdot \bb + j$ for $i \in [\ww], j \in [\bb]$.

Group-layered graphs allow us to capture two main ideas: $(i)$ Reachability between multiple groups; and $(ii)$ Permuting within each group. We formalize this in the following. 


Consider any tuple $(i, a_1, a_2, \sigmaId) \in [d] \times [w] \times [w] \times S_b$.  The edges associated with this tuple simply connect group $V^{i, a_1}$ to $V^{i+1, a_2}$ with the identity permutation. Such edges allow us to \emph{add paths between the groups} of a group-layered graph. 
We formalize this concept next.

\begin{Definition}[Group Permuting Graph]\label{def:permute-groups}
	Given a permutation $\sigma \in S_{\ww}$ and any integer $\bb \geq 1$, we define the \textbf{group permuting graph}, denoted by $\permgroups{\sigma, b} = (V, E)$ as the group layered graph from $\LL_{\ww,2,\bb}$ with the edges $(1, i, \sigma(i), \sigmaId(b))$ for each $i \in [\ww]$ (see \Cref{fig:perm-groups}).
\end{Definition}

\begin{figure}[h!]
	\centering
	\scalebox{0.8}{\input{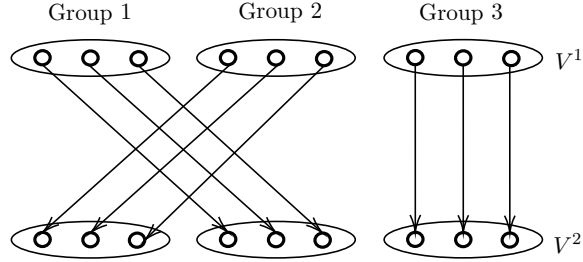}}\caption{A group permuting graph from \Cref{def:permute-groups} with $\sigma = (2,1,3)$ and $b = 3$. }\label{fig:perm-groups}
\end{figure}

The second main idea in group-layered graphs is to permute within the groups. Here, for any edge tuple $(i, a_1, a_2, \sigma)$, if $a_1, a_2$ are fixed, $\sigma \in S_{b}$ allows us to \emph{permute within the groups}. The following definition is an instance of permuting within the groups.

\begin{Definition}[Encoded RS graph]\label{def:encodedRS}
	Given an $(r,t)$-RS-graph $\Grs = (\Lrs \cup \Rrs, \Ers)$, and a permutation matrix $\Sigma \in \paren{S_b}^{t \times r}$, the \textbf{encoded RS graph}, denoted by $\encodedRS(\Grs, \Sigma) = (V, E)$ is a group layered graph from $\LL_{\nrs, 2, \bb}$ with the edges $(1, \leftRS{j}, \rightRS{j}, \sigma_{i,j})$ for edge $j \in \Mrs{i}$, for each $i \in [t], j \in [r]$ (see \Cref{fig:encoded-RS}).
\end{Definition}

\begin{figure}[h!]
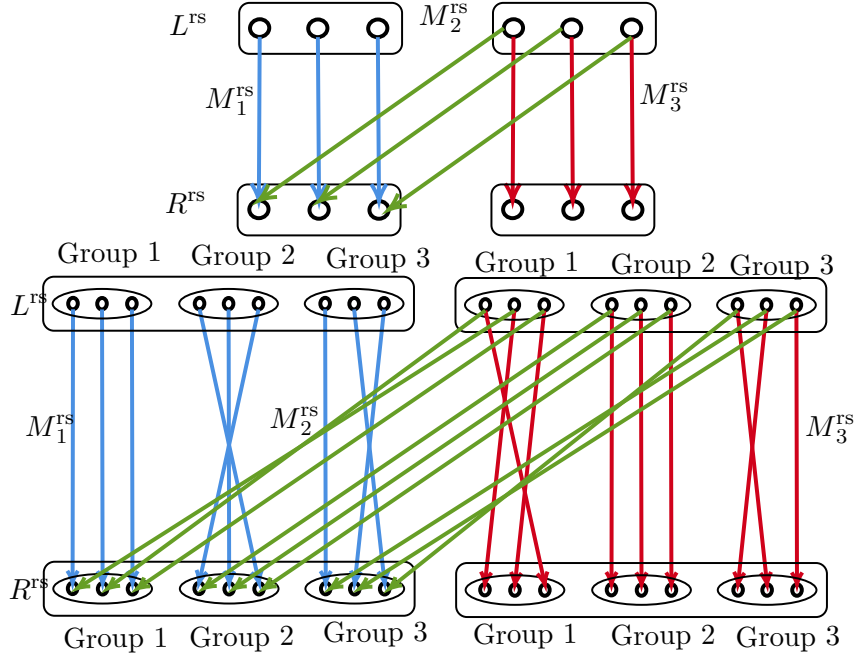

	\centering
	\input{Figures/rsgraph}
	\input{Figures/encodedrs}\caption{An illustration of an encoded $(r,t)$ RS-graph with $r = 3, t = 3$ from \Cref{def:encodedRS} with $\Sigma \in \paren{S_b}^{3 \times 3}$ and $\sigma_{1,1} = (1,2,3), \sigma_{1,2} = (3,2,1), \sigma_{1,3} = (1,3,2)$ as blue edges; $\sigma_{2,1} = (2,1,3), \sigma_{2,2} = (1,2,3), \sigma_{2,3} = (3,1,2)$ as green edges; $\sigma_{3,1} = (3,1,2), \sigma_{3,2} = (1,2,3), \sigma_{3,3} = (2,1,3)$ as red edges.}\label{fig:encoded-RS}

\end{figure}

Observe that if the RS-graph $\Grs$ is fixed, the tuple $(\leftRS{j}, \rightRS{j})$ is fixed for edge $j \in \Mrs{i}$ for each $i \in [t], j \in [r]$. Based on these fixed tuples, we are permuting within the groups based on matrix $\Sigma$. We combine the two main ideas when we construct blocks. Before we proceed, we need a few more definitions. 
The first  is that of a permutation that agrees with a given matching.

\begin{Definition}[Match-Aligned Permutation]\label{def:matching-aligned-perm}
	Given a matching $M$ on $G = (L \cup R, E)$ with $\card{L} = \card{R} = m$, $\sigma \in S_m$ is said to be a \textbf{match-aligned permutation}, denoted by $\extFull{M}$ if it is the lexicographically first permutation with $\sigma(u) = v$ for all~$(u,v) \in~M$. 
\end{Definition}

In other words, we view any matching as a partial permutation where the edges in the matching fix certain assignments of the permutation, and \Cref{def:matching-aligned-perm} extends this to a complete permutation. 
We also need permutations that pick certain edges from a given RS-graph.

\begin{Definition}[Edge Picking Permutations]\label{def:edge-picking-permutations}
	Given an $(r,t)$-RS-graph $\Grs = (\Lrs \cup \Rrs, \Ers)$, an index $\ell \in [t]$ for an induced matching $\Mrs{\ell}$ in $\Grs$ and $(r/2)$ edges $\Estar = (e_1, e_2, \ldots, e_{r/2})$ from matching $\Mrs{\ell}$, define matchings $\Mleft, \Mright$  on vertices $L \cup R$ with $\card{L} = \card{R} = \nrs$ obeying,
	\begin{align*}
		\Mleft= \set{(i, \leftRS{e_i}) \mid i \in [r/2]}  \qquad \textnormal{ and } \qquad \Mright = \set{(\rightRS{e_i}, i) \mid i \in [r/2]}.
	\end{align*} 
	
	Then,   \textbf{edge picking permutations} (see \Cref{fig:edge-pick}), denoted by $\edgepick(\Grs, \ell, \Estar)$ is an ordered pair $(\sigmaLeft, \sigmaRight) \in S_{\nrs} \times S_{\nrs}$, where
	\begin{align*}
		\sigmaLeft = \extFull{\Mleft} \qquad \textnormal{ and } \qquad \sigmaRight = \extFull{\Mright}.
	\end{align*}
\end{Definition}

\begin{figure}[h!]
	\centering
	\scalebox{0.8}{\input{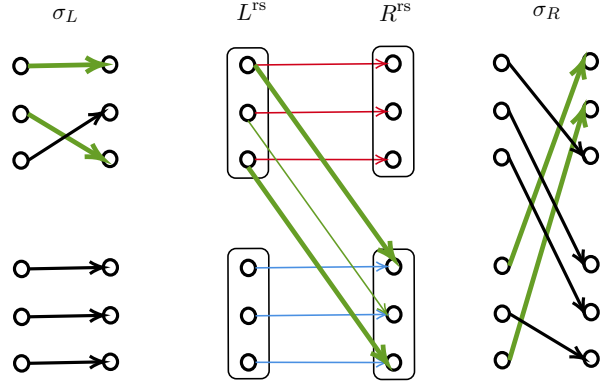}}\caption{An illustration of edge picking permutations from\ \Cref{def:edge-picking-permutations}. The RS-graph from \Cref{fig:encoded-RS} is picked with $\ell = 2$ and (bold) edges $e_1, e_2 \in \Mrs{2}$.  The permutation $\sigmaLeft$ and $\sigmaRight$ are $(1,3,2,4,5,6)$ and $(3,4,5,1,6,2)$, respectively.} \label{fig:edge-pick}
\end{figure}

\subsubsection*{Blocks}
A key gadget in our construction, called a block is built using the preceding definitions. It combines the two main ideas from group layered graphs. 

\begin{Definition}[Block]\label{def:block}
	For any $r, t, b \geq 1$, given
	\begin{itemize}
		\item An $(r,t)$ RS-graph $\Grs = (\Lrs \cup \Rrs, \Ers)$ with $\Lrs = \Rrs = [\nrs]$ and $t$ induced matchings $\Mrs{1}, \ldots \Mrs{t}$,
		\item A permutation matrix $\Sigma \in \paren{S_b}^{t \times r}$,
		\item An index $\ell \in [t]$ and $(r/2)$ edge tuple $\Estar = (e_1, e_2, \ldots, e_{r/2})$ with each edge belonging to  $\Mrs{\ell}$,
	\end{itemize}
	let $(\sigmaLeft, \sigmaRight) \in S_{\nrs} \times S_{\nrs}$ be $\edgepick(\Grs, \ell, \Estar)$. We define a \textbf{block} (see \Cref{fig:block}), denoted by $\block(\Grs, \Sigma, \ell, \Estar)$ as,
	\begin{align*}
		\permgroups{\sigmaRight, b} \circ \encodedRS(\Grs, \Sigma) \circ \permgroups{\sigmaLeft, b}.
	\end{align*}
\end{Definition}

\begin{figure}[h!]
	\centering
	\scalebox{0.75}{\input{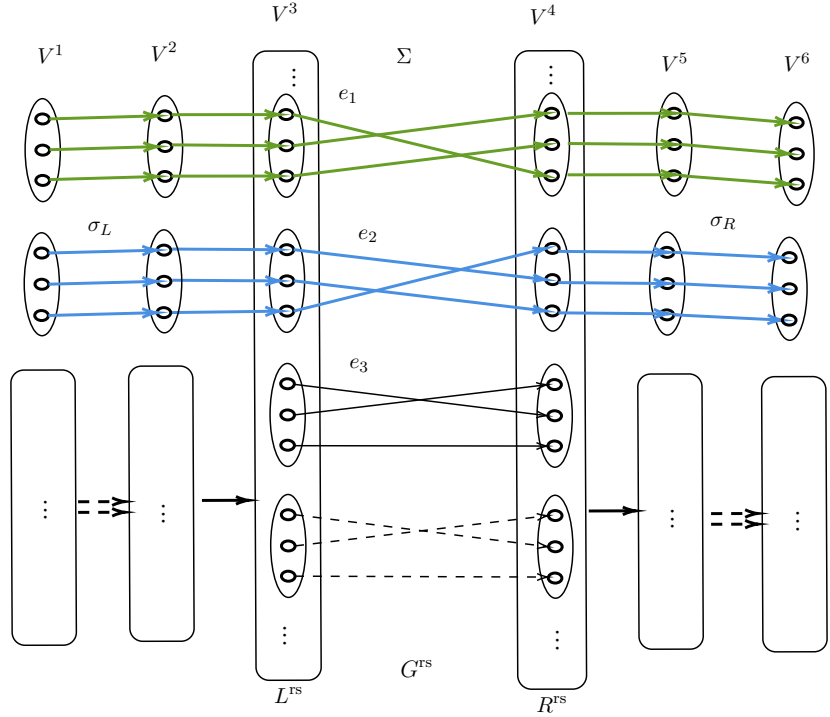}}\caption{An illustration of block from \Cref{def:block} with $r = 3, b = 3$ showing the choosing of one induced matching $\Mrs{\jstar}$ with edges $e_1$ (green) and $e_2$ (blue). The dashed light edge in $\Grs$ represents an edge in another induced matching. The edges from $V^1$ to $V^2$ are from $\permgroups{\sigmaLeft, b}$, the edges from $V^5$ to $V^6$ are from $\permgroups{\sigmaRight, b}$. The edges from $V^3$ to $V^4$ are from $\encodedRS(\Grs, \Sigma)$.}\label{fig:block}
\end{figure}

Let us prove some useful properties of blocks. 

\begin{claim}\label{clm:block-struct}
	The graph $\block(\Grs, \Sigma, \ell, \Estar) $ is a layered graph in $\LL_{\nrs, 6, \bb}$.
\end{claim}

\begin{proof}
	The graph $\block(\Grs, \Sigma, \ell, \Estar)$ is made of concatenating three group layered graphs, $\permgroups{\sigmaRight, b}$, $\encodedRS(\Grs, \Sigma)$, and $\permgroups{\sigmaLeft,b}$. As both $\sigmaLeft, \sigmaRight$ are permutations from $S_{\nrs}$, $\permgroups{\sigmaLeft, b}$ and $\permgroups{\sigmaRight,b}$ both belong to $\LL_{\nrs, 2, \bb}$. By \Cref{def:encodedRS}, as $\encodedRS(\Grs, \Sigma) \in \LL_{\nrs, 2, \bb}$, the concatenation of all three graphs belongs to $\LL_{\nrs, 6, \bb}$. 
\end{proof}

Next, let us see how the edges in a block between various layers are determined by the inputs.

\begin{claim}\label{clm:block-edge-part}
	In any graph $G = \block(\Grs, \Sigma, \ell, \Estar)$ with layers $V^1, V^2, \ldots, V^6$ (see \Cref{clm:block-struct}), 
	\begin{enumerate}[label=$(\roman*)$]
		\item The edges between $V^{3}$ to $V^4$ are determined only by $\Grs $ and $\Sigma$.
		\item The edges from layers $V^{1}$ to $V^2$ and from $V^5$ to $V^6$ are fixed by $\Grs, \ell$ and $\Estar$. 
		\item The other edges from $V^2$ to $V^3$ and from $V^4$ to $V^5$ are fixed perfect matchings.
	\end{enumerate}
\end{claim}

\begin{proof}
	The edges from $V^3$ to $V^4$ are from $\encodedRS(\Grs,\Sigma)$, and hence depend only on $\Grs$ and $\Sigma$. The edges between $V^1$ to $V^2$ and from $V^5$ to $V^6$ are based on permutations $\sigmaLeft$ and $\sigmaRight$ from \Cref{def:edge-picking-permutations}, and depend only on $\Grs, \ell $ and $\Estar$. The other edges are added when concatenating group-layered graphs, and as the layers have same size $[\nrs \cdot \bb]$, they are perfect matchings. 
\end{proof}

Let $\lexpart$ be an equipartition of the set $[m]$ into $m/\bb$ groups of size $\bb$ each, partitioning the elements lexicographically throughout this section. That is, for each $i \in [m/\bb]$, the set $\lexpart_{i} = \set{(i-1)\bb + a \mid a \in [\bb]}$ is a group in partition $\lexpart$. Recall from \Cref{def:simple-perm} that a simple permutation $\rho \in S_m$ on partition $\lexpart$ is such that for all $a \in \lexpart_i$, $\rho(a) \in \lexpart_i$. 	Given $\Sigma, \ell$ and $\Estar = (e_1, e_2, \ldots, e_{r/2})$ with $r/2 = m/\bb$ edges from $\Mrs{\ell}$ as the inputs from \Cref{def:block} , let $\rho$ be the following simple permutation on partition $\lexpart$:
\begin{equation}\label{eq:def-rho}
\forall i \in [m/\bb], a \in [\bb], \rho((i-1)\bb + a) = (i-1)\bb + \sigma_{\ell, e_i}(a). 
\end{equation}

\begin{claim}\label{clm:block-connections}
	$\block(\Grs, \Sigma, \ell, \Estar)$ is a permutation graph for $\rho$ defined in \Cref{eq:def-rho}.
\end{claim}

\begin{proof}
	We know that each vertex $(i,a) \in V^1$ for $i \in [m/\bb], a \in [\bb]$, also indexed by $(i-1)\bb+a \in [m]$ is connected to $(\sigmaLeft(i), a) \in V^2$ by the edge added from $\permgroups{\sigmaLeft}$.  By \Cref{def:edge-picking-permutations}, we know that $\sigmaLeft(i) = \leftRS{e_i}$ for $i \in [r/2]$ and edge $e_i \in \Mrs{\ell}$. The vertex $(\leftRS{e_i}, a) \in V^3$ is connected to $(\rightRS{e_i}, \sigma_{\ell, e_i}(a)) \in V^4$ by the edges from $\encodedRS(\Grs, \Sigma)$. Lastly, the vertex $(\rightRS{e_i}, \sigma_{\ell, e_i}(a)) \in V^5$ is connected to $(i, \sigma_{\ell, e_i}(a)) \in V^6$ based on permutation $\sigmaRight$. The proof is complete when we also look at the identity perfect matchings connecting layers $V^2$ to $V^3$ and layers $V^4$ to $V^5$. 
\end{proof}

\subsubsection*{Multi-blocks}

The next step is to combine multiple blocks to get a larger graph and to simulate the concatenation of multiple simple permutations.

\begin{Definition}[Multi-Block]\label{def:multi-block}
	For any $r, t, b, k \geq 1$, given
	\begin{itemize}
		\item An $(r,t)$ RS-graph $\Grs = (\Lrs \cup \Rrs, \Ers)$ with $\Lrs = \Rrs = [\nrs]$ and $t$ induced matchings $\Mrs{1}, \ldots \Mrs{t}$,
		\item A collection of $k$ permutation matrices $\Sigma  = (\SigmaP{1}, \SigmaP{2},\ldots,  \SigmaP{k}) \in \paren{\paren{S_b}^{t \times r}}^k$,
		\item A tuple $L = (\ell_1, \ell_2, \ldots, \ell_k) \in [t]^k$ and a hypermatching $M $ on the $k$-layered hypergraph $[r]^k$ of size $r/2$,
	\end{itemize}
	let $\Estar_i$ be $(M_{1,i}, M_{2,i}, \ldots, M_{r/2,i})$ be $r/2$ edges in $\Mrs{\ell_i}$ for $i \in [k]$. We define a \textbf{multi-block} (see \Cref{fig:multiblock}), denoted by $\mulblock(\Grs, \Sigma, L, M)$ as,
	\begin{align*}
		\block(\Grs, \SigmaP{1}, \ell_1, \Estar_1) \circ \block(\Grs, \SigmaP{2}, \ell_2, \Estar_2) \circ \ldots \circ \block(\Grs, \SigmaP{k}, \ell_k, \Estar_k).
	\end{align*}
\end{Definition}

\begin{figure}[h!]
	\centering
	\scalebox{0.8}{\input{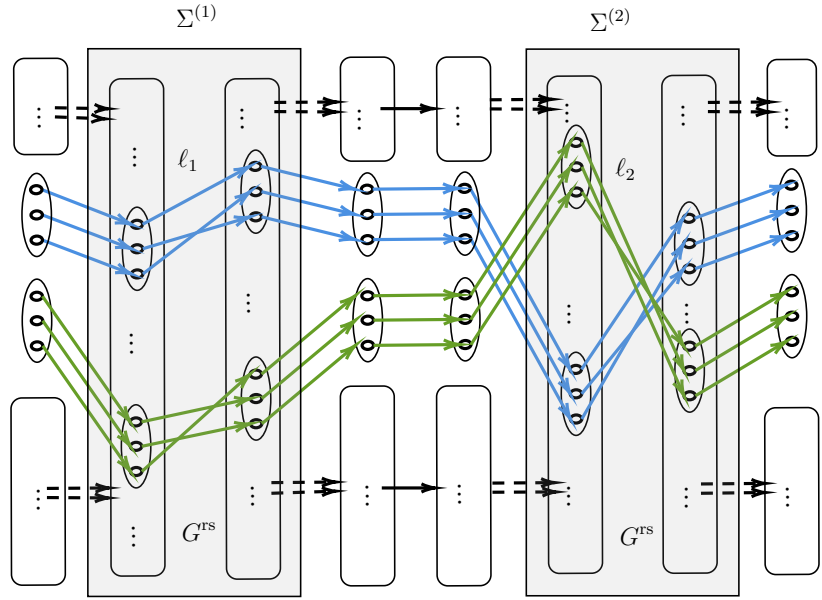}}\caption{An illustration of multi-block from \Cref{def:multi-block} with 2 blocks. The blue path corresponds to one specific edge in each induced matching, and the green corresponds to the other path. }\label{fig:multiblock}
\end{figure}

We can extend the properties we proved about blocks to multi-blocks. 

\begin{claim}\label{clm:mulblock-struct}
	The graph $\mulblock(\Grs, \Sigma, L, M)$ belongs to $\LL_{\nrs, 6k, \bb}$. 
\end{claim}

\begin{proof}
	Follows directly from \Cref{clm:block-struct}, and \Cref{def:multi-block}, as we concatenate $k$ blocks. 
\end{proof}

\begin{claim}\label{clm:mulblock-edge-part}
	In any graph $G = \mulblock(\Grs, \Sigma, L, M)$,
	\begin{enumerate}[label=$(\roman*)$]
		\item For any $i \in[k]$, the edges between layer $V^{(i-1)6 + 3}$ to $V^{(i-1)6+4}$ are determined only by $\Grs, \SigmaP{i}$. 
		\item The edges between $V^{(i-1)6 + 1}$ to $V^{(i-1)6+2}$ and between $V^{(i-1)6+ 5}$ to $V^{6i}$ for any $i \in [k]$ are determined by $\Grs, L, M$.
		\item All the other edges are fixed identity perfect matchings.
	\end{enumerate}
\end{claim}

\begin{proof}
	For any $i \in [k]$, the edges between $V^{(i-1)6 + 3}$ to $V^{(i-1)6+4}$ come from edges between layer $V^3$ to $V^4$ of $\block(\Grs, \SigmaP{i}, \ell_i, \Estar_i)$, and by \Cref{clm:block-edge-part}, are determined only by $\Grs, \SigmaP{i}$. The edges between $V^{(i-1)6 + 1}$ to $V^{(i-1)6+2}$ and from $V^{(i-1)6+5}$ to $V^{6i}$ are determined by $\Grs, \ell_i, \Estar_i$ for $i \in [k]$, again by \Cref{clm:block-edge-part}, and thus depend only on $\Grs, L, M$. Finally, in the concatenations, we add identity-perfect matchings between the layers. 
\end{proof}

For $i \in [k]$, $\SigmaP{i}, \ell_i$ and $\Estar_i = (e_1, e_2, \ldots, e_{r/2})$ as in \Cref{def:multi-block} , let $\rho_i \in S_m$ be the following simple permutation on partition $\lexpart$:
\begin{equation}\label{eq:def-rho-is}
\forall j \in [m/\bb], a \in [\bb], \rho_i((j-1)\bb + a) = (j-1)\bb + \sigma_{\ell_i, e_j}(a).
\end{equation}
Let $\gammastar \in S_m$ be another simple permutation on $\lexpart$ defined as,
\begin{equation}\label{eq:def-gammastar}
	\gammastar = \rho_1 \circ \rho_2 \circ \ldots \circ \rho_k. 
\end{equation}

\begin{lemma}\label{lem:mulblock-concat}
	$\mulblock(\Grs, \Sigma, L, M)$ is a permutation graph of $\gammastar$. 
\end{lemma}

\begin{proof}
	By \Cref{clm:block-connections}, we know that $\block(\Grs, \SigmaP{i}, \ell_i, \Estar_i)$ is a permutation graph for $\rho_i \in S_m$ for each $i \in[k]$. By \Cref{clm:concatenation}, $\mulblock(\Grs, \Sigma, L, M)$ is a permutation graph for $\gammastar$. 
\end{proof}

This concludes our subsection for describing the building blocks and constructs needed for our permutation hiding generators. 

\subsection{Simple Permutation Hiding in One Pass}

In this section, we will construct permutation hiding generators for simple permutations. Let us start by defining them.

\begin{Definition}\label{def:simple-perm-hider}
	For any integers $ n, p, s\geq 1$, partition $\cP$ of $[m]$ into $m/b \geq 1$ blocks of size $b \geq 1$, and error parameter $\delta \in [0,1]$, a simple permutation hiding generator $\GGsim(n,p,s, \cP, \delta)$ is defined as a function  $\GGsim$ from the set of all simple permutations under partition $\cP$ to $\cD_m$ such that,
	\begin{enumerate}[label=$(\roman*)$]
		\item For any $\rho \in S_m$ which is simple under $\cP$, any $G \sim \GGsim(\rho)$ is a permutation graph for $\rho$ with $n$ vertices.
		\item For any $\rho_1, \rho_2 \in S_m$, both simple under $\cP$,  the two distributions $\GGsim(\rho_1)$ and $\GGsim(\rho_2)$ are $\delta$-indistinguishable for $p$-pass $s$-space streaming algorithms. 
	\end{enumerate}
\end{Definition}
In this subsection, we will prove that such generators can be constructed for the partition $\lexpart$.

\begin{lemma}\label{lem:simple-perm-hide}
	For any $b \geq 2$, there is a simple permutation hiding generator $\GGsim: \Ssim_m \rightarrow \cD_m$ with respect to partition $\lexpart$ for $1$-pass streaming algorithms using space $s = o(\paren{\frac{m}{b}}^{1+2\beta/3}) $ with 
	parameters $n = (2\cdot 10^4 \cdot m/\alpha \beta)$ and $\delta =  (2m/b)^{-6} \cdot (\alpha^{\beta} \cdot \beta^{-1})^{50/\beta}$.
\end{lemma}

We will construct these generators using multi-blocks, and the inputs to the multi-blocks are from instances of $\IndexHPH$.  By \Cref{clm:simple-perm-to-perm-vec}, it is sufficient to hide multiple smaller permutations from $S_b$ ($m/b$ of them, to be exact) to hide any simple permutation.  We will be talking about one instance of \Cref{prob:mph} in the construction, and we will use $\Sigma \in \paren{\paren{S_b}^{t \times r}}^{k}, L \in [t]^{k}$, hypermatching $M \subset [r]^k$ and $\Gammastar, \GammaY, \GammaN , \Gamma \in \paren{S_b}^{r/2}$ to denote the variables in the instance. See \Cref{fig:simple-perm-1-pass} for an illustration of this construction. 

\begin{ourbox}
	\textbf{Construction of $\GGsim(n,p,s,\lexpart, \delta) : \Ssim_m \rightarrow \cD_m$ on input $\rho \in \Ssim$:}
	\begin{enumerate}[label=$(\roman*)$]
		\item Fix an $(r,t)$-RS graph $\Grs$ with $r = 2m/b = \nrs/\alpha$ and $t = \paren{\nrs}^{\beta}$.
		\item Instantiate $k$ players and the referee in $\IndexHPH_{r,t,b,k}$ (\Cref{prob:mph}) with parameter $k = 1600/\beta$ such that the solution to the instance, $\Gammastar \circ \Gamma$ is $\tovec(\rho)$.
		\item Output $\mulblock(\Grs, \Sigma, L, M) \circ \basic{\joinit(\Gamma)}$.
	\end{enumerate}
\end{ourbox}

\begin{figure}[h!]
	\centering
	\scalebox{0.75}{\input{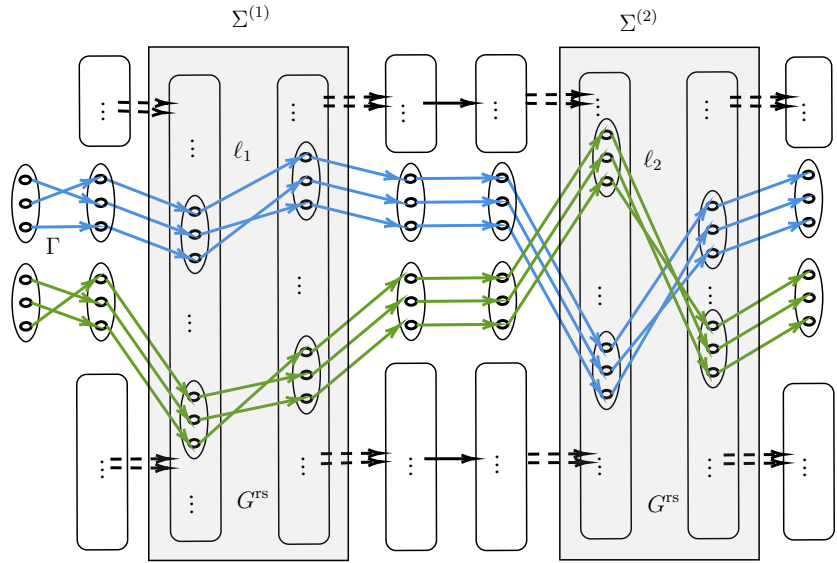}}\caption{An illustration of a simple permutation hiding generator with $k = 2$. Compared to~\Cref{fig:multiblock}, the basic permutation graph of $\joinit(\Gamma)$ is added to the left.}\label{fig:simple-perm-1-pass}
\end{figure}

Now let us show that our construction possesses the required properties. We begin by proving that it is indeed a valid permutation graph for $\rho$ with the required number of vertices.

\begin{claim}\label{clm:valid}
	For any $\rho \in \Ssim_m$, the graph $\GGsim(\rho)$ is a permutation graph for $\rho$ with $n \leq  2 \cdot 10^4 \cdot m /\alpha \beta$. 
\end{claim}

\begin{proof}
	We know that $\mulblock(\Grs, \Sigma, L, M)$ is a permutation graph for $\joinit(\Gammastar)$, where $\Gammastar$ is defined as in \Cref{prob:mph} by \Cref{lem:mulblock-concat}. By \Cref{clm:concatenation}, the output $\GGsim(\rho)$ is a permutation graph for $\rho$, as $\joinit(\Gammastar) \circ \joinit(\Gamma) = \rho$ by construction. 
	
	The total number of vertices in $\mulblock(\Grs, \Sigma, L, M)$ is $6k\cdot \nrs \cdot b$ by \Cref{clm:mulblock-struct}. When we concatenate with $\basic{\joinit(\Gamma)}$, we add $2m$ vertices, so the total number is $6k \cdot \nrs \cdot b + 2m = 6 \cdot \frac{1600}{\beta} \cdot \frac{2m}{b\alpha} \cdot b  + 2m \leq 2 \cdot 10^4 \cdot m/ \alpha \beta$.  
\end{proof}

Next, we will show that any 1-pass streaming algorithm can be run by the players of $\IndexHPH$ and the referee on the graph $\GGsim(\rho)$.
\begin{claim}\label{clm:running-alg}
	For any $\rho \in \Ssim_m$ and given any 1-pass $s$-space streaming algorithm $\alg$ using space at most $s$, the players and the referee of $\IndexHPH$ (\Cref{prob:mph}) can run $\alg$ on $\GGsim(\rho)$ using at most $s$ bits of communication per player.
\end{claim}

\begin{proof}
	
	We know that the edges between layers $V^{6(i-1) + 3}$ to $V^{6(i-1) + 4}$ depend only on $\SigmaP{i}$ for $i \in [k]$ in $\mulblock(\Grs, \Sigma, L, M)$ and all other edges are known to the referee, as they either depend on $L, M$ or are fixed by \Cref{clm:mulblock-edge-part}. The players and the referee run $\alg$ on $G = \GGsim(\rho)$ as follows. 
	\begin{enumerate}[label=$(\roman*)$]
		\item Player $\QP{1}$ runs $\alg$ on the edges fixed by $\SigmaP{1}$  in $G$ and writes the memory state on the board. 
		\item In increasing order of $i \in [k] \setminus \set{1}$, player $\QP{i}$ gets the memory state of $\alg $ as written on the board by player $\QP{i-1}$ and continues to run $\alg$ on the edges based on $\SigmaP{i}$ in $G$. Player $\QP{i}$, then writes the memory state back on the board.
		\item The referee gets the state of $\alg $ as written by $\QP{k}$, and adds all the other edges based on $L, M$ and $\Gamma$ to get the final output.
	\end{enumerate}
	Each player writes on the board exactly once, and the communication per player is $s$ bits. 
\end{proof}

We will conclude this subsection by proving \Cref{lem:simple-perm-hide}.

\begin{proof}[Proof of \Cref{lem:simple-perm-hide}]
	The bound on the number of vertices follows readily from \Cref{clm:valid}.
	Let us assume that there is a 1-pass $s$-space streaming algorithm $\alg $ which distinguishes between $\GGsim(\rho_1)$ and $\GGsim(\rho_2)$ for some $\rho_1, \rho_2 \in \Ssim_m$ with $s=o((m/b)^{1+2\beta/3})$ and advantage $\delta$ more than $(2m/b)^{-6}$. We will argue that it can be used to solve an instance of $\IndexHPH$. 
	
	Create an instance of $\IndexHPH_{r,t, b , k}$  with $\GammaY = \tovec(\rho_1)$ and $\GammaN = \tovec(\rho_2)$. We know from \Cref{clm:running-alg} that the referee and players can run $\alg$ using at most $s = o((m/b)^{1+2\beta/3}) = o(r^{1+2\beta/3})$ bits of communication per player.
	The total number of bits is $k \cdot s = O(1/\beta) \cdot  o(r^{1+2\beta/3})= o(rt)$, as $t = (r/\alpha)^{\beta}$.
	By \Cref{thm:mph}, we know that the advantage the referee gains is at most
	\begin{align*}
		r \cdot O(\frac{k \cdot s}{r\cdot t})^{k/32} = r \cdot o\paren{\frac{r^{1+2\beta/3} \cdot \alpha^{\beta}}{r^{1+\beta} \cdot \beta}}^{50/\beta} = o(1/r^6) \cdot (\alpha^{\beta} \cdot \beta^{-1})^{50/\beta},
	\end{align*}
	which is a contradiction. 
\end{proof}

\subsection{General Permutation Hiding in One Pass}

In this subsection, we will hide any general permutation from 1-pass streaming algorithms by hiding multiple simple permutations constructed in \Cref{lem:simple-perm-hide} and prove \Cref{lem:1-pass-hiding}.  We first discuss how to divide any general $\rho \in S_m$ into multiple simple permutations with sorting networks. 

The simple permutations we hid in \Cref{lem:simple-perm-hide} were all under the fixed partition $\lexpart$. It is easy to hide them even if they were under different partitions, as we will show. 

\begin{lemma}\label{lem:general-partitions}
	Given any permutation $\rho \in S_m$ which is simple under partition $\cP$ with $ m/b$ groups of size $b$, there is a simple permutation hiding generator for partition $\cP$ with the same parameters as in \Cref{lem:simple-perm-hide}.
\end{lemma}
\begin{proof}
	We permute the elements of $[m]$ based on $\cP$ so that the partition will be $\lexpart$. To do so, let
	$f_{i}: P_i \rightarrow [b]$ be the lexicographic bijective mapping from group $P_i$ to $[b]$ for $i \in [m/b]$, and let $g : [m] \rightarrow [m/b]$ be such that $g(j)$ denotes which group among the $m/b$ groups $j \in [m]$ belongs to under $\cP$.  Define permutation $\swap \in  S_m $ as,
	\begin{align*}
		\swap(j) = (g(j)-1)\cdot b +  f_{g(j)}(j)
	\end{align*}
	for $j \in [m]$.  Define permutation $\rho' \in S_m$ as,
	\begin{align*}
		\rho' = \swap \circ \rho \circ \swap^{-1}
	\end{align*}
	
	\begin{claim}\label{clm:local-maps}
		The permutation $\rho'$ is simple under partition $\lexpart$.
	\end{claim}
	\begin{proof}
		Let $h:[m] \rightarrow [m/b]$ be such that $h(j) = \floor{\frac{j-1}{b}} + 1$. This denotes the group each element belongs to under partition $\lexpart$.
		We know that for any $x \in P_{g(x)}$, $\swap(x) \in \lexpart_{g(x)}$, and for any $y \in \lexpart_{h(y)}$, $\swap^{-1}(y) \in P_{h(y)}$ by the definition of $\swap$.
		Thus, for any $y \in \lexpart_{h(y)}$, we know that $\rho(\swap^{-1}(y)) \in P_{h(y)}$ as $\rho$ is simple under $\cP$. This implies that $\swap(\rho(\swap^{-1}(y))) \in \lexpart_{h(y)}$.
	\end{proof}
	
	Now we can easily construct permutation graphs for $\rho$, by sampling $G \sim \GGsim(\rho')$ under partition $\lexpart $ from \Cref{lem:simple-perm-hide}, and outputting
	\begin{align*}
		\basic{\swap} \circ G \circ \basic{\swap^{-1}}.
	\end{align*}
	We add $4m$ vertices to $G$, and the output is a permutation graph for $\rho$ by \Cref{clm:concatenation}.
\end{proof}
Refer to \Cref{fig:swap-perm-sorting} for an illustration of the property proved by \Cref{lem:general-partitions}.
Now we are ready to construct our generators.

\begin{figure}[h!]
		\centering
	\subcaptionbox{\footnotesize The permutation $\rho = (3,4,1,2)$ under partition $\cP$ with $P_1 = \{1,3\}$ and $P_2 = \{2,4\}$, and permutation $\swap = (1,3,2,4)$ based on $\cP, \lexpart$. }%
	[.4\linewidth]{
		\scalebox{0.75}{\input{Figures/gene-parts}}
	} 
	\hspace{0.2cm} 
	\subcaptionbox{An illustration of permutation $\rho' = \swap^{-1} \circ \rho \circ \swap  = (2,1,4,3)$ which is simple under partition $\lexpart$. }%
	[.4\linewidth]{
		\scalebox{0.75}{\input{Figures/gene-parts-2}}
	}
	\caption{An illustration of \Cref{lem:general-partitions}.}
\label{fig:swap-perm-sorting}
\end{figure}
\begin{ourbox}
	\textbf{Construction of family $\GG(m,n,1,s,\delta): S_m \rightarrow \PP_m$ on input $\rho \in S_m$}
	\begin{enumerate}[label=$(\roman*)$]
		\item Get $\daks = \caks \cdot \log_b(m) = 100\caks/\beta $ simple permutations  $\gamma_1, \gamma_2, \ldots, \gamma_{\daks} \in \paren{S_m}$ under partitions $\cP_1, \cP_2, \ldots \cP_{\daks}$ from \Cref{prop:AKSsorting}
		for $b = m^{\beta/100}$. 
		\item Sample $G_i \sim \GGsim(\gamma_i)$ for $i \in [\daks]$ as in \Cref{lem:general-partitions} and output $G = G_1 \circ G_2 \circ \ldots \circ G_{\daks}$.
	\end{enumerate}
\end{ourbox}

\begin{proof}[Proof of \Cref{lem:1-pass-hiding}]
	We know that $G$ is a permutation graph for $\rho $ by \Cref{lem:general-partitions} and \Cref{clm:concatenation}. The total number of vertices in $G$ is, 
	\begin{align*}
		\daks \cdot 2 \cdot 10^4 \cdot m \cdot \frac{1}{\alpha \beta} = 2 \cdot \caks \cdot 10^6 \cdot m \cdot \frac{1}{\alpha\beta}.
	\end{align*}
	Let $\alg$ be a 1-pass streaming algorithm that distinguishes between graphs sampled from $\GG(\rho_1)$ and $\GG(\rho_2)$ using space $s = o(m^{1+\beta/2})$ and advantage more than $m^{-5}$. That is, 
	\begin{align*}
		\tvd{\mem{1}{\alg}(\GG(\rho_1))}{\mem{1}{\alg}(\GG(\rho_2))} \geq m^{-5}.
	\end{align*}
	
	Let $\gamma_{1,i}, \gamma_{2,i}$ be the simple permutations under partition $\cP_i$ from step $(i)$ of the construction for $i \in [\daks]$ for $\rho_1,\rho_2$ respectively. For any $\rho \in S_m$ distribution $\GG(\rho)$ is fixed by sampling from $\GGsim(\gamma_i)$ for all $i \in [\daks]$. By \Cref{fact:tvd-data-processing}, we know that, 
\begin{align*}
		&\tvd{\memn(\GG(\rho_1))}{\memn(\GG(\rho_2))}  \\
		&\leq \tvd{\memn(\GGsim(\gamma_{1,1}), \ldots, \GGsim(\gamma_{1,\daks}))}{\memn(\GGsim(\gamma_{2,1}), \ldots, \GGsim(\gamma_{2,\daks}))} \\
		&\leq \sum_{i=1}^{\daks} \tvd{\memn(\GGsim(\gamma_{1,i}))}{\memn(\GGsim(\gamma_{2,i}))} \tag{by the hybrid argument \Cref{prop:hybrid-arg}} \\
		&\leq \daks \cdot (r)^{-6} (\alpha^{\beta} \cdot \beta^{-1})^{50/\beta} \leq (100 \caks /\beta) \cdot m^{-5} \cdot (\alpha^{\beta} \cdot \beta^{-1})^{50/\beta},
	\end{align*}
	where in the last inequality, we have used \Cref{lem:simple-perm-hide} because algorithm $\alg$ uses $s = o(m^{1+\beta/2}) = o(\paren{\frac{2m}{b}}^{1+2\beta/3})$ space. The proof follows as for our choice of $b = m^{\beta/100}$, $r^{-6} < m^{-5}$.
\end{proof}

\newcommand{\pblock}[1]{\ensuremath{\textnormal{\textsf{p-Block}}(#1)}}
\newcommand{\pmultiblock}[1]{\ensuremath{\textnormal{\textsf{p-Multi-Block}}(#1)}}
\newcommand{\Gammaspec}{\ensuremath{\Gamma^*}}
\newcommand{\Gbar}{\ensuremath{\overline{G}}}
\newcommand{\Lbar}{\ensuremath{\overline{L}}}
\newcommand{\Gammabar}{\ensuremath{\overline{\Gamma}}}
\newcommand{\Mbar}{\ensuremath{\overline{M}}}
\renewcommand{\protFake}{\ensuremath{\Prot^{\textnormal{\textsf{fake}}}}}
\newcommand{\cB}{\ensuremath{\mathcal{B}}}

\renewcommand{\protFake}{\ensuremath{\Prot^{\textnormal{\textsf{fake}}}}}
\newcommand{\Grest}{\ensuremath{\GG^{\textnormal{\textsf{rest}}}}}
\newcommand{\Gfixed}{\ensuremath{G^{\textnormal{\textsf{fixed}}}}}
\newcommand{\Grestbar}{\ensuremath{\overline{\GG^{\textnormal{\textsf{rest}}}}}}

\section{Multi-pass Permutation Hiding}\label{sec:multi}

In this section, we will construct permutation hiding graphs against $p$-pass semi-streaming algorithms for any general integer $p \geq 1$ with induction. We restate \Cref{thm:perm-hiding} here and prove it in this section.

\begin{theorem*}[Restatement of \Cref{thm:perm-hiding}]
	There exists a permutation hiding generator $\GG(m,n,p,s,\delta)$ where $p$ is any positive integer with the following parameters:
	\begin{itemize}
		\item Space $s = o(m^{1+\beta/2})$,
		\item Number of vertices $n = \Theta(1/\alpha \beta^2)^p \cdot m$,
		\item Error parameter $\delta = (p/\beta)^{\Theta(1)/\beta}\cdot \Theta(1/\beta)^{2p} \cdot 1/\poly{(m)}$.
	\end{itemize} 
\end{theorem*}

Our overall strategy is the same as that of \Cref{lem:1-pass-hiding}, with the key difference being that we hide simple permutations from $p$-pass streaming algorithms instead by using permutation hiding generators for $p-1$-pass streaming algorithms. Our inductive hypothesis is as follows.

\begin{assumption}[\textbf{Inductive Hypothesis}]\label{asmp:p-1-pass}
	When considering $p$ pass algorithms, there exists a permutation hiding generator $\GG(m,n,p-1,s,\delta)$ for the following parameters:
	\begin{align*}
		&n = \paren{\frac{2.5\cdot \caks \cdot 10^6}{\alpha \cdot \beta^2}}^{p-1} \cdot m \tag{number of vertices}; \\
		&s := o(m^{1+\beta/2}) \tag{space of streaming algorithm}; \\
		&\delta :=   ((p-1)/\beta)^{50/\beta} \cdot ((4\cdot 10^6 \cdot \caks/\beta^2)^{(p-1)})  \cdot m^{-5} \cdot \alpha. \tag{probability of success of the algorithm},
	\end{align*}
\end{assumption}

\paragraph{Notation.}
We let $\GG_{p-1}$ refer to the permutation hiding generator in \Cref{asmp:p-1-pass}. We use $N_{p-1}(m)$ and $\Delta_{p-1}(m)$ to denote the number of vertices and the probability of success for $\GG_{p-1}$ respectively when the size of the permutation is $m$.

The first step again is to hide simple permutations.
For the rest of this section, we will define some more constructs to hide simple permutations, and then prove \Cref{thm:perm-hiding}.

\subsection{Building Blocks for Multi-Pass Hiding}

In this subsection, we will adapt the definitions of blocks and multi-blocks based on the number of passes of the streaming algorithm (denoted by $p$) which we want to guard against. 

First, we want to extend any permutation $\sigma \in S_{\nrs}$ to $S_{\nrs \cdot b}$ in a specific way. 

\begin{Definition}[Extended Permutation]\label{def:ext-perm}
	Given a permutation $\sigma \in S_{\nrs}$, and an integer $b$, the extended permutation $\sigma' \in S_{\nrs \cdot b}$, denoted by $\extend(\sigma, b) $ is defined as,
	\begin{align*}
		\sigma'((x-1) \cdot b + j) = (\sigma(x)-1) \cdot b + j 
	\end{align*}
	for all $x \in [\nrs], j \in [b-1]$. 
\end{Definition}

Informally, we picked the lexicographic partition of $[\nrs \cdot b]$ into $\nrs$ groups of size $b$, and used $\sigma$ to permute the groups among each other. Permutation $\sigma'$ does not permute the elements within each of the groups. 
The following observation connects extended permutations to one of the primitives we constructed in \Cref{sec:one-pass-build}.

\begin{observation}\label{obs:ext-perm-graphs}
	For any $\sigma \in S_{\nrs}$ and integer $b \geq 1$, the graph $G = \permgroups{\sigma, b}$ from \Cref{def:permute-groups} is a permutation graph for $\extend(\sigma, b) \in S_{\nrs \cdot b}$. 
\end{observation}

\begin{proof}
	Let $\rho = \extend(\sigma, b) \in S_{\nrs \cdot b}$.
	For any $i \in [\nrs], j \in [b]$ element $(i-1)b + j \in [\nrs \cdot b]$, we know that $\rho((i-1)b+j) = \sigma(i) + j$. In the graph $\permgroups{\sigma, b}$ (with layers $V^1$ and $V^2$), the vertex identified by $(i-1)b+j$ in $V^1$ is connected to $(\sigma(i)-1)b + j$ by definition. 
\end{proof}

\begin{Definition}[$p$-Pass Block Distribution]\label{def:p-pass-block-dist}
	For any $r,t,b \geq 1$, and inputs $(r,t)$-RS graph $\Grs$, permutation matrix $\Sigma \in \paren{S_b}^{t \times r}$, index $\ell \in [t]$ and edge tuple $\Estar = (e_1, e_2, \ldots e_{r/2})$,  let:
	\begin{align*}
		(\sigmaLeft, \sigmaRight) &= \edgepick(\Grs, \ell, \Estar).
	\end{align*}
	be permutations in $S_{\nrs}$ as similar to \Cref{def:block}. We extend them to $\sigma_1, \sigma_2 \in S_{\nrs \cdot b}$ as:
	\begin{align*}
		\sigma_1 = \extend(\sigmaLeft, b) \qquad\textnormal{ and }\qquad \sigma_2 = \extend(\sigmaRight, b).
	\end{align*}
	
	We define the \textbf{$p$-pass block distribution}, denoted by $\pblock{\Grs, \Sigma, \jstar, \Estar}$ as follows. (See \Cref{fig:p-block} for an illustration.)
	\begin{enumerate}[label=$(\roman*)$]
		\item Let $G^{\Sigma} $ be the encoded RS-graph of $\Grs$ on $\Sigma$ from \Cref{def:encodedRS}. 
		\item Sample $G_L \sim \GG_{p-1}(\sigma_1)$ and $G_R \sim \GG_{p-1}(\sigma_2)$ from \Cref{asmp:p-1-pass}. 
		\item Output the graph $G  = G_R \conc G^{\Sigma} \conc G_L$. 
	\end{enumerate} 
\end{Definition}

\begin{figure}[h!]
	\centering
	\input{Figures/ppassblock}\caption{An illustration of a graph sampled from $\pblock{\Grs, \Sigma, \jstar, \Estar}$.}\label{fig:p-block}
\end{figure}

Let us compare \Cref{def:p-pass-block-dist} to that of blocks in \Cref{def:block}.  In a block, we concatenate basic permutation graphs for $\extend(\sigmaLeft, b)$ and $\extend(\sigmaRight, b)$ to either side of the encoded RS-graph based on $\Sigma$, by \Cref{obs:ext-perm-graphs}. However, in a $p$-pass block distribution, we \emph{hide} the two permutations $\extend(\sigmaLeft, b)$ and $\extend(\sigmaRight, b)$ from $S_{\nrs \cdot b}$ from $p-1$ pass streaming algorithms inductively. This is the key change that allows us to hide our permutations from $p$-pass algorithms.  

Let us show that graphs sampled from $p$-pass block distributions are also valid permutation graphs for some specific permutations, similar to blocks.

\begin{claim}\label{clm:p-pass-block-valid}
	Any graph $G \sim \pblock{\Grs, \Sigma, \jstar, \Estar}$ is a permutation graph for $\rho \in S_m$ defined by \Cref{eq:def-rho} with $2\nrs \cdot b + 2N_{p-1}(\nrs \cdot b)$ vertices.
\end{claim}

\begin{proof}
	For any graph $G$ sampled from $\pblock{\Grs, \Sigma, \jstar, \Estar}$, \Cref{clm:block-connections} applies because $G$ is also a concatenation of permutation graphs of the same permutations, by \Cref{obs:ext-perm-graphs}. The bound on the total number of vertices follows from \Cref{def:encodedRS} and \Cref{asmp:p-1-pass}.
\end{proof}

We can show that two different $p$-pass block distributions sharing the input $\Sigma$ cannot be distinguished by $(p-1)$-pass streaming algorithms with a high advantage, based on the guarantees in \Cref{asmp:p-1-pass}.

\begin{claim}\label{clm:p-block-p-1-pass}
	For any $\ell_1, \ell_2 \in [t]$ and two edge tuples $\Estar_1, \Estar_2$ from $\Mrs{\ell_1}$ and $\Mrs{\ell_2}$ respectively,  
	the two distributions $\cD_1 = \pblock{\Grs, \Sigma, \ell_1, \Estar_1}$ and $\cD_2 = \pblock{\Grs, \Sigma, \ell_2, \Estar_2}$ are $2\Delta_{p-1}(\nrs \cdot b)$-indistinguishable for $(p-1)$-pass streaming algorithms using space at most $o((\nrs \cdot b)^{1+\beta/2})$. 
\end{claim}
\begin{proof}
	Let $A$ be a $p-1$ pass streaming algorithm using space $o((\nrs \cdot b)^{1+\beta/2})$. Define the following permutations:
	
	\begin{minipage}{0.45\textwidth}
		\begin{align*}
			(\sigmaLeft^1, \sigmaRight^1) &= \edgepick(\Grs, \ell_1, \Estar_1) \\
			\tau_{1,1} &= \extend(\sigmaLeft^1, b) \\
			\tau_{1,2} &= \extend(\sigmaRight^1, b)
		\end{align*}
	\end{minipage}
	\begin{minipage}{0.45\textwidth}
		\begin{align*}
			(\sigmaLeft^2, \sigmaRight^2) &= \edgepick(\Grs, \ell_2, \Estar_2) \\
			\tau_{2,1} &= \extend(\sigmaLeft^2, b) \\
			\tau_{2,2} &= \extend(\sigmaRight^2, b)
		\end{align*}
	\end{minipage}
	
	where $\tau_{1,1}, \tau_{1,2}, \tau_{2,1}$ and $\tau_{2,2}$ belong to $S_{\nrs \cdot b}$.
	Distribution $\cD_i$ is fixed based on samples from $\GG_{p-1}(\tau_{i,1})$ and $\GG_{p-1}(\tau_{i,2})$ for $i = 1,2$. By \Cref{fact:tvd-data-processing}, we have that,
	.	\begin{align*}
		\tvd{\mem{p-1}{A}(\cD_1)}{\mem{p-1}{A}(\cD_2)} &\leq \tvd{(\GG_{p-1}(\tau_{1,1}), \GG_{p-1}(\tau_{1,2}))}{(\GG_{p-1}(\tau_{2,1}), \GG_{p-1}(\tau_{2,2}))} \\
		&  \leq \sum_{j=1,2} \tvd{\GG_{p-1}(\tau_{1,j})}{\GG_{p-1}(\tau_{2,j})} \tag{by \Cref{prop:hybrid-arg}} \\
		&\leq 2 \Delta_{p-1}(\nrs\cdot b). \tag{by \Cref{asmp:p-1-pass}} \qedhere
	\end{align*}
\end{proof}

We can extend \Cref{def:multi-block} to distributed $p$-pass multi-blocks as follows. 

\begin{Definition}[$p$-pass Multi-Block Distribution]\label{def:dist-multi-block}
	Given inputs $\Grs, \Sigma \in \paren{S_b}^{t \times r}, L \in [t]^k$ a hypermatching $ M \subset [r]^k$ of size $r/2$ with $\Estar_i$ defined similar to \Cref{def:multi-block} for $i \in [k]$, the $p$-pass multi-block distribution, denoted by $\pmultiblock{\Grs, \Sigma, L, M}$ is defined as,
	\begin{enumerate}[label=$(\roman*)$]
		\item For each $i \in [k]$, sample graph $G_i \sim \pblock{\Grs, \SigmaP{i}, \ell_i, \Estar_i}$. 
		\item Output $G_1 \circ G_2 \circ \ldots \circ G_{k}$.
	\end{enumerate}
\end{Definition}

Next, we show that graphs sampled from multi-block distributions are valid permutation graphs.
\begin{claim}\label{clm:p-pass-mulblock-valid}
	Any graph sampled from $\pmultiblock{\Grs, \Sigma, L, M}$ is a permutation graph for $\gammastar$ defined by \Cref{eq:def-gammastar} with $k (2 \nrs \cdot b + 2N_{p-1}(\nrs \cdot b))$ vertices.
\end{claim}
\begin{proof}
	Any graph sampled from $\pmultiblock{\Grs, \Sigma, L, M, \Gamma}$ is a permutation graph for $\gammastar$ by \Cref{clm:concatenation} and \Cref{clm:p-pass-block-valid}. As we concatenate $k$ graphs sampled from $p$-pass block distributions, the total number of vertices follows from \Cref{clm:p-pass-block-valid}.
\end{proof}

Next, we will prove the analog of \Cref{clm:p-block-p-1-pass} for multi-blocks,  with a slightly larger advantage. 
\begin{claim}\label{clm:multi-indist-p-1-pass}
	For any $L_1, L_2 \in [t]^k$ and two hypermatchings $M_1, M_2$, the two distributions 
	\[
	\cD_1 = \pmultiblock{\Grs, \Sigma, L_1, M_1} \quad \text{and} \quad \cD_2 = \pmultiblock{\Grs, \Sigma, L_2, M_2}
	\]
	 are $2k\cdot \Delta_{p-1}(\nrs \cdot b)$-indistinguishable for $(p-1)$-pass streaming algorithms $o((\nrs \cdot b)^{1+\beta/2})$ space.
\end{claim}
\begin{proof}
	Let $A$ be a $p-1$ pass algorithm using space $o((\nrs \cdot b)^{1+\beta/2})$. For $i \in [k]$, let $\jstar_{1,i}, \jstar_{2,i}$ and $\Estar_{1, i}, \Estar_{2,i}$ be the corresponding values of $\jstar_i, \Estar_i$ from \Cref{def:dist-multi-block} on inputs $ \Grs, \Sigma, L_1, M_1$ and $\Grs, \Sigma, L_2, M_2$ respectively. 
	Let the distribution $\pblock{\Grs, \Sigma, \jstar_{j,i}, \Estar_{j,i}}$ be referred to by $\cB_{j,i}$ for $j = 1,2$ and $i \in [k]$.
	For $j = 1,2$, distribution $\cD_j$ is fixed based on all samples from $\cB_{j,i}$ for $i \in [k]$. By \Cref{fact:tvd-data-processing}, we have that,
	.	\begin{align*}
		\tvd{\mem{p-1}{A}(\cD_1)}{\mem{p-1}{A}(\cD_2)} &\leq \tvd{(\cB_{1,1},\ldots,\cB_{1,k})}{(\cB_{2,1}, \ldots, \cB_{2,k})} \\
		&  \leq \sum_{i \in[k]} \tvd{\cB_{1,i}}{\cB_{2,i}}\tag{by hybrid argument \Cref{prop:hybrid-arg}} \\
		&\leq 2k \Delta_{p-1}(\nrs\cdot b). \tag{by \Cref{asmp:p-1-pass} } \qedhere
	\end{align*}
\end{proof}

\subsection{Permutation Hiding in Multiple Passes}

Our approach is similar to the proof of \Cref{lem:1-pass-hiding}, where we hide multiple simple permutations to hide any $\sigma$ from $S_m$. Let us show that simple permutations under partition $\lexpart$ can be hidden from $p$-pass algorithms.

\begin{lemma}\label{lem:simple-perm-p-pass}
	Under \Cref{asmp:p-1-pass}, there exists a simple permutation hiding generator for partition $\lexpart$, $\GGsim: \Ssim \rightarrow \cD_m$ for $p$-pass streaming algorithms using space $s = o(\paren{m/b}^{1+2\beta/3})$ 
	when $b = m^{\beta/100}$ with the following parameters: 
	\begin{enumerate}[label=$(\roman*)$]
		\item The total number of vertices  $ n = (3200/\beta) \cdot (\nrs \cdot b + N_{p-1}(\nrs \cdot b)) + N_{p-1}(m)$.
		\item The advantage gained is at most $ (3200/\beta) \cdot \Delta_{p-1}(\nrs \cdot b) + \Delta_{p-1}(m) +\alpha \cdot  m^{-5} \cdot (p/\beta)^{50/\beta}$.
	\end{enumerate}
	
\end{lemma}

The explicit construction of simple permutation hiding generators for $p$ passes follows.
\begin{ourbox}
	\textbf{Construction of $\GGsim(n,p,s,\lexpart, \delta): \Ssim \rightarrow \cD_m$ (denoted by $\GGsim_p$) on input $\rho \in \Ssim$. }
	
	\begin{enumerate}[label=$(\roman*)$]
		\item Fix an $(r,t)$-RS graph $\Grs$ with $r = 2m/b = \nrs/\alpha$ and $t = \paren{\nrs}^{\beta}$.
		\item Instantiate $k$ players and the referee in $\IndexHPH_{r,t,b,k}$ (\Cref{prob:mph}) with $k = 1600/\beta$ such that the solution to the instance, $\Gammastar \circ \Gamma$ is $\tovec(\rho)$. Let the variables in the instance be $\Sigma, L,M$, and $\Gamma$.
		\item Sample graph $G^{\Sigma} \sim \pmultiblock{\Grs, \Sigma, L, M}$ and $G^{\Gamma} \sim \GG_{p-1}(\joinit(\Gamma))$. 
		\item Output $G^{\Sigma} \circ G^{\Gamma}$.
	\end{enumerate}
\end{ourbox}

First, let us show that any graph output by the distribution $\GGsim$ for $p$-pass algorithms is a valid permutation graph.

\begin{claim}\label{clm:p-pass-simple-valid}
	Given any $\rho \in \Ssim_m$, any graph $G \sim \GGsim_p(\rho)$ is a permutation graph for $\rho$ with $n = 2k (N_{p-1}(\nrs \cdot b) + \nrs \cdot b) + N_{p-1}(m)$ vertices. 
\end{claim}

\begin{proof}
	We know that $\pmultiblock{\Grs, \Sigma, L, M}$ is a permutation graph for $\joinit(\Gammastar)$, where $\Gammastar$ is defined as in \Cref{prob:mph} by \Cref{clm:p-pass-mulblock-valid}. $\GGsim_p(\rho)$ is a permutation graph for $\rho$ by \Cref{clm:concatenation}.
	
	The bound on the total number of vertices in $\pmultiblock{\Grs, \Sigma, L, M}$ follows from \Cref{clm:p-pass-mulblock-valid} and \Cref{asmp:p-1-pass}.
\end{proof}

To prove that the preceding construction is a valid simple permutation hiding generator for $p$-pass algorithms, we have to argue that no $p$-pass algorithm using space $s = o((m/b)^{1+2\beta/3})$ can distinguish between two graphs output from different distributions. Let $A$ be one such algorithm which distinguishes between $\GG(\rho_1)$ and $\GG(\rho_2)$ for some $\rho_1, \rho_2 \in \Ssim_m$. We can assume $A$ is deterministic by Yao's minimax principle. 

Here our approach will differ  from the 1-pass case. To run streaming algorithm $A$ on the input graph for $p$-passes the naive way, back and forth communication between the referee and the players is required, and this is not possible. Instead, the players will run algorithm $A$ for $p-1$ passes on a \emph{different graph}, assuming some input on behalf of the referee, and the last pass is run on the original graph. We will show that this will be sufficient to solve $\IndexHPH$. Before giving the protocol to run $p$-passes of $A$, let us see that one pass of $A$ can be run using the inputs of both the referee and the players. 

\begin{claim}\label{clm:running-1-pass-extension}
	Suppose we are given an instance of $\IndexHPH$ with inputs $\Grs, \Sigma, L, M$ and $\Gamma$. 
	Let $G_1 \sim \pmultiblock{\Grs, \Sigma, L, M}$ and let $G_2 \sim \GG_{p-1}(\joinit(\Gamma))$. The players and the referee can run any streaming algorithm $A$ using space at most $s$ for one pass on graph $G_1 \circ G_2$ using at most $k \cdot s$ bits of communication. 
\end{claim}
\begin{proof}
	For $i \in [k]$, in any graph $G_i \sim\pblock{\Grs, \SigmaP{i}, \ell_i, \Estar_i}$, the edges based on $\SigmaP{i}$ can be added by player $\QP{i}$ and the other edges based on $L, M$ and $\Estar_i$ can be added by the referee. No edge depends on both the inputs from the player and the referee. 
	
	Each player $\QP{i}$ in increasing order of $i \in [k]$, takes the memory state of $A$ after adding all the edges based on $\SigmaP{1}, \SigmaP{2}, \ldots, \SigmaP{i-1}$, runs $A$ on the edges based on $\SigmaP{i}$ and then uses $s$ bits of communication to write the new memory state of $A$ on the board. This is totally $k \cdot s$ bits of communication. The referee can use the state from $\QP{k}$ and add the edges based on $L, M$, and $\Gamma$, which concludes the proof.  
\end{proof}

Let us see how the players and the referee can \emph{simulate} algorithm $A$ for $p$ passes. 
\begin{ourbox}
	\textbf{Protocol $\Prot_A$ for $\IndexHPH$ using the $p$-pass streaming algorithm $A$}
	\begin{enumerate}[label=$(\roman*)$]
		\item The players $\QP{i}$ for $i \in [k]$ pick the lexicographically first $\Lbar \in [t]^k$, $\Gammabar \in \paren{S_b}^{r/2}$ and a hypermatching $\Mbar$ from $[r]^k$ collectively. 
		\item Sample a graph $\overline{G_{\Sigma}} \sim \pmultiblock{\Grs, \Sigma, \Lbar, \Mbar}$ and $\overline{G_{\Gamma}} \sim \GG_{p-1}(\joinit(\Gammabar))$. 
		\item Run algorithm $A$ for $p-1$ passes on $\Gbar  = \overline{G_{\Sigma}} \circ \overline{G_{\Gamma}}$.
		\item The players and the referee jointly sample $G_{\Sigma} \sim \pmultiblock{\Grs, \Sigma, L, M}$, and the referee samples $G_{\Gamma} \sim \GG_{p-1}(\joinit(\Gamma))$.
		\item In the last pass, the players and the referee run algorithm $A$ on $G  = G_{\Sigma} \circ G_{\Gamma}$. 
	\end{enumerate}
\end{ourbox}

Let us argue that running such a protocol $\Prot_A$ is possible for any $p$ pass streaming algorithm $A$. 
\begin{claim}\label{clm:running-A-for-MPH}
	There is a way for the players and the referee to run protocol $\Prot_A$ using at most $ k \cdot p\cdot s$ total communication.
\end{claim}
\begin{proof}
	The values of $\Gammabar, \Lbar$ and $\Mbar$ are known to all
	players $\QP{i}$ for $i \in [k]$. The distribution in \Cref{asmp:p-1-pass} is known to all the players too. Let $\Lbar = (\overline{\ell_1}, \overline{\ell_2}, \ldots, \overline{\ell_k})$, and  let edge tuple $\overline{\Estar_i} = (\Mbar_{1,i}, \Mbar_{2,i}, \ldots, \Mbar_{r/2,i})$ be $r/2$ edges in $\Mrs{\overline{\ell_i}}$ for $i \in [k]$.

	In \Cref{def:dist-multi-block}, we know that graphs are sampled from distribution $\pblock{\Grs, \SigmaP{i}, \overline{\ell_i}, \overline{\Estar_i}}$ for each $i \in [k]$. The graph $\Grs$, index $\overline{\ell_i}$ and edges $\overline{\Estar_i}$ for $i \in [k]$ are known to all the players since they fix graph $\Grs$. Hence, player $\QP{i}$ can sample a graph from $\overline{G_i} \sim \pblock{\Grs, \SigmaP{i}, \overline{\ell_i}, \overline{\Estar_i}}$ for each $i \in [k]$ privately. 
	
	Sampling from $\GG_{p-1}(\Gammabar)$ can be done by any player since they all know what $\Gammabar $ is. To execute 1 pass of step $(iii)$ of $\Prot_A$ on graph $\Gbar$, the player $\QP{i}$ can, in order, add the edges from $\overline{G_i}$ for $i \in [k]$ and then any player can add the edges from $\overline{G_{\Gamma}}$. The players and the referee can execute step $(iv)$ using \Cref{clm:running-1-pass-extension}.
\end{proof}

Now we argue the correctness of this protocol. We need some extra notation before we proceed.

\paragraph{Notation.}
We fix the input of the $\IndexHPH_{r,t,b,k}$ instance to be $\Grs, \Sigma, L, \cM$ and $\Gamma$. Let $\delta_A $ be the advantage gained by $p$-pass streaming algorithm $A$ that distinguishes between the two distributions $\GGsim_p(\rho_1)$ and $\GGsim_p(\rho_2)$. Let $s_A = o((m/b)^{1+2\beta/3})$ denote the space used by $A$.

Let $\protFake$ be the protocol of running $A$ for $p$-passes on graph $G$ which is the concatenation of graphs sampled from distributions $\pmultiblock{\Grs, \Sigma, L, \cM}$ and $\GG_{p-1}(\joinit(\Gamma))$. (Note that $\protFake$ is an impossible protocol to run because there is no back and forth communication between the referee and the players; however, it is a well-defined random variable/distribution.)

Let $\mem{j}{A}(\protFake)$ be the random variable denoting the memory state of $A$ after running $j$ passes of $G$ in $\protFake$ for $j \leq p$.  Let $\mem{j}{A}(\Prot_A)$ be the random variable denoting the memory contents of the board after running $j$ passes of $\Prot_A$ on graph $\Gbar $ for $j \leq p-1$ and graph $G$ for pass $p$. 

Let graph $\Gfixed$ denote the subgraph of $G$ with all the edges added by the players $\QP{i}$ for $i \in [k]$ based on $\SigmaP{i}$. This graph is fixed because we have fixed the input to $\mph_{t,r,b,k}$. Let $\Grest$ be the distribution of the rest of the edges in $G$.

\begin{claim}\label{clm:correctness-A-on-MPH}
	Protocol $\Prot_A$ solves $\IndexHPH_{r,t,b,k}$ with advantage at least 
	\[
	\delta_A-(2k\cdot \Delta_{p-1}(\nrs \cdot b) + \Delta_{p-1}(m))
	\]
	 and total communication at most $k \cdot p \cdot s_A$ when $s_A = o(m^{1+\beta/2})$. 
\end{claim}
\begin{proof}
	The contents of the last pass are a deterministic function of the contents of pass $p-1$ and the edges added. By \Cref{fact:tvd-data-processing}, we have that,
	\begin{align*}
		\tvd{\mem{p}{A}(\Prot_A)}{\mem{p}{A}(\protFake)} 
		\leq \tvd{(\mem{p-1}{A}(\Prot_A), \Gfixed, \Grest)}{(\mem{p-1}{A}(\protFake), \Gfixed, \Grest)}.
	\end{align*}
	The edges added in the last pass for both protocols are the ones from $\Gfixed$ which are fixed, and the others from distribution $\Grest$. 
	We can write the RHS term as,
	\begin{align*}
		&\tvd{(\mem{p-1}{A}(\Prot_A), \Gfixed, \Grest)}{(\mem{p-1}{\protFake}, \Gfixed, \Grest)} \\
		&\hspace{5mm}\leq \tvd{\mem{p-1}{A}(\Prot_A)}{\mem{p-1}{A}(\protFake)} \\
		& \hspace{10mm}+ \Exp_{\prot \sim \mem{p-1}{A}(\Prot_A)} \paren{\tvd{\paren{(\Gfixed, \Grest)\mid \mem{p-1}{A}(\Prot_A) = \prot}}{\paren{(\Gfixed, \Grest)\mid \mem{p-1}{A}(\protFake) = \prot}}} \tag{by \Cref{fact:tvd-chain-rule}} \\
		&\hspace{5mm} \leq \tvd{\mem{p-1}{A}(\Prot_A)}{\mem{p-1}{A}(\protFake)},
	\end{align*}
	where in the last step, the second term becomes zero, as they are the exact same distribution.
	For the first $p-1$ rounds, $\Prot_A$ and $\protFake$ run algorithm $A$ on input graphs sampled from $\pmultiblock{\Grs, \Sigma, \Lbar, \Mbar}, \GG_{p-1}(\joinit(\Gammabar))$ and $\pmultiblock{\Grs, \Sigma, L, \cM}, \GG_{p-1}(\joinit(\Gamma))$ respectively. As this is a $p-1$ pass streaming algorithm using space $o(m^{1+\beta/2})$, by \Cref{clm:multi-indist-p-1-pass} and \Cref{asmp:p-1-pass},
	\begin{align*}
		\tvd{\mem{p-1}{A}(\Prot_A)}{\mem{p-1}{A}(\protFake)} \leq  2k\cdot \Delta_{p-1}(\nrs \cdot b) + \Delta_{p-1}(m).
	\end{align*}
	For any input, by our assumption, we know that protocol $\protFake$ distinguishes between $\tilde{\GG}(\Gamma_1)$ and $\tilde{\GG}(\Gamma_2)$ with advantage at least $\delta_A$. However, we have shown that for any input, the total variation distance between the transcripts of protocols $\protFake$ and $\Prot_A$ is low. Hence,
	\begin{align*}
		\Pr\paren{\textnormal{Success of } \protFake} &\leq \Pr\paren{\textnormal{Success of } \Prot_A} +  \tvd{\mem{p}{A}(\Prot_A)}{\mem{p}{A}(\protFake)} \\
		\Pr\paren{\textnormal{Success of } \Prot_A} &\geq\frac12 +  \delta_A - (2k\cdot \Delta_{p-1}(\nrs \cdot b) + \Delta_{p-1}(m)). \qedhere
	\end{align*}
\end{proof}

We have all we need to prove \Cref{lem:simple-perm-p-pass} now.

\begin{proof}[Proof of \Cref{lem:simple-perm-p-pass}]
	The total number of vertices in our construction follows directly from \Cref{clm:p-pass-simple-valid}. 
	
	By \Cref{clm:running-A-for-MPH} and \Cref{clm:correctness-A-on-MPH}, we know that $\Prot_A$ can solve instances of $\mph_{t,r,b,k}$ with advantage at least $\delta_A - (2k \cdot \Delta_{p-1}(\nrs \cdot b)  + \Delta_{p-1}(rb/2)) $ with total communication $k \cdot p \cdot s_A$. 
	
	 When $p = o(\beta \cdot r^{\beta/3})$, we have that $k \cdot p \cdot s_A = O(1/\beta) \cdot o((m/b)^{1+2\beta/3})  \cdot o(\beta \cdot r^{\beta/3})= o(rt)$. Hence \Cref{thm:mph} is applicable. We know that the advantage gained by $\Prot_A$ is at most,
	\begin{align*}
		r \cdot o\paren{\frac{k \cdot p \cdot s_A}{r\cdot t}}^{k/32} \leq 	r \cdot o\paren{\frac{\alpha^{\beta} \cdot k \cdot p \cdot r^{1+2\beta/3}}{r^{1+\beta}}}^{50/\beta} \leq \frac{1}{r^{15}} \cdot (\alpha^{\beta} \cdot \beta^{-1} \cdot p)^{50/\beta} 
	\end{align*}
	Algorithm $A$ is given $o(r^{1+2\beta/3}) = o((m)^{1+\beta/2})$ when $b = m^{\beta/100}$. So we can lower bound the advantage of $\Prot_A$ from \Cref{clm:correctness-A-on-MPH} as,
	\begin{align*}
		\delta_A\leq (2k\cdot \Delta_{p-1}(\nrs \cdot b) + \Delta_{p-1}(m)) + m^{-5} \cdot (p\cdot \beta^{-1})^{50/\beta} \cdot \alpha.\qedhere
	\end{align*}
\end{proof}

Using \Cref{lem:simple-perm-p-pass}, we can construct the permutation hiding generator for $p$-passes similar to the construction in \Cref{sec:one-pass} using simple permutation hiding generators for $p$ passes.

\begin{proof}[Proof of \Cref{thm:perm-hiding}]
	The base case was proved in \Cref{lem:1-pass-hiding}.  We use \Cref{asmp:p-1-pass} to infer the statement of the theorem, completing the proof by induction. We fix $b = m^{\beta/100}$, and we concatenate $\caks \cdot 100/\beta$ graphs sampled from the simple permutation hiding generator for $p$-pass algorithms.  
	Using \Cref{lem:simple-perm-p-pass}, we can bound the total number of vertices as, 
	\begin{align*}
		n &\leq 	\daks \cdot\paren{ (3200/\beta) \cdot (\nrs \cdot b + N_{p-1}(\nrs \cdot b)) + N_{p-1}(m)}  \\
		&\leq 3.2 \cdot 10^5 \cdot \caks \cdot \frac1{\beta^2} \cdot \paren{\frac{2.5 \cdot \caks \cdot 10^6}{\alpha \beta^2}}^{p-1} \cdot \paren{6m/\alpha} \\
		&\leq \paren{\frac{2.5 \cdot \caks \cdot 10^6}{\alpha \beta^2}}^{p} \cdot m.
	\end{align*}
	Using the bound on the advantage in \Cref{lem:simple-perm-p-pass}, we get that, for any $\rho_1, \rho_2 \in S_m$, 
	\begin{align*}
		&\tvd{\mem{p}{A}(\GG(\rho_1))}{\mem{p}{A}(\GG(\rho_2))} \\  		
		&\hspace{10mm}\leq \daks \cdot  \paren{(3200/\beta) \cdot \Delta_{p-1}(\nrs \cdot b) + \Delta_{p-1}(m) + \alpha \cdot m^{-5} \cdot (p/\beta)^{50/\beta}} \\
		&\hspace{10mm}\leq \daks \cdot  \paren{(3200/\beta+1) \cdot \Delta_{p-1}(m) + \alpha \cdot m^{-5} \cdot (p/\beta)^{50/\beta}} \tag{as $\nrs \cdot b > m$} \\
		&\hspace{10mm}\leq \paren{\frac{4 \cdot 10^6 \cdot \caks}{\beta^2}}^{p-1} \cdot m^{-5} \cdot \alpha \cdot \paren{\frac{3.2 \cdot 10^5 \cdot \caks}{\beta^2}} \cdot (((p-1)/\beta)^{50/\beta} + (p/\beta)^{50/\beta}) \\
		&\hspace{10mm}\leq \paren{\frac{4 \cdot 10^6 \cdot \caks}{\beta^2}}^{p} \cdot m^{-5} \cdot \alpha \cdot (p/\beta)^{50/\beta},
	\end{align*}
	completing the proof.
\end{proof}

This concludes our construction of permutation hiding generators and their analysis.


\section{A Multi-Pass Streaming Lower Bound for Matchings}\label{sec:matching-lower}

We are now ready to present the proof of our main result in~\Cref{thm:main-lb}, restated below for the convenience of the reader. 

\begin{theorem*}[Restatement of~\Cref{res:main-lb}]
Suppose that for infinitely many choices of $N \geq 1$, there exists $(2N)$-vertex bipartite $(r,t)$-RS graphs with $r = \alpha \cdot N$ and $t = N^{\beta}$ for some fixed parameters $\alpha, \beta \in (0,1)$; the parameters $\alpha,\beta$ can depend
on $N$. 

Then, there exists an $\eps_0 = \eps_0(\alpha,\beta)$ such that the following is true. For any $0 < \eps < \eps_0$, any streaming algorithm that uses $o(\eps^{2} \cdot n^{1+\beta/2})$ space on $n$-vertex bipartite graphs 
and can determine with constant probability whether the input graph has a perfect matching or its maximum matchings have size at most $(1-\eps) \cdot n/2$ requires 
\[
\Omega\Paren{\dfrac{\log{(1/\eps)}}{\log{(1/\alpha\beta)}}}
\]
passes over the stream. 
\end{theorem*}

The proof is based on a standard reduction, e.g., in~\cite{AssadiR20}, from reachability to bipartite matching (a similar reduction also appeared in~\cite{ChenKPSSY21}). We present the reduction for completeness.

Let $m \geq 1$ be an even integer and define the following two permutations: 
\begin{itemize}
	\item $\sigma_{=}$: the identity permutation over $[m/2]$; 
	\item $\sigma_{\times}$: the ``cross identity'' permutation that  maps $[m/2]$ to $[m/2+1:m]$ and $[m/2+1:m]$ to $[m/2]$ using identity permutations.  
\end{itemize}

Let $\GG = \GG_{m,n,p,s,\delta}$ be a permutation hiding generator for $m$ and any integer $p \geq 1$ with the parameters $n,s,\delta$ as in~\Cref{thm:perm-hiding}. We are going to use the 
indistinguishability of graphs sampled from $\GG(\sigma_{=})$ versus $\GG(\sigma_{\times})$ to prove~\Cref{thm:main-lb}. To do so, we need to turn the graphs sampled from this distribution 
into bipartite graphs wherein size of the matching, even approximately, is a distinguisher for the sample. This is done in the following. 

\subsubsection*{The bipartite graph construction} 

Let $G = (V,E)$ be sampled from $\GG(\sigma_{=})$ or $\GG(\sigma_{\times})$. Create the following bipartite graph $H$ from $G$:  

\begin{itemize}
	\item The vertices of $H$ are $(L \cup S)$ and $(R \cup T)$ on each side defined as follows: 
	\begin{enumerate}[label=$(\roman*)$]
	\item Let $L$ and $R$ be each a separate copy of $V$. For any vertex $v_i \in V$, we use $v^{\Left}_i$ and $v^{\Right}_i$ to denote the copy of $v_i$ in $L$ and $R$, respectively. 
	\item Additionally, create two new sets $S$ and $T$ of vertices with size $m/2$ each.
	\end{enumerate} 
	\item The edges of $H$ are $E_H$ plus a matching $M$ defined as follows: 
	\begin{enumerate}[label=$(\roman*)$]
	\item For any $i \in [m/2]$ and the $i$-th vertex $v_i \in \first{G}$ (resp. $v_i \in \last{G}$), add the edge between $i$-th vertex $u_i \in S$ and $v^{\Right}_i \in R$ (resp. $u_i \in T$ and $v^{\Left}_i \in L$) to $E_H$. 
	Moreover, for any edge $(u_i,v_j) \in E$, add the edge between $u^{\Left}_i \in L$ and $v^{\Right}_j \in R$. 
	\item Finally, for every $v_i \in V$, add the edge between $v^{\Left}_i \in L$ and $v^{\Right}_i \in R$ to $M$. 
	\end{enumerate} 
\end{itemize}
This way, we have a bipartite graph $H=(L \cup S, R \cup T, E_H \cup G)$ associated with $G$ with $n+m/2$ vertices on each side of the bipartition. The following lemma establishes the key property of $H$. 

\begin{lemma}\label{lem:H-associated-G}
	Consider a bipartite graph $H$ associated with a graph $G$ defined as above. We have:
	\begin{enumerate}[label=$(\roman*)$]
		\item\label{item:HaG-p1} if $G \sim \GG(\sigma_{=})$, then $H$ has a perfect matching of size $n+m/2$; 
		\item\label{item:HaG-p2} on the other hand, if $G \sim \GG(\sigma_{\times})$, the maximum matching size in $H$ is $n$. 
	\end{enumerate}
\end{lemma}
\begin{proof}
	The idea in both cases is to start with the matching $M$ in $H$ and consider augmenting it. We prove each part separately. 
	\paragraph{Proof of~\Cref{item:HaG-p1}.} Consider any vertex $u_i \in S$ in $H$. These vertices are all unmatched by $M$. Moreover, $u_i$ only has a single edge to some vertex $v^{\Right}_i \in R$. 
	Consider the original copy $v_i \in \first{G}$ of $v^{\Right}_i$ in the original graph $G$. Since $G$ is a permutation graph for $\sigma_{=}$, there is a path $P(v_i)$ in $G$ as follows: 
	\[
		P(v_i) := v_i \in \first{G} ~\rightarrow~ w_1 ~\rightarrow~ w_2 ~\rightarrow~ \cdots ~\rightsquigarrow~  \sigma_{=}(v_i) \in \last{G}.
	\]
	This path, translates to the following augmenting path in $H$ for the matching $M$:
	\[
		u_i \in S ~\rightarrow_{E_H}~ v^{\Right}_i ~\rightarrow_{M}~ v^{\Left}_i ~\rightarrow_{E_H}~ w^{\Right}_1 ~\rightarrow_{M}~ w^{\Left}_1 ~\rightarrow_{E_H}~ \cdots ~\rightsquigarrow~ \sigma_{=}(v^{\Left}_i) ~\rightarrow_{E_H}~ u_i \in T,
	\]
	using the fact that $\sigma_{=}$ is the identity matching. Thus, this is a valid augmenting path. 
	
	Moreover, for all vertices $u_i \in S$, these augmenting paths should be vertex disjoint. This is because these augmenting paths follow the same paths from $v_i \in \first{G}$ to $v_i \in \last{G}$ and the paths from the first layer 
	of a permutation graph to its last layer are vertex disjoint (if two paths collide, then starting point of each of these paths, can reach two separate vertices in $\last{G}$, violating the permutation graph property). 
	
	This implies that the matching $M$ of size $n$ admits $m/2$ vertex-disjoint augmenting paths, thus by augmenting it we obtain a matching of size $n+m/2$. Given the bound on the number of vertices of the graph, this is a perfect matching.
	
	\paragraph{Proof of~\Cref{item:HaG-p2}.} All unmatched vertices by $M$ in $L \cup S$ in $H$ actually belong to $S$. Consider any vertex $u_i \in S$ then and recall it has a single edge to some vertex $v^{\Right}_i \in R$. 
	For $u_i$ to be part of an augmenting path, it needs to reach some vertex $w_j \in T$. Such a vertex $w_j$ in turn only has one edge from a single vertex $z^{\Left}_j \in L$ by construction. 
	
	Using the correspondence between augmenting paths in $H$ and the directed paths in $G$, we thus need to have a path from $v_i \in \first{G}$ to $z_j \in \last{G}$. Moreover, by the construction of $H$, 
	we additionally need both $v_i$ to be in the first $m/2$ vertices of $\first{G}$ and $z_j$ to be in the first $m/2$ vertices of $\last{G}$. But, since $G$ is a permutation graph for $\sigma_{\times}$, 
	$v_i \in \first{G}$ can only reach $\sigma_{\times}(v_i) \in \last{G}$ which cannot be among the first $m/2$ vertices. 
	
	This implies that in this case, the matching $M$ has no augmenting paths in $H$ and is thus a maximum matching of size $n$. 
\end{proof}

We now use this lemma to conclude the proof of~\Cref{thm:main-lb}. 

\subsubsection*{Concluding the proof}

Let $A$ be any $p$-pass $s$-space streaming algorithm for the maximum matching problem on bipartite graphs with $n+m/2$ vertices on each side of the bipartition
that can distinguish between graphs with a perfect matching versus ones with maximum matching of size at most $n$. We turn $A$ into a $p$-pass $s$-space streaming algorithm $B$ for distinguishing 
between graphs sampled from $\GG(\sigma_{=})$ and $\GG(\sigma_{\times})$. We again use the extra power provided in~\Cref{def:streaming} to the streaming algorithms.

Given a graph $G$ in the stream, the algorithm $B$ can first create vertices of $H$ and edges in $M$ plus edges from $S$ to $R$ and $T$ to $L$ before even reading the edges of $G$, and pass these edges of $H$ to $A$;
then, it starts reading the stream of edges of $G$ and each edge $(u,v)$ translates into an edge $(u^{\Left},v^{\Right})$ of $H$ by construction, which $B$ passes to $A$ again. This way, $B$ can implement pass of $A$ 
and at the beginning of the next pass, it again creates edges of $M$ and edges from $S$ to $R$ and $T$ to $L$ in $H$ before reading its own stream, and continues as before until all $p$ passes are finished. 

We thus have, 
\[
	\delta \geq \tvd{\mem{p}{B}(\GG(\sigma_{=}))}{\mem{p}{B}(\GG(\sigma_{\times}))} = \tvd{\mem{p}{A}(H(\GG(\sigma_{=})))}{\mem{p}{A}(H(\GG(\sigma_{\times})))},
\]
where $\delta$ on the left is the maximum advantage because of the indistinguishability guarantee of $\GG$ in~\Cref{thm:perm-hiding}. 

On the other hand, by~\Cref{lem:H-associated-G}, size of the maximum matching in $H(\GG(\sigma_{=})))$ is always $n+m/2$ and in $H(\GG(\sigma_{\times}))$ is always $n$. 
This, together with~\Cref{fact:tvd-sample} (on probability of success of distinguishing distributions via a single sample) implies that $A$ also cannot distinguish between these two cases of perfect matching versus maximum matching of size $n$ with probability better than $1/2+\delta$. 

We only now need to instantiate the parameters of $A$. In the following, let $n_H := n+m/2$ denote the number of vertices on each side of the bipartition of $H$ and $\eps = m/4n$, so that in one case 
$H$ has a perfect matching of size $n_H$ and in the other case its maximum matching size is at most $(1-\eps) \cdot n_H$. We further have, 
\begin{align*}
	&n_H = n+m/2 = {\Theta(\frac{1}{\alpha \cdot \beta^2})}^{p} \cdot \Theta(m) \tag{number of vertices}, \\
	&s := o(m^{1+\beta/2}) \tag{space of streaming algorithm},
\end{align*}
by the guarantees of~\Cref{thm:perm-hiding}. By calculating parameters $p$ and $s$ in terms of $\eps$ and $n_H$, we get, 
\begin{align*}
	p &= \Omega\Paren{\frac{\log{({n_H}/{m})}}{\log({1}/{\alpha \cdot \beta^2})}} = \Omega\Paren{\frac{\log{(1/\eps)}}{\log{(1/\alpha\beta)}}}, \\
	s &= o(m^{1+\beta/2}) = o((\alpha \cdot \beta^2)^{p \cdot (1+\beta/2)} \cdot n_H^{1+\beta/2}) = o(\eps^{2} \cdot n_H^{1+\beta/2}).
\end{align*}
We set our range of $\eps $  to be $\eps = \Omega(n^{-\beta/4})$. We just have to see that for this range of $\eps$ and $p$, our advantage remains low. From \Cref{thm:perm-hiding},
\begin{align*}
	\delta &\leq (p/\beta)^{\Theta(1/\beta)} \cdot \Theta(1/\beta)^{2p} \cdot m^{-5} \\
	&\leq \paren{\frac{\log(1/\eps)}{\log(1/\alpha\beta)}}^{\Theta(1/\beta)} \cdot \Theta(1/\beta)^{\log(1/\eps)/\log(1/\alpha\beta)} \cdot m^{-5} \cdot (1/\beta)^{\Theta(1/\beta)}\\
	&\leq n \cdot \Theta(1)^{\log(1/\eps)} \cdot (4\eps \cdot n)^{-5} \cdot  (1/\beta)^{\Theta(1/\beta)}\\
	&\leq\Theta( n^{1+\beta/4} \cdot n^{(1-\beta/4)\cdot(-5)}) \leq 1/\poly(n).
\end{align*}
This concludes the proof of~\Cref{thm:main-lb}.

\begin{Remark}
{Our lower bound in~\Cref{thm:main-lb} can also be extended to some other fundamental problems including \textbf{reachability} and \textbf{shortest path} studied extensively in the streaming model in~\cite{FeigenbaumKMSZ08,GuruswamiO13,ChakrabartiGMV20,AssadiR20,ChenKPSSY21}, because 
	permutation hiding graphs also provide a lower bound for those problems; see~\cite{ChenKPSSY21}.}
	
	\medskip
	
	{In particular, under the (plausible) hypothesis that $\beta$ can be $\Omega(1)$, 
	by using a parameter $m/n \approx n^{-\beta/6}$ in the construction of permutation hiding graphs of~\Cref{thm:perm-hiding} (roughly, setting $\eps \approx n^{-\beta/6}$), we obtain a 
	\textbf{lower bound of 
	$\Omega(\log{n})$
	passes for solving directed reachability or shortest path via semi-streaming algorithms}. 
	
	\medskip
	
	This slightly improves the lower bound of $\Omega(\log{n}/\log\log{n})$ passes for these problems obtained in~\cite{GuruswamiO13,ChakrabartiGMV20,ChenKPSSY21} 
	(albeit under the hypothesis that $\beta = \Omega(1)$; using the current state-of-the-art bound of $\Omega(1/\log\log{n})$ on $\beta$ leads to asymptotically same bounds as in~\cite{GuruswamiO13,ChakrabartiGMV20,ChenKPSSY21}).} 
\end{Remark}

\subsection*{Acknowledgement} 

We are thankful to Soheil Behnezhad, Surya Teja Gavva, Michael Kapralov, and Huacheng Yu for helpful discussions. We are
also grateful to Vladimir V. Podolskii for pointers to the sorting network literature.

\bibliographystyle{halpha-abbrv}
\bibliography{new,dump}

\clearpage

\appendix
\clearpage
\section{Background on Information Theory}\label{app:info}

We now briefly introduce some definitions and facts from information theory that are used in our proofs. We refer the interested reader to the text by Cover and Thomas~\cite{CoverT06} for an excellent introduction to this field, 
and the proofs of the statements used in this Appendix. 

For a random variable $\rA$, we use $\supp{\rA}$ to denote the support of $\rA$ and $\distribution{\rA}$ to denote its distribution. 
When it is clear from the context, we may abuse the notation and use $\rA$ directly instead of $\distribution{\rA}$, for example, write 
$A \sim \rA$ to mean $A \sim \distribution{\rA}$, i.e., $A$ is sampled from the distribution of random variable $\rA$. 

\begin{itemize}[leftmargin=10pt]
\item We denote the \emph{Shannon Entropy} of a random variable $\rA$ by
$\en{\rA}$, which is defined as: 
\begin{align}
	\en{\rA} := \sum_{A \in \supp{\rA}} \Pr\paren{\rA = A} \cdot \log{\paren{1/\Pr\paren{\rA = A}}} \label{eq:entropy}
\end{align} 
\noindent
\item The \emph{conditional entropy} of $\rA$ conditioned on $\rB$ is denoted by $\en{\rA \mid \rB}$ and defined as:
\begin{align}
\en{\rA \mid \rB} := \Ex_{B \sim \rB} \bracket{\en{\rA \mid \rB = B}}, \label{eq:cond-entropy}
\end{align}
where 
$\en{\rA \mid \rB = B}$ is defined in a standard way by using the distribution of $\rA$ conditioned on the event $\rB = B$ in Eq~(\ref{eq:entropy}).

\item The \emph{mutual information} of two random variables $\rA$ and $\rB$ is denoted by
$\mi{\rA}{\rB}$ and is defined:
\begin{align}
\mi{\rA}{\rB} := \en{A} - \en{A \mid  B} = \en{B} - \en{B \mid  A}. \label{eq:mi}
\end{align}
\noindent
\item The \emph{conditional mutual information} $\mi{\rA}{\rB \mid \rC}$ is $\en{\rA \mid \rC} - \en{\rA \mid \rB,\rC}$ and hence by linearity of expectation:
\begin{align}
	\mi{\rA}{\rB \mid \rC} = \Ex_{C \sim \rC} \bracket{\mi{\rA}{\rB \mid \rC = C}}. \label{eq:cond-mi}
\end{align} 
\end{itemize}

\subsection{Useful Properties of Entropy and Mutual Information}\label{sec:prop-en-mi}

We shall use the following basic properties of entropy and mutual information throughout. 

\begin{fact}\label{fact:it-facts}
  Let $\rA$, $\rB$, $\rC$, and $\rD$ be four (possibly correlated) random variables.
   \begin{enumerate}
  \item \label{part:uniform} $0 \leq \en{\rA} \leq \log{\card{\supp{\rA}}}$. The right equality holds
    iff $\distribution{\rA}$ is uniform.
  \item \label{part:info-zero} $\mi{\rA}{\rB \mid \rC} \geq 0$. The equality holds iff $\rA$ and
    $\rB$ are \emph{independent} conditioned on $\rC$.
  \item \label{part:cond-reduce} \emph{Conditioning on a random variable reduces entropy}:
    $\en{\rA \mid \rB,\rC} \leq \en{\rA \mid  \rB}$.  The equality holds iff $\rA \perp \rC \mid \rB$.
    \item \label{part:sub-additivity} \emph{Subadditivity of entropy}: $\en{\rA,\rB \mid \rC}
    \leq \en{\rA \mid C} + \en{\rB \mid  \rC}$.
   \item \label{part:ent-chain-rule} \emph{Chain rule for entropy}: $\en{\rA,\rB \mid \rC} = \en{\rA \mid \rC} + \en{\rB \mid \rC,\rA}$.
  \item \label{part:chain-rule} \emph{Chain rule for mutual information}: $\mi{\rA,\rB}{\rC \mid \rD} = \mi{\rA}{\rC \mid \rD} + \mi{\rB}{\rC \mid  \rA,\rD}$.
  \item \label{part:data-processing} \emph{Data processing inequality}: for a function $f(\rA)$ of $\rA$, $\mi{f(\rA)}{\rB \mid \rC} \leq \mi{\rA}{\rB \mid \rC}$. 
   \end{enumerate}
\end{fact}


\noindent
We also use the following two standard propositions, regarding the effect of conditioning on mutual information.

\begin{proposition}\label{prop:info-increase}
  For random variables $\rA, \rB, \rC, \rD$, if $\rA \perp \rD \mid \rC$, then, 
  \[\mi{\rA}{\rB \mid \rC} \leq \mi{\rA}{\rB \mid  \rC,  \rD}.\]
\end{proposition}
 \begin{proof}
  Since $\rA$ and $\rD$ are independent conditioned on $\rC$, by
  \itfacts{cond-reduce}, $\HH(\rA \mid  \rC) = \HH(\rA \mid \rC, \rD)$ and $\HH(\rA \mid  \rC, \rB) \ge \HH(\rA \mid  \rC, \rB, \rD)$.  We have,
	 \begin{align*}
	  \mi{\rA}{\rB \mid  \rC} &= \HH(\rA \mid \rC) - \HH(\rA \mid \rC, \rB) = \HH(\rA \mid  \rC, \rD) - \HH(\rA \mid \rC, \rB) \\
	  &\leq \HH(\rA \mid \rC, \rD) - \HH(\rA \mid \rC, \rB, \rD) = \mi{\rA}{\rB \mid \rC, \rD}. \qed
	\end{align*}
	
\end{proof}

\begin{proposition}\label{prop:info-decrease}
  For random variables $\rA, \rB, \rC,\rD$, if $ \rA \perp \rD \mid \rB,\rC$, then, 
  \[\mi{\rA}{\rB \mid \rC} \geq \mi{\rA}{\rB \mid \rC, \rD}.\]
\end{proposition}
 \begin{proof}
 Since $\rA \perp \rD \mid \rB,\rC$, by \itfacts{cond-reduce}, $\HH(\rA \mid \rB,\rC) = \HH(\rA \mid \rB,\rC,\rD)$. Moreover, since conditioning can only reduce the entropy (again by \itfacts{cond-reduce}), 
  \begin{align*}
 	\mi{\rA}{\rB \mid  \rC} &= \HH(\rA \mid \rC) - \HH(\rA \mid \rB,\rC) \geq \HH(\rA \mid \rD,\rC) - \HH(\rA \mid \rB,\rC) \\
	&= \HH(\rA \mid \rD,\rC) - \HH(\rA \mid \rB,\rC,\rD) = \mi{\rA}{\rB \mid \rC,\rD}. \qed
 \end{align*}

\end{proof}

\subsection{Measures of Distance Between Distributions}\label{sec:prob-distance}

We use two main measures of distance (or divergence) between distributions, namely the \emph{Kullback-Leibler divergence} (KL-divergence) and the \emph{total variation distance}. 

\paragraph{KL-divergence.} For two distributions $\mu$ and $\nu$ over the same probability space, the \textbf{Kullback-Leibler (KL) divergence} between $\mu$ and $\nu$ is denoted by $\kl{\mu}{\nu}$ and defined as: 
\begin{align}
\kl{\mu}{\nu}:= \Ex_{a \sim \mu}\Bracket{\log\frac{\mu(a)}{{\nu}(a)}}. \label{eq:kl}
\end{align}
We also have the following relation between mutual information and KL-divergence. 
\begin{fact}\label{fact:kl-info}
	For random variables $\rA,\rB,\rC$, 
	\[\mi{\rA}{\rB \mid \rC} = \Ex_{(B,C) \sim {(\rB,\rC)}}\Bracket{ \kl{\distribution{\rA \mid \rB=B,\rC=C}}{\distribution{\rA \mid \rC=C}}}.\] 
\end{fact}

We use the following standard facts about KL-divergence. 
\begin{fact}[Chain rule of KL-divergence]\label{fact:kl-chain-rule}
  Let $\mu(\rX,\rY)$ and $\nu(\rX,\rY)$ be two distributions for random variables $\rX,\rY$. Then, 
\[
	\kl{\mu(\rX,\rY)}{\nu(\rX,\rY)} = \kl{\mu(\rX)}{\nu(\rX)} + \Exp_{x \sim \mu(\rX)} \kl{\mu(\rY \mid \rX=x)}{\nu(\rY \mid \rX=x)}. 
\]
Moreover, if $\rX \perp \rY$ in $\nu$ (the second argument of the KL-divergence), then, 
\[
	\kl{\mu(\rX,\rY)}{\nu(\rX,\rY)} \geq \kl{\mu(\rX)}{\nu(\rX)} + \kl{\mu(\rY)}{\nu(\rY)}. 
\]
\end{fact}
\begin{fact}[Conditioning in KL-divergence]\label{fact:kl-event}
  For any random variable $\rX$ and any event $\event$, 
\[
	\kl{\rX \mid \event}{\rX} \leq \log{\paren{\frac{1}{\Pr(\event)}}}. 
\]
Moreover, if $\event$ is a deterministic function of $\rX$, then this equation holds with equality. 
\end{fact}

\paragraph{Total variation distance.} We denote the \textbf{total variation distance} between two distributions $\mu$ and $\nu$ on the same 
support $\Omega$ by $\tvd{\mu}{\nu}$, defined as: 
\begin{align}
\tvd{\mu}{\nu}:= \max_{\Omega' \subseteq \Omega} \paren{\mu(\Omega')-\nu(\Omega')} = \frac{1}{2} \cdot \sum_{x \in \Omega} \card{\mu(x) - \nu(x)}.  \label{eq:tvd}
\end{align}
\noindent
We use the following basic properties of total variation distance. 
\begin{fact}\label{fact:tvd-small}
	Suppose $\mu$ and $\nu$ are two distributions for $\event$, then, 
	$
	{\mu}(\event) \leq {\nu}(\event) + \tvd{\mu}{\nu}.
$
\end{fact}

\begin{fact}\label{fact:tvd-sample}
	Suppose $\mu$ and $\nu$ are two distributions with same support $\Omega$; then, given a single sample from either $\mu$ or $\nu$, the best probability of successfully deciding whether $s$ came from $\mu$ or $\nu$ (achieved by the maximum likelihood estimator) is 
	\[
	\frac12 + \frac12\cdot\tvd{\mu}{\nu}.
	\]
\end{fact}

We also have the following (chain-rule) bound on the total variation distance of joint variables.

\begin{fact}\label{fact:tvd-chain-rule}
	For any distributions $\mu$ and $\nu$ on $n$-tuples $(X_1,\ldots,X_n)$, 
	\[
		\tvd{\mu}{\nu} \leq \sum_{i=1}^{n} \Exp_{X_{<i} \sim \mu} \tvd{\mu(X_i \mid X_{<i})}{\nu(X_i \mid X_{<i})}. 
	\]
\end{fact}

A simple consequence of this fact gives us the following ``over conditioning'' property as well. 

\begin{fact}\label{fact:tvd-over-conditioning}
	For any random variables $\rX,\rY,\rZ$, 
	\[
		\tvd{\rX}{\rY} \leq \tvd{\rX\rZ}{\rY\rZ} = \Exp_{Z} \tvd{(\rX \mid \rZ=Z)}{(\rY \mid \rZ=Z)}.  
	\]
\end{fact}
\begin{proof}
	By the non-negativity of TVD and~\Cref{fact:tvd-chain-rule}, 
	\[
		 \tvd{\rX}{\rY} \leq \tvd{\rX}{\rY} + \Exp_{X} \tvd{(\rZ \mid \rX=X)}{(\rZ \mid \rY=X)} = \tvd{\rX\rZ}{\rY\rZ} 
	\]
	Applying~\Cref{fact:tvd-chain-rule} again gives us
	\[
		\tvd{\rX\rZ}{\rY\rZ}  = \tvd{\rZ}{\rZ} + \Exp_{Z} \tvd{(\rX \mid \rZ=Z)}{(\rY \mid \rZ=Z)} = \Exp_{Z} \tvd{(\rX \mid \rZ=Z)}{(\rY \mid \rZ=Z)},
	\]
	which concludes the proof. 
\end{proof}

Similarly, we have the following data processing inequality for total variation distance as a consequence of the above. 

\begin{fact}\label{fact:tvd-data-processing}
	Suppose $\rX$ and $\rY$ are two random variables with the same support $\Omega$ and $f: \Omega \rightarrow \Omega$ is a fixed function.  Then,
	\[
		\tvd{f(\rX)}{f(\rY)} \leq \tvd{\rX}{\rY}. 
	\]
\end{fact}

\subsubsection*{Connections Between KL-Divergence and Total Variation Distance} 

The following Pinsker's inequality bounds the total variation distance between two distributions based on their KL-divergence, 

\begin{fact}[Pinsker's inequality]\label{fact:pinskers}
	For any distributions $\mu$ and $\nu$, 
	$
	\tvd{\mu}{\nu} \leq \sqrt{\frac{1}{2} \cdot \kl{\mu}{\nu}}.
	$ 
\end{fact}

We shall also use the following strengthening of Pinsker's inequality due to~\cite{ChakrabartyK18} that allows to lower bound KL-divergence between two distributions by a combination of $\ell_1$- and $\ell_2$-distance of the two distributions (instead of purely $\ell_1$-distance in the original Pinsker's inequality). 

\begin{proposition}[Strengthened Pinsker's Inequality{~\cite[KL vs $\ell_1/\ell_2$-inequality]{ChakrabartyK18}}]\label{prop:pinsker++}
	Given any pair of distributions $\mu$ and $\nu$ over the same finite domain $\Omega$, define 
	\[
	A := \set{x \in \Omega \mid \mu(x) > 2 \cdot \nu(x)} \quad \text{and} \qquad B := \Omega \setminus A. 
	\]
	Then, 
	\[
	\kl{\mu}{\nu} \geq (1-\ln{2}) \cdot \paren{\,\sum_{x \in A} \card{\mu(x)-\nu(x)} + \sum_{x \in B} \frac{(\mu(x)-\nu(x))^2}{\mu(x)}}. 
	\]
\end{proposition}


\section{Background on Fourier Analysis on Permutations}\label{app:fourier-permutation}

We review basics of representation theory and Fourier transform on permutations. We refer the interested reader to~\cite{HuangGGFourier09} for more details and background.

\paragraph{Representation theory.} We can define representations for any (symmetric) group, including the group of permutations on a fixed domain. 

\begin{Definition}\label{def:representation}
A \textnormal{\textbf{representation}} of a group $G$  is a map $\rho$ from $G$ to a set of invertible $d_{\rho}\times d_{\rho}$ (complex)
matrix operators ($\rho : G \rightarrow \mathbb{C}^{d_{\rho} \times d_{\rho}}$) which preserves algebraic structure in the sense that for all
$\sigma_1, \sigma_2 \in  G, \rho(\sigma_1\conc \sigma_2) = \rho(\sigma_1)\cdot  \rho(\sigma_2)$. The matrices which lie in the image of $\rho$ are called the
\textnormal{\textbf{representation matrices}}, and we will refer to $d_{\rho}$ as the \textnormal{\textbf{degree}} of the representation.
\end{Definition}

Two representations $\rho_1, \rho_2$ are said to be \textbf{equivalent} if there exists an invertible matrix $C $ such that for all $\sigma \in G$, 
\[
	\rho_2(\sigma) = C^{-1} \cdot \rho_1(\sigma) \cdot C.
\]
Given two representations $\rho_1, \rho_2$, we write the direct sum of $\rho_1$ and $\rho_2$, denoted by $\rho_1 \oplus \rho_2$ as,
\[
\rho_1 \oplus \rho_2(\sigma) = 
		\begin{bmatrix}
			\rho_1(\sigma)& 0 \\
			0 & \rho_2(\sigma)
		\end{bmatrix}.
\]
The degree of $\rho_1 \oplus \rho_2$ is $d_{\rho_1} + d_{\rho_2}$.

\begin{Definition}\label{def:irred-representation}
	A representation $\rho$ is said to be \textnormal{\textbf{reducible}} if it can be decomposed as $\rho = \rho_1 \oplus \rho_2$.
The set of \textnormal{\textbf{irreducible representations}} of any group $G$ is the collection of representations (up to equivalence) which are not reducible.
\end{Definition}
An example of an irreducible representation is the trivial representation, denoted by $\rho_0$, that takes every group element to $1 \in \mathbb{R}$. 

\paragraph{Fourier transform over permutations.} Fix any integer $b \geq 1$ and let $S_b$ be the set of permutations on $[b]$ in the following. 
Let $\reps := \reps(b)$ denote the (finite) set of all irreducible representations of the symmetric group $S_b$. 
 
For any function $f: S_b \rightarrow\mathbb{R}$ and representation $\rho \in \reps$, the Fourier transform of $f$ at $\rho$ is a $d_{\rho} \times d_{\rho}$ dimensional matrix defined as,
\begin{align}
	\hrho{f} := \sum_{\sigma \in S_b} f(\sigma) \cdot \rho(\sigma). \label{eq:fourier-transform}
\end{align}
The set of Fourier transforms at all representations in $\reps$ form the Fourier transform of $f$.

The Fourier Inversion theorem then implies that for any function $f: S_b \rightarrow \IR$,  
\begin{align}
	f(\sigma) = \frac{1}{b!} \cdot \sum_{\rho \in \reps} d_{\rho} \cdot \textnormal{\textsc{Trace}}(\hrho{f}^{\top} \cdot \rho(\sigma)). \label{eq:fourier-inverse}
\end{align}

\subsection*{Fourier Transforms of Distributions Over Permutations}

In the following, let $\cU := \cU_{S_b}$ denote the uniform distribution over $S_b$ and $\nu$ be any arbitrary distribution on $S_b$, where for $\sigma \in S_b$, $\nu(\sigma)$ denote the probability of $\sigma$ under $\nu$. 
Moreover, $\rho_0$ is the trivial representation in $\reps$ that maps all permutations to $1 \times 1$ dimensional matrix with entry $1$. 
The proofs of all following standard results can be found in~\cite{HuangGGFourier09}. 



\begin{fact}[{c.f.~\cite[Section 4.2]{HuangGGFourier09}}]\label{fact:Four-trans}
	For the Fourier transform over permutations,
	\begin{enumerate}
		\item\label{prop:Four-trivial-rep} For any probability distribution $\nu$ over $S_b$, $\hrhot{\nu} = 1$.
		\item\label{prop:Four-uniform-dist} For the uniform distribution $\cU$, $\hrhot{\cU} = 1$, and for any $\rho \in \reps \setminus \set{\rho_0}$, $\hrho{\cU} = \mathbf{0}$, where $\mathbf{0}$ is the $d_{\rho} \times d_{\rho}$ dimensional
		matrix of all zeros. 
	\end{enumerate}
\end{fact}

We can also use the Fourier convolution theorem to relate Fourier coefficient of distribution on concatenated permutations to each other. Given two distributions $\nu_1$ and $\nu_2$ over $S_b$, define $\nu := \nu_1 \circ \nu_2$ as the distribution obtained by sampling $\sigma_1 \sim \nu_1$ and $\sigma_2 \sim \nu_2$ independently and returning $\sigma_1 \circ \sigma_2$ as the sample
	of $\nu$. 

\begin{fact}[{c.f.~\cite[Definition 7 and Proposition 8]{HuangGGFourier09}}]\label{fact:Four-convo}
	For any distributions $\nu_1$ and $\nu_2$ over $S_b$ and $\nu = \nu_1 \circ \nu_2$, and any $\rho \in \reps$, 
	\[
		\hrho{\nu} = \hrho{\nu_1} \cdot \hrho{\nu_2}. 
	\] 
\end{fact}

Finally, we also have the following Plancherel's identity for this Fourier transform. 

\begin{proposition}[\!\!{\cite[Proposition 13]{HuangGGFourier09}}]\label{prop:Four-plancherel}
	For any distributions $\nu_1$ and $\nu_2$ over $S_b$, 
	\[
		\sum_{\sigma \in S_b} \paren{\nu_1(\sigma) - \nu_2(\sigma)}^2 = \frac1{b!} \cdot \sum_{\rho \in \reps} d_{\rho} \cdot \sum_{i, j \in [d_{\rho}]} \paren{\hrho{\nu_1}-\hrho{\nu_2}}_{i,j}^2.
	\]
\end{proposition}



\clearpage

\clearpage

\section{Tightness of~\Cref{lem:conc-amplify}}\label{app:tight-remark}

Our proof in~\Cref{lem:conc-amplify} can be restated as follows. Let $\nu_1,\ldots,\nu_g$ be $g$ distributions over $S_b$ and define $\nu = \nu_1 \circ \cdots \circ \nu_g$ (as defined in~\Cref{app:fourier-permutation}).
Suppose 
\[
	\norm{\nu_i - \cU_{S_b}}_2^2 = \frac{1}{b!} \cdot \eps_i,
\]
for all $i \in [g]$. Then, 
\[
	\norm{\nu - \cU_{S_b}}_2^2 \leq \frac{1}{b!} \cdot \prod_{i=1}^{g} \eps_i. 
\]
(the proof is verbatim as in~\Cref{lem:conc-amplify}). We now argue this bound is actually sharp.

Let $\nu_i$ for $i \in [g]$ be the following distribution. 
\begin{align*}
	\nu_i(\sigma) = \begin{cases}
		&\dfrac{1+\sqrt{\eps_i}}{b!} \qquad \textnormal{if $\sigma$ has an even number of inversions,} \\
		&\dfrac{1-\sqrt{\eps_i}}{b!} \qquad\textnormal{otherwise.}
	\end{cases}
\end{align*}
The $\ell_2$ distance of $\nu_i$ from $\cU_{S_b}$ is $\eps_i/b!$ by construction. Now, consider the distribution $\nu_1 \circ \nu_2$. 

\begin{itemize}
\item To get an even permutation, either two even permutations $\sigma_1 \sim \nu_1$ and $\sigma _2 \sim \nu_2$ can be picked, or two odd permutations can be picked. Thus for an even $\sigma \in S_b$, 
\begin{align*}
	\nu_1 \circ \nu_2(\sigma) &= \frac{b!}2 \cdot \paren{\paren{\frac1{b!}}^2 \cdot (1+\sqrt{\eps_1})(1+\sqrt{\eps_2})}  +  \frac{b!}2 \cdot \paren{\paren{\frac1{b!}}^2 \cdot (1-\sqrt{\eps_1})(1-\sqrt{\eps_2})}  \\
	&= \frac{1+\sqrt{\eps_1\eps_2}}{b!}.  
\end{align*}

\item To get an odd permutation,  we sample either an odd permutation from $\nu_1$ and even from $\nu_2 $ or vice-versa. Hence for any odd $\sigma \in S_b$, 
\begin{align*}
	\nu_1 \circ \nu_2(\sigma) &= \frac{b!}2 \cdot \paren{\paren{\frac1{b!}}^2 \cdot (1+\sqrt{\eps_1})(1-\sqrt{\eps_2})}  +  \frac{b!}2 \cdot \paren{\paren{\frac1{b!}}^2 \cdot (1-\sqrt{\eps_1})(1+\sqrt{\eps_2})}  \\
	&= \frac{1-\sqrt{\eps_1\eps_2}}{b!}. 
\end{align*}
\end{itemize}
If we proceed by induction we see that distribution $\nu = \nu_1 \circ \cdots \nu_g$ is as follows:
\begin{align*}
	\nu(\sigma) = \begin{cases}
		&\frac1{b!} \cdot \paren{1+ \paren{\prod_{i=1}^{g}\eps_i}^{1/2}} \qquad \textnormal{if $\sigma$ has even number of inversions and} \\
		&\frac1{b!} \cdot \paren{1- \paren{\prod_{i=1}^{g}\eps_i}^{1/2}} \qquad\textnormal{otherwise.}
	\end{cases}
\end{align*}
One can then calculate that the $\ell_2$-distance of $\nu$ from $\cU_{S_b}$ is 
\[
\frac1{b!} \cdot \prod_{i=1}^{\ell} \eps_i,
\] 
exactly matching the bounds in~\Cref{lem:conc-amplify}.

\clearpage
\section{Sorting Networks with Large Comparators}\label{app:sorting-stuff}

\newcommand{\merge}{\ensuremath{\textnormal{\textsc{Merge}}}}
\newcommand{\sort}{\ensuremath{\textnormal{\textsc{Sort}}}}
\newcommand{\kstar}{\ensuremath{k^*}}
\newcommand{\num}[1]{\ensuremath{\textnormal{\textsc{num(\texttt{#1})}}}}

\newcommand{\ones}[1]{\bm{1}(#1)}
\newcommand{\zeros}[1]{\bm{0}(#1)}

We provide a weaker version of~\Cref{prop:AKSsorting} that suffices for the proofs in our paper. The proof of the following result is adapted
from~\cite{ParkerP89} and we provide it here only for completeness. 

\begin{proposition}[\!\!\cite{ParkerP89}]\label{prop:simple-sorting}
		There exists an absolute constant $\caks > 0$ such that the following is true. For every pair of integers $m, b \geq 1$, there exists 
	\[
	\daks = \daks(r,b) = \caks \cdot \log^2_b(m)
	\]
	fixed equipartitions of $[m]$ into $\cP_1,\ldots,\cP_{\daks}$, each one consisting of $m/b$ sets of size $b$, with the following property. 
	Given any permutation $\sigma \in S_m$, there are $\daks$ permutations $\gamma_1,\ldots,\gamma_{\daks}$ where for every $i \in [\daks]$, $\gamma_i$ is simple on partition $\cP_i$ so that we have
	$
	\sigma = \gamma_1 \circ \cdots \circ \gamma_{\daks}. 
	$
\end{proposition}

This is equivalent to proving that there is a sorting network for $m$ elements with $b$-comparators with depth $O(\log_b^2 m)$. Each partition $\cP$ of $[m]$ with sets $P_1, P_2, \ldots, P_{m/b}$ of size $b$ each can be interpreted as $m/b$ wires as follows: wire $i$ compares elements $x_1, x_2, \ldots x_b \in P_i$ for each $i \in [m/b]$. The permutation which is simple on $\cP$ \emph{only} permutes elements from $ P_i$ with each other based on how wire $i$ arranges them. 

\subsection{Merge Subroutine with Large Comparators} 
The proof of \Cref{prop:simple-sorting} is based on the merge sort algorithm. We assume $m $ is a power of $b$ for simplicity. This can be removed by padding the input with dummy elements. We start by providing a 
$\merge$ subroutine first that is used in the sorting network.  For simplicity of exposition, \textbf{we assume we are working with $b^2$-sorters for $\merge$ instead of $b$-sorters}; we will address this easily by re-parameterizing 
when proving~\Cref{prop:simple-sorting}.

\begin{ourbox}
	\textbf{Input:} Array $A$ which is a permutation of $[m]$ such that $A[i \cdot m/b + 1]$ to $A[(i+1) \cdot m/b]$ is in increasing order for $0 \leq i \leq (b-1)$. 
	
	\smallskip
	
	\noindent
	\textbf{Output:} Array $A$ sorted in increasing order.  
	
	\smallskip
	\noindent
	\textbf{Algorithm $\merge(A,m,b)$:} (with $b^2$-sorters instead of $b$-sorters)
	\begin{enumerate}[label=$(\roman*)$]
		\item If $m \leq b^2$, compare all indices directly and sort them using a single sorter.  
		\item Otherwise, write array $A$ into a  $(m/b) \times b$ matrix $B$  such that $B[i][j] = A[(i-1)b + j]$ for $1 \leq i \leq m/b$ and $1 \leq j \leq b$. 
		\item Run $\merge(B[*][j],m/b,b)$ on each column $j \in [b]$ of $B$ separately. 
		\item Create matrix $C $ with dimensions $(m/b  + b) \times b$ where $C[i+j-1][ j] = B[i][ j]$, the empty elements in $C$ are filled with $-\infty$ for the top right corner, and $+\infty$ for the bottom left. 
		
		For any $\ell \in [m/b^2 + 1]$, we define the \textbf{square} $\ell$ as all $C[i][j]$ with $ (\ell-1) b + 1 \leq i \leq \ell \cdot b$ and $1 \leq j \leq b$ (see \Cref{fig:square}.)
		\item Divide $C$ into $m/b^2$ squares of size $b \times b$ each, and sort them each using a sorter. 
		\item For each $\ell \in [m/b^2-1]$, merge the last $b^2/2$ cells of square $\ell$ with the first $b^2/2$ cells of square $\ell+1$ using a single sorter for each one.  
	\end{enumerate}
\end{ourbox}

\begin{figure}[h!]
	\centering
	\tikzset{every picture/.style={line width=0.75pt}} 

\begin{tikzpicture}[x=0.75pt,y=0.75pt,yscale=-1,xscale=1]
\begin{scope}[scale=0.85]	
	\draw  [line width=1.5]  (159,68.82) -- (267.55,68.82) -- (267.55,321.09) -- (159,321.09) -- cycle ;
	\draw [fill={rgb, 255:red, 188; green, 94; blue, 104 }  ,fill opacity=1 ]   (158.55,171.73) -- (267.55,94.73) ;
	\draw    (160.55,271.73) -- (267.55,192.73) ;
	\draw  [line width=1.5]  (425,68.82) -- (533.55,68.82) -- (533.55,393.73) -- (425,393.73) -- cycle ;
	\draw [fill={rgb, 255:red, 188; green, 94; blue, 104 }  ,fill opacity=1 ]   (425.55,162.73) -- (532.55,163.09) ;
	\draw [fill={rgb, 255:red, 188; green, 94; blue, 104 }  ,fill opacity=1 ]   (424.55,254.73) -- (532.55,255.09) ;
	\draw    (274.55,176.32) -- (291,176) ;
	\draw [shift={(272.55,176.36)}, rotate = 358.87] [color={rgb, 255:red, 0; green, 0; blue, 0 }  ][line width=0.75]    (10.93,-3.29) .. controls (6.95,-1.4) and (3.31,-0.3) .. (0,0) .. controls (3.31,0.3) and (6.95,1.4) .. (10.93,3.29)   ;
	\draw    (397.55,174.36) -- (414,174.04) ;
	\draw [shift={(416,174)}, rotate = 178.87] [color={rgb, 255:red, 0; green, 0; blue, 0 }  ][line width=0.75]    (10.93,-3.29) .. controls (6.95,-1.4) and (3.31,-0.3) .. (0,0) .. controls (3.31,0.3) and (6.95,1.4) .. (10.93,3.29)   ;
	
	\draw (205,331.82) node [anchor=north west][inner sep=0.75pt]  [font=\large] [align=left] {$\displaystyle B$};
	\draw (272,81.82) node [anchor=north west][inner sep=0.75pt]  [font=\large] [align=left] {$ $};
	\draw (269.55,97.73) node [anchor=north west][inner sep=0.75pt]  [font=\large] [align=left] {$\displaystyle ( \ell -2) b\ +\ 2$};
	\draw (57,166.82) node [anchor=north west][inner sep=0.75pt]  [font=\large] [align=left] {$\displaystyle ( \ell -1) b+1\ $};
	\draw (133,251.82) node [anchor=north west][inner sep=0.75pt]  [font=\large] [align=left] {$\displaystyle \ell b$};
	\draw (474,401.82) node [anchor=north west][inner sep=0.75pt]  [font=\large] [align=left] {$\displaystyle C$};
	\draw (296,163.82) node [anchor=north west][inner sep=0.75pt]  [font=\large] [align=left] {$\displaystyle ( \ell -1) b+1\ $};
	\draw (537,227.82) node [anchor=north west][inner sep=0.75pt]  [font=\large] [align=left] {$\displaystyle \ell b$};
\end{scope}	
	
\end{tikzpicture}\caption{An illustration of a square $\ell \in [m/b^2 + 1]$ in the matrix $C$.}\label{fig:square}
\end{figure}
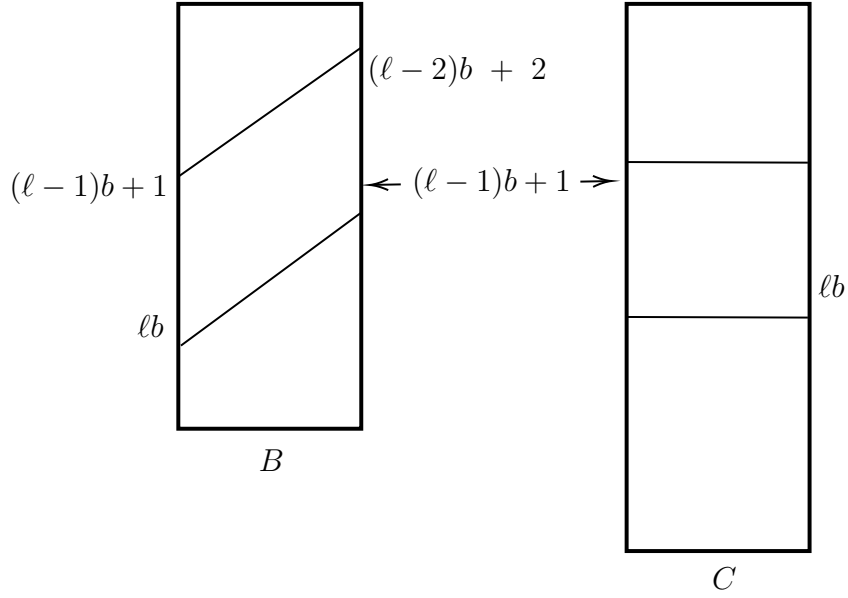

We will prove that the depth of $\merge$ is $O(  \log_b m)$, and it works correctly. The proof of correctness is based on the standard 0-1 principle for sorting networks. 

\begin{proposition}[c.f.{\cite[Section 5.3]{knuth1997art}}]\label{prop:0-1}
	Any sorting network which sorts correctly when the input is from the set $\set{0,1}^m$ can sort any arbitrary permutation of $[m]$ correctly also. 
\end{proposition}

Thus, from now on we assume the input $A$ consists of only 0's and 1's. 
To continue, we need a quick definition. We say a region (a region can be a row, a column or a square) \textbf{good} if it is comprised entirely of 1's or entirely of 0's. We call it a \textbf{bad} region otherwise. 
For any region $x$, we use $\ones{x}$ and $\zeros{x}$ to denote the number of 1's and 0's in $x$.

\begin{claim}\label{clm:step-3}
	At the end of step $(iii)$, each row of $B$ is still sorted. Moreover, there are at most $b$ bad rows in $B$ and they 
	form a contiguous region such that for each row $i < j$ in this region, we have $\ones{i} \leq \ones{j}$. 
\end{claim}

\begin{proof}
	Given that we copy each sorted region of $A$ into $m/b^2$ rows of $B$, we have that rows of $B$ are sorted after step $(ii)$. Thus, 
	for columns $p , q \in [b]$ with $p < q$,  $\ones{p} \leq \ones{q}$ in matrix $B$ at the end of step $(ii)$. 
	Moreover, step $(iii)$ does not change the number of 0's and 1's in each column. 
	After step $(iii)$, if some row $i \in [m/b]$ of $B$ is not sorted, there will be columns $p, q $ such that $p < q $ but $\ones{p} > \ones{q}$, a contradiction. 
	 Thus, after step $(iii)$, each row of $B$ will still be sorted.  
	
	Let us partition the rows as follows: $P_i = \set{j \mid (i-1) \cdot m/b^2 + 1 \leq j \leq i \cdot m/b^2}$, for $i \in [b]$, namely, 
	contiguous regions of $m/b^2$ rows each. After step $(ii)$, each $P_i$ can only have one bad row as we copied a sorted region of $A$ into all entries in $P_i$. 
	Thus, there are at most $b$ bad rows in the matrix $B$ overall.  After step $(iii)$, we get that these $b$ bad rows should appear next to each other as the columns are now sorted. 
	The sortedness also implies that for each $i < j$ in these bad rows, $\ones{i} \leq \ones{j}$ as desired. 	Refer to \Cref{fig:sort-1} for an illustration.
\end{proof}

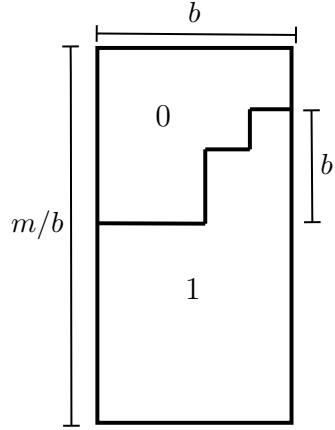
\begin{figure}[h!]
	\centering
	\tikzset{every picture/.style={line width=0.75pt}} 

\begin{tikzpicture}[x=0.75pt,y=0.75pt,yscale=-1,xscale=1]
\begin{scope}[scale=0.75]		
	\draw  [line width=1.5]  (275,46) -- (405.55,46) -- (405.55,298.27) -- (275,298.27) -- cycle ;
	\draw [line width=1.5]    (276,164) -- (347.55,164.27) ;
	\draw [line width=1.5]    (347.55,164.27) -- (347.55,114.27) ;
	\draw [line width=1.5]    (347.55,114.27) -- (377.55,114.27) ;
	\draw [line width=1.5]    (377.55,114.27) -- (377.55,87.27) ;
	\draw [line width=1.5]    (377.55,87.27) -- (406.55,87.27) ;
	\draw    (419,88) -- (419.55,163.27) ;
	\draw [shift={(419.55,163.27)}, rotate = 269.58] [color={rgb, 255:red, 0; green, 0; blue, 0 }  ][line width=0.75]    (0,5.59) -- (0,-5.59)   ;
	\draw [shift={(419,88)}, rotate = 269.58] [color={rgb, 255:red, 0; green, 0; blue, 0 }  ][line width=0.75]    (0,5.59) -- (0,-5.59)   ;
	\draw    (274.55,36.27) -- (408.55,37.27) ;
	\draw [shift={(408.55,37.27)}, rotate = 180.43] [color={rgb, 255:red, 0; green, 0; blue, 0 }  ][line width=0.75]    (0,5.59) -- (0,-5.59)   ;
	\draw [shift={(274.55,36.27)}, rotate = 180.43] [color={rgb, 255:red, 0; green, 0; blue, 0 }  ][line width=0.75]    (0,5.59) -- (0,-5.59)   ;
	\draw    (257,45) -- (257.13,62.28) -- (257.55,300.27) ;
	\draw [shift={(257.55,300.27)}, rotate = 269.9] [color={rgb, 255:red, 0; green, 0; blue, 0 }  ][line width=0.75]    (0,5.59) -- (0,-5.59)   ;
	\draw [shift={(257,45)}, rotate = 269.58] [color={rgb, 255:red, 0; green, 0; blue, 0 }  ][line width=0.75]    (0,5.59) -- (0,-5.59)   ;
	
	\draw (423,115) node [anchor=north west][inner sep=0.75pt]   [align=left] {$\displaystyle b$};
	\draw (334,13) node [anchor=north west][inner sep=0.75pt]   [align=left] {$\displaystyle b$};
	\draw (215,153) node [anchor=north west][inner sep=0.75pt]   [align=left] {$\displaystyle m/b$};
	\draw (332,199) node [anchor=north west][inner sep=0.75pt]  [font=\large] [align=left] {$\displaystyle 1$};
	\draw (312,83) node [anchor=north west][inner sep=0.75pt]  [font=\large] [align=left] {$\displaystyle 0$};
\end{scope}	
	
\end{tikzpicture}\caption{An illustration of matrix $B$ after step $(ii)$. All the cells below the partition are 1, and all the cells above it are 0.}\label{fig:sort-1}
\end{figure}

Now, with standard techniques, it is sufficient to sort every square of dimension $b \times b$ in $B$ to get the final sorted array. However, \cite{ParkerP89} employs a diagonalization technique to reduce the number of squares we need to sort. 
The next lemma is the main reason why we do not need to sort every square of dimension $b \times b$.

\begin{lemma}\label{lem:sort-main}
	At the end of step $(iv)$, for every $\ell \in [m/b^2-1]$: (1) if square $\ell$ has at least one $1$, then $\zeros{\ell+1} < b^2/2$, and (2) if 
	square $\ell+1$ has at least one $0$, then $\ones{\ell} < b^2/2$. 
\end{lemma}

\begin{proof}
Let square $\ell$ have at least one 1 in it. We will look at matrix $B$ after step $(iii)$ to upper bound $\zeros{\ell+1}$. Each square $\ell$ in matrix $C$ will be a rhombus in matrix $B$. It will contain all the cells $B[i,j]$ such that: 
\[
  (\ell-1) b + 1 \leq i + j \leq \ell \cdot b \quad \text{and} \quad 1 \leq j \leq b.
\]
\Cref{fig:areas} gives an illustration of this (and also specifies several key regions needed for the proof). 
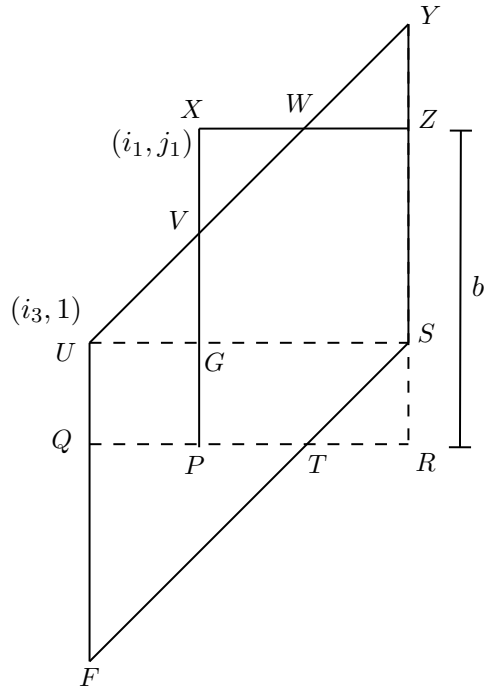
\begin{figure}[h!]
	\centering
	\tikzset{every picture/.style={line width=0.75pt}} 

\begin{tikzpicture}[x=0.75pt,y=0.75pt,yscale=-1,xscale=1]
	
	\draw [fill={rgb, 255:red, 188; green, 94; blue, 104 }  ,fill opacity=1 ]   (246.73,218.1) -- (407.55,57.27) ;
	\draw [fill={rgb, 255:red, 188; green, 94; blue, 104 }  ,fill opacity=1 ]   (246.73,378.92) -- (407.55,218.1) ;
	\draw    (302,110) -- (407.55,109.91) ;
	\draw  [dash pattern={on 4.5pt off 4.5pt}]  (407.55,57.27) -- (407.55,269.27) ;
	\draw    (407.55,57.27) -- (407.55,218.1) ;
	\draw    (246.73,218.1) -- (246.73,378.92) ;
	\draw    (302,110) -- (302,270.82) ;
	\draw  [dash pattern={on 4.5pt off 4.5pt}]  (246.73,269.27) -- (407.55,269.27) ;
	\draw  [dash pattern={on 4.5pt off 4.5pt}]  (246.73,218.1) -- (407.55,218.1) ;
	\draw    (434,111) -- (433.55,270.82) ;
	\draw [shift={(433.55,270.82)}, rotate = 270.16] [color={rgb, 255:red, 0; green, 0; blue, 0 }  ][line width=0.75]    (0,5.59) -- (0,-5.59)   ;
	\draw [shift={(434,111)}, rotate = 270.16] [color={rgb, 255:red, 0; green, 0; blue, 0 }  ][line width=0.75]    (0,5.59) -- (0,-5.59)   ;
	
	\draw (379,376.82) node [anchor=north west][inner sep=0.75pt]  [font=\large] [align=left] {$ $};
	\draw (291,93.82) node [anchor=north west][inner sep=0.75pt]  [font=\small] [align=left] {$\displaystyle X\ $};
	\draw (412,46.82) node [anchor=north west][inner sep=0.75pt]  [font=\small] [align=left] {$\displaystyle Y\ $};
	\draw (411,98.82) node [anchor=north west][inner sep=0.75pt]  [font=\small] [align=left] {$\displaystyle Z\ $};
	\draw (344,90.82) node [anchor=north west][inner sep=0.75pt]  [font=\small] [align=left] {$\displaystyle W\ $};
	\draw (228,216.82) node [anchor=north west][inner sep=0.75pt]  [font=\small] [align=left] {$\displaystyle U\ $};
	\draw (285,149.82) node [anchor=north west][inner sep=0.75pt]  [font=\small] [align=left] {$\displaystyle V\ $};
	\draw (226,260.82) node [anchor=north west][inner sep=0.75pt]  [font=\small] [align=left] {$\displaystyle Q$};
	\draw (293,273.82) node [anchor=north west][inner sep=0.75pt]  [font=\small] [align=left] {$\displaystyle P$};
	\draw (409.55,272.27) node [anchor=north west][inner sep=0.75pt]  [font=\small] [align=left] {$\displaystyle R$};
	\draw (355.55,273.27) node [anchor=north west][inner sep=0.75pt]  [font=\small] [align=left] {$\displaystyle T$};
	\draw (410.55,207.27) node [anchor=north west][inner sep=0.75pt]  [font=\small] [align=left] {$\displaystyle S$};
	\draw (240,380.82) node [anchor=north west][inner sep=0.75pt]  [font=\small] [align=left] {$\displaystyle F$};
	\draw (303,221.82) node [anchor=north west][inner sep=0.75pt]  [font=\small] [align=left] {$\displaystyle G$};
	\draw (438,183) node [anchor=north west][inner sep=0.75pt]   [align=left] {$\displaystyle b$};
	\draw (256,109) node [anchor=north west][inner sep=0.75pt]   [align=left] {$\displaystyle ( i_{1} ,j_{1})$};
	\draw (205,193) node [anchor=north west][inner sep=0.75pt]   [align=left] {$\displaystyle ( i_{3} ,1)$};

\end{tikzpicture}\caption{Upper bounding the number of 0's in square $\ell+1$ in \Cref{lem:sort-main}.}\label{fig:areas}
\end{figure}

Let $(i_1,j_1)$ be the top most cell in square $\ell$ with a 1. In case of ties, we pick the left most cell. By \Cref{clm:step-3}, we know that all the cells $(i, j)$ with $i \geq i_1$ and $j \geq j_1$ must be a 1 in $B$.
Let $i_2 = i_1 + b$. Cell $(i_2, j_1)$ must be present in square $\ell+1$ if cell $(i_1, j_1)$ is in square $\ell$. 

Let $(i_3, 1)$ be the point which is the left most and bottom most cell in square $\ell $ (by our construction, $i_3 = \ceil{\frac{i_1}{b}} \cdot b$). Then, $(i_3-j +1, j)$ is the last cell in column $j$ which is  in square $\ell$.

Refer to \Cref{fig:areas} for the various possible areas in matrix $B$. We use $\num{x}$ to denote the number of cells in region $x$. The following regions are defined in the figure, and we make some simple observations:

\begin{itemize}
	\item Square $\ell+1$ is all the cells inside rhombus $\texttt{UYSF}$.
	\item The top most and left most 1 in square $\ell$ is the first cell in region $\texttt{XVW}$. 
	\item $\num{XZRP} = b(b-j_1+1)$.
	\item $\num{XVW}= \nicefrac12 \cdot (i_3-j_1-i_1) \cdot (i_3-j_1-i_1+1)$. (The boundary of this region is marked by the cells $(i_1, j_1), (i_1, i_3-i_1+1), (i_3-j_1+1, j_1)$.)
	\item  $\num{STR} = \nicefrac12 \cdot (i_2-i_3-1) \cdot (i_2-i_3)$.
	\item $\num{QTF} = \nicefrac12 \cdot (i_3+b-i_2) \cdot (i_3+b-i_2+1)$.  
\end{itemize}

The total number of cells with 1 in square $\ell+1$ is lower bounded by,
\begin{align*}
	\ones{\ell+1} &\geq \num{VWZSTP} + \num{QTF} \\
	 &= \num{XZRP}-\num{XWV}-\num{SRT}+\num{QTF} \\
	&= b(b-j_1+1) - \\
	&\hspace{1cm} \nicefrac12 \cdot \Paren{(i_3-j_1-i_1)(i_3-j_1-i_1+1) +(i_2-i_3-1)(i_2-i_3)- (i_3+c-i_2)(i_3+c-i_2+1)} \\
	&= b(b-j_1+1) - \nicefrac12 \cdot \paren{(x-j_1)(x-j_1+1) + (b-x)(b-x-1) - x(x+1)} \tag{by defining $x := i_3-i_1$} \\ 
	&= b^2-bj_1 + b - \nicefrac12 \cdot \paren{(x-j_1)(x-j_1+1) + b^2+x^2-2bx-b+x-x^2-x} \\
	&=  b^2-bj_1 + b - \nicefrac12 \cdot \paren{(x-j_1)(x-j_1+1) + b^2-2bx-b} \\
	&= \nicefrac12 \cdot (b^2 + 3b +(x-j_1)(x-j_1+1+2b)).
\end{align*}

This quantity is minimized when $x = j_1$. Note that $i_1 + j_1 \leq i_3 $ for cell $(i_1, j_1)$ to belong to square $\ell$, and thus, $x \geq j_1$. This proves one part of the lemma, as the number of 0s is upper bounded by $\nicefrac12 \cdot (b^2 - 3b)$ in square $\ell+1$, whenever there is a cell with 1 in square $\ell$. 

We can prove that if there is a 0 in square $\ell+1$, the number of 1's in square $\ell$ is upper bounded by $(b^2-3b)/2 < b^2/2$ similarly. 
\end{proof}

\begin{lemma}\label{lem:main-correct}
	The network $\merge$ merges the input array $A$ correctly. Moreover, with $b^2$-size sorters, it has depth $O(\log_b m)$. 
\end{lemma}

\begin{proof}
	By~\Cref{clm:step-3}, there are only $b$ contiguous bad rows in $B$, which implies that there are most two consecutive bad squares $\ell$ and $\ell+1$ in $C$.  
	By~\Cref{lem:sort-main}, once we merge the last $b^2/2$ cells of square $\ell$ with first $b^2/2$ cells of square $\ell+1$, we will definitely move all the 0's before the 1's, thus 
	making $C$ entirely sorted. This proves the correctness of $\merge$. 
	
	If we have sorters of size $b^2$, we can perform steps $(v)$ and $(vi)$ easily with 1 layer each. Let $M(m)$ denote the depth of $\merge$ on input size $m$. We get the following recurrence:
	\[
		M(m) = M(m/b) + 2.
	\]
	The total depth is thus $M(m) = O(\log_b m)$. 
\end{proof}

\subsection{The Final Sorting Network}

We are now ready to give the final sorting network to sort an array of size $m$ using $\merge$ primitive.
For this proof, \textbf{we revert back to using $b$-sorters as was the original problem}.

\begin{ourbox}
	\textbf{Input:} Array $A$ which is a permutation of $[m]$. 
	
	\smallskip
	\noindent
	\textbf{Output:} Array $A$ sorted in increasing order.  
	
	\smallskip
	\noindent
	\textbf{Algorithm $\sort(A)$:}
	\begin{enumerate}[label=$(\roman*)$]
		\item If $m \leq b$, compare all indices directly using one $b$-sorter.
		\item  Otherwise, run $\sort$ on $(A[\ell \cdot m/\sqrt{b}+1], \ldots, A[(\ell+1)\cdot m/\sqrt{b}])$ for each $0 \leq \ell \leq \sqrt{b}-1$.
		\item Run $\merge(A,m/\sqrt{b},\sqrt{b})$ (with $(\sqrt{b})^2 = b$-sorters)  
	\end{enumerate}
\end{ourbox}

\begin{proof}[Proof of \Cref{prop:simple-sorting}]
	The correctness of primitive $\sort$ follows from \Cref{lem:main-correct}, as we have $b$-sorters, and we merge $\sqrt{b}$ arrays of size $m/\sqrt{b}$ each.
	The depth of $\merge$ is $O(\log_{\sqrt{b}} m) = O(\log_b m)$. Let $S(m)$ be the depth of the sorting network on input of size $m$. Then,
	\[
		S(m) = S(\frac{m}{\sqrt{b}}) + O(\log_b m),
	\]
	giving us a final depth of $S(m) = O(\log_b^2 m)$ as desired. 
\end{proof}

\clearpage

\end{document}